\documentclass[11pt]{article}

\RequirePackage{color}
\RequirePackage{titling} 
\RequirePackage[letterpaper, left=1.2truein, right=1.2truein, top = 1.2truein,bottom = 1.2truein]{geometry} 
\RequirePackage[affil-it]{authblk} 
\RequirePackage{natbib} 
\RequirePackage[colorlinks,citecolor=blue,urlcolor=blue,breaklinks=true]{hyperref} 
\RequirePackage{algorithm, algorithmic} 
\RequirePackage{appendix} 
\RequirePackage{amsthm,amsmath,amsfonts,amssymb,mathtools} 
\RequirePackage{booktabs} 
\RequirePackage{enumitem} 
\RequirePackage{graphicx} 
\RequirePackage{subfigure} 
\RequirePackage{tikz} 
\usetikzlibrary{positioning,graphs,quotes,shapes,shadows,arrows}
\tikzstyle{pattern}=[circle,draw,fill=white]
\tikzstyle{line}=[draw,-stealth,thick]
\tikzstyle{dotline}=[draw,-stealth,thick,dotted]

\theoremstyle{plain}
\newtheorem{theorem}{Theorem}[section]

\newtheorem{lemma}[theorem]{Lemma}

\theoremstyle{definition}

\newtheorem{assumption}[theorem]{Assumption}
\theoremstyle{remark}

\newcommand{\E}{\mathbb{E}}
\newcommand{\quads}{\quad\ }
\newcommand{\tr}{^\intercal}
\newcommand{\domr}{\mathrm{dom}_r}

\newcommand{\bbR}{\mathbb R}
\newcommand{\bbP}{\mathbb P}

\newcommand{\calL}{\mathcal L}
\newcommand{\calLr}{\mathcal L^r}
\newcommand{\calR}{\mathcal R}
\newcommand{\calT}{\mathcal T}
\newcommand{\calF}{\mathcal F}
\newcommand{\calM}{\mathcal M}
\newcommand{\calH}{\mathcal H}

\newcommand{\lob}{l^{[r]}}
\newcommand{\lms}{l^{[\overline{r}]}}

\newcommand{\Lob}{L^{[r]}}
\newcommand{\Lms}{L^{[\overline{r}]}}
\newcommand{\Lobi}{L^{[r]}_i}

\newcommand{\Phir}{\Phi^{r}}
\newcommand{\phir}{\phi^{r}}
\newcommand{\alphar}{\alpha^{[r]}}
\newcommand{\hatalphar}{\hat{\alpha}^{[r]}}

\newcommand{\bfD}{\mathbf D}

\newcommand{\PAr}{\mathsf{Pa}(r)}
\newcommand{\PAs}{\mathsf{Pa}(s)}
\newcommand{\PAsk}{\mathsf{Pa}(s_k)}

\newcommand{\calN}{\mathcal N}
\newcommand{\bbG}{\mathbb G}
\newcommand{\bfA}{\mathbf A}
\newcommand{\bfB}{\mathbf B}
\newcommand{\bfC}{\mathbf C}
\newcommand{\bfE}{\mathbf E}
\newcommand{\bfF}{\mathbf F}
\newcommand{\calJ}{\mathcal J}
\newcommand{\calG}{\mathcal G}
\newcommand{\calO}{\mathcal O}
\newcommand{\calU}{\mathcal U}
\newcommand{\PEN}{\mathtt{PEN}}

\title{Efficient Estimation under Multiple Missing Patterns via Balancing Weights}
\author[1]{Jianing Dong}
\author[2]{Raymond K. W. Wong}
\author[3]{Kwun Chuen Gary Chan}
\affil[1,2]{Department of Statistics, Texas A\&M University}
\affil[3]{Department of Biostatistics, University of Washington}
\date{}

\begin{document}

\maketitle

\begin{abstract}
As one of the most commonly seen data challenges, missing data, in particular, multiple, non-monotone missing patterns, complicates estimation and inference due to the fact that missingness mechanisms are often not missing at random, and conventional methods cannot be applied.  Pattern graphs have recently been proposed as a tool to systematically relate various observed patterns in the sample.  We extend its scope to the estimation of parameters defined by moment equations, including common regression models, via solving weighted estimating equations with weights constructed using a sequential balancing approach.  These novel weights are carefully crafted to address the instability issue of the straightforward approach based on local balancing.  We derive the efficiency bound for the model parameters and show that our proposed method, albeit relatively simple, is asymptotically efficient.  Simulation results demonstrate the superior performance of the proposed method, and real-data applications illustrate how the results are robust to the choice of identification assumptions.
\end{abstract}


\section{Introduction}
Incomplete data is a prevalent issue in data analysis, arising in a variety of fields, including clinical trials, social sciences, and machine learning. Proper handling of missing data is crucial, as inappropriate assumptions or methods can lead to biased inferences and invalid conclusions. The theoretical framework for handling missing data was first formalized by \citet{rubin1976inference}, who categorized missingness mechanisms into three broad classes: Missing Completely at Random (MCAR), Missing at Random (MAR), and Missing Not at Random (MNAR). While MCAR assumes that missingness is entirely unrelated to the data, and MAR posits that missingness depends only on observed data, MNAR describes scenarios where the probability of missingness depends on unobserved variables. These definitions provide a foundation for understanding the relationship between observed and missing data, but practical applications often require more nuanced structures to model real-world missingness mechanisms.

Although complete-case analysis excludes missing data from the dataset, offering a convenient approach, this straightforward technique works under the stringent MCAR assumption, which rarely holds in real-world data. The MAR assumption is often restrictive and may not align with real-world missing data scenarios, such as non-monotone missingness, where the missing data does not follow a structured sequence \citep{robins1997nonignorable, troxel1998analysis}. In general, MNAR is a more appropriate assumption than MAR. One real-world example is that high-income individuals are less likely to report their income due to privacy concerns, stigma, or social desirability bias. The missingness is directly related to the unobserved income values themselves. 

However, the MNAR assumption introduces significant challenges since missingness cannot be ignored when making inferences from incomplete data. Identifying assumption is required to specify which parameters in the full data model can be estimated despite the missing data. The pattern mixture model is a widely used method \citep{little1993pattern, tchetgen2018discrete} that classifies data based on missing patterns and imposes assumptions on the conditional densities of missing variables for specific patterns. The pattern graph introduced in \citet{chen2022pattern} provides a visualization of more complex identifying assumptions, which hierarchically link the conditional densities of missing variables across multiple patterns. It is worth noting that pattern graphs are different from graphical models \citep{mohan2021graphical, nabi2020full} and the casual graph \citep{bhattacharya2020causal, shpitser2016consistent}.

Imputation-based methods are commonly used with the pattern mixture model to estimate the parameters of interest, but they can be computationally intensive due to repeated iterations or samplings. Alternatively, identification conditions can often be stated in a selection model, leading to inverse probability weighting (IPW) estimation. However, they may lack stability due to extreme estimated weights. Balancing weights  \citep{zubizarreta2015stable, wong2018kernel, dong2024balancing} are attractive since they are designed to achieve a more balanced distribution of variables between groups, which can lead to a more stable and efficient estimation.   

In this paper, we extend the balancing approach to missing mechanisms that can be visualized by a pattern graph or its generalization.  A local estimation encourages the covariate balance between patterns directly connected by edges when hierarchical structures exist in the graph, while it may lead to extreme weights due to model extrapolation or fail to account for the errors accumulated through the multiplications used to construct the inverse propensity weights.  A sequential estimation procedure is proposed to address instability in the estimation.  We expand the scope of estimation under pattern graphs to model equations that include common regression models.  We study the semiparametric efficiency bound and show the consistency and asymptotic efficiency of the proposed estimator.

\section{Missing data assumptions}
\subsection{Preliminaries}
In this section, we formally describe the setup of the problem. Let $L=(L_{(1)},\ldots,L_{(d)})\in \prod_{j=1}^d\calL_{(j)}$, where $\calL_{(j)}\subseteq\bbR$,
be a vector of potentially observable random variables. To indicate the observation of these variables, let $R=(R_{(1)},\ldots,R_{(d)})\in\{0,1\}^d$ be a binary random vector such that $R_{(j)}=1$ when $L_{(j)}$ is observed. Let $\calR=\{r \in \{0,1\}^d:P(R=r)>0\}$ be the set of all possible missing patterns and $M=|\calR|$ be the number of missing patterns in the study. We define a partial ordering of missing pattern vectors: for two patterns $s,r\in\calR$ such that $s\neq r$, we say $s>r$ if $s_{(j)}\ge r_{(j)}$ for all $1\le j\le d$. Denote the complete-case pattern by $1_d=(1,\ldots,1)$. For each missing pattern $r$, we denote the observed variables by $\Lob$ and the missing variables by $\Lms$. For example, $L^{[101]}=(L_{(1)},L_{(3)})$ and  $L^{[\overline{101}]}=L_{(2)}$. So, the observations are $\{(L^{[R_i]}_i,R_i)\}_{i=1}^N$. 
Denote $\domr:=\prod_{\{j: R_{(j)}=1\}}\calL_{(j)}\subseteq\bbR^{d_r}$ where $d_r$ is the number of observed variables in pattern $r$, then $\Lob\in\domr$ and $\Lms\in\prod_{\{j: R_{(j)}=0\}}\calL_{(j)}\subseteq\bbR^{d-d_r}$. 

Let $\theta_0\in\bbR^q$ be the parameter of interest which is the unique solution to $\E\{\psi_\theta(L)\}=0$, with a known vector-valued estimating function $\psi_\theta(L)=\psi(L,\theta)$ that takes values in $\bbR^q$.
For instance, we could use the quasi-likelihood estimating functions for the generalized linear models. If full data were observed, a solution to the estimating equations $N^{-1}\sum_{i=1}^N\psi_\theta(L_i)=0$ is a common Z-estimator. However, $\psi_\theta(L_i)$ can only be evaluated at samples with complete observations of $L_i$. When missing data is present, practitioners often solve the complete-case estimating equation $N^{-1}\sum_{i=1}^N \mathsf{1}_{R_i=1_d}\psi_\theta(L_i)=0$, but it is typically biased unless $R$ and $L$ are independent, i.e., missing completely at random.  There are two directions to reconstruct the full data density and address the bias issues.
    \paragraph{Conditional density of missing variables.} The joint density of $L$ can be expressed as
        \begin{align*}
        p(l)=\sum_{r\in\calR}P(R=r)p(\lob\mid R=r)p(\lms\mid\lob,R=r)\ .
        \end{align*}
    Note that $p(\lms\mid\lob,R=r)$ cannot be identified without assumptions since $\lms$ is never observed when $R=r$.  Given assumptions such that estimators $\hat{p}(\lms\mid\lob,R=r)$ for each pattern $r$ are available, imputation can be performed repeatedly to generate multiple complete datasets, or the conditional density can be directly integrated into the analysis.    

    \paragraph{Selection probability.} The joint density of $L$ can be expressed as $p(l)=p(l,r)/P(R=r\mid l)$. Using the selection probability, the population-level expectation can be reconstructed by weighting the complete cases:
        \begin{align}
        \E\{\psi_\theta(L)\}
        =\E\left\{\frac{\mathsf{1}_{R=1_d}}{P(R=1_d\mid L)}\psi_\theta(L)\right\}\ .        
        \label{eqn:estimating}
        \end{align}
   Given an estimator $\hat{\pi}(l)$ for $\pi(l)=P(R=1_d\mid l)$, an estimator of $\theta$ can be obtained by solving the weighted estimating equation: $N^{-1}\sum_{i=1}^N \mathsf{1}_{R_i=1_d}\psi_\theta(L_i)/\hat{\pi}(L_i)=0$. The modeling and estimation of $\hat{\pi}(l)$ is not straightforward under missing not at random because $\pi(l)$ depends on components of $L$ that are not fully observed when $R\neq 1_d$.

\subsection{Regular pattern graphs}
Various identifying assumptions have been considered in the literature.  \citet{little1993pattern} proposed the complete-case missing variable (CCMV) assumption, that matches the unidentifiable conditional distribution of missing variables for missing patterns to the identifiable distribution for complete cases. That is, for any $r\in\calR \backslash \{1_d\}$ and all $\lob\in\domr$, $p(\lms\mid\lob,R=r)=p(\lms\mid\lob,R=1_d)$. For monotone missingness, \citet{Molenberghs1998} considered the available-case (AC) restriction, $p(\lms\mid\lob,R=r)=p(\lms\mid\lob,R>r)$. \citet{thijs2002strategies} introduced the neighboring-case (NC) restriction, $p(\lms\mid\lob,R=r)=p(\lms\mid\lob,R=s)$ where $s>r$ and $|s|=|r|+1$.

Recently, \citet{chen2022pattern} proposed regular pattern graphs to encode a set of identifying assumptions, which recursively identify the unknown conditional density. A pattern graph is a directed graph $G=(\calR,E)$, where each vertex represents a missing pattern, and the directed edges indicate connections in the distribution of $(L,R)$ across different patterns (to be described clearly later). We define the notion of parents and children in the graph. For two patterns $s,r\in\calR$, $s$ is called as a parent of $r$, and $r$ is called as a child of $s$ if there is a directed edge $s\to r$. Each pattern may have multiple parents and/or children.  A regular pattern graph $G=(\calR, E)$ is a pattern graph such that (i) $G$ is a directed acyclic graph (DAG), (ii) the complete-case pattern $1_d$ is the only node without a parent, and (iii) for any $s,r \in \calR$ with an edge $(s\rightarrow r)\in E$, then $s>r$. Figure \ref{fig1-regular-graph} depicts a few examples of regular pattern graphs.
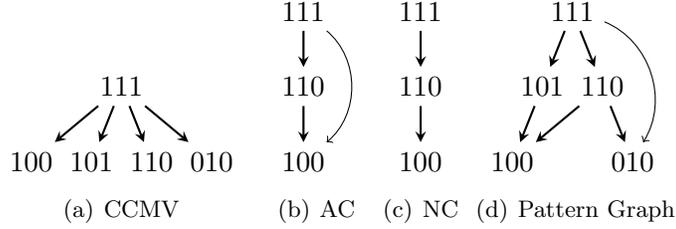
\begin{figure}[ht]
\centering
\subfigure[CCMV]{\label{fig1-CCMV}
\begin{tikzpicture}[node distance=1cm, scale=0.5]
    \node(111){111};
    \node[below of = 111, xshift = 0.4cm](110){110};
    \node[below of = 111, xshift = -0.4cm](101){101};
    \node[below of = 111, xshift = 1.2cm](010){010};
    \node[below of = 111, xshift = -1.2cm](100){100};
    \path[line](111)--(110);
    \path[line](111)--(101);
    \path[line](111)--(010);
    \path[line](111)--(100);
\end{tikzpicture}
}
\subfigure[AC]{
\begin{tikzpicture}[node distance=1cm, scale=0.5]
    \node(111){111};
    \node[below of = 111](110){110};
    \node[below of = 110](100){100};
    \path[line](111)--(110);
    \path[line](110)--(100);
    \draw[->] (111) to[bend left = 50] (100);  
\end{tikzpicture}
}
\subfigure[NC]{
\begin{tikzpicture}[node distance=1cm, scale=0.5]
    \node(111){111};
    \node[below of = 111](110){110};
    \node[below of = 110](100){100};
    \path[line](111)--(110);
    \path[line](110)--(100);
\end{tikzpicture}
}
\subfigure[Pattern Graph]{
\begin{tikzpicture}[node distance=1cm, scale=0.5]
    \node(111){111};
    \node[below of = 111, xshift = 0.4cm](110){110};
    \node[below of = 111, xshift = -0.4cm](101){101};
    \node[below of = 110, xshift = 0.4cm](010){010};
    \node[below of = 101, xshift = -0.4cm](100){100};
    \path[line](111)--(110);
    \path[line](111)--(101);
    \path[line](110)--(010);
    \path[line](110)--(100);
    \path[line](101)--(100);
    \draw[->] (111) to[bend left = 50] (010);
\end{tikzpicture}
}
\caption{Examples of regular pattern graphs.}
\label{fig1-regular-graph}
\end{figure}

Since a parent pattern is more informative than its
child pattern, \citet{chen2022pattern} models the unobserved part of pattern $r$ using the information from its parents. Let $\PAr$ be the set of parents of a pattern $r\in\calR$. Specifically, for any pattern $r\in G$ and $r\neq 1_d$, the identification assumption being encoded in $G$ is
\begin{equation}
p(\lms\mid\lob,R=r)=p(\lms\mid\lob,R\in\PAr)\ .
\label{eqn:chen-density}
\end{equation}

It can be shown that $p(l)$ is nonparametrically identified if the above assumption holds for every missing pattern. Besides, the above assumption can be equivalently stated as a selection odds model: 
\begin{equation}
\frac{P(R=r\mid l)}{P(R=\PAr\mid l)}
=\frac{P(R=r\mid\lob)}{P(R=\PAr\mid\lob)}\ .
\label{eqn:chen-propensity}
\end{equation}
To connect the odds model with selection probability, we define the walk and path in the graph. A walk on the graph $G=(\calR,E)$ is defined as a sequence of directed edges $r_0\to r_1\to \ldots\to r_m$ such that $(r_{j-1}\to r_j)\in E$ for $j=1,\ldots,m$. A path is a walk in which all vertices (and therefore also all edges) are distinct. Since a regular pattern graph is a DAG, we can also represent a path from $r_0$ to $r_m$ by its sequence of vertices along the path:
\begin{align*}
\Xi_{r_0,r_m}=\{r_0, r_1, \ldots, r_m\}
\end{align*}

Let $\Pi_{s,r}$ denote the collection of all paths from $s$ to $r$. Write $\Pi_r:=\Pi_{1_d,r}$ which is the collection of all paths from $1_d$ to $r$. And let $\Pi:=\cup_{r\in\calR}\Pi_r$ denote the collection of all paths from the source $1_d$ in $G$. Also let $O^r(\lob)=P(R=r\mid\lob)/P(R\in\PAr\mid\lob)$ and $Q^r(l)=P(R=r\mid l)/P(R=1_d\mid l)$ for any pattern $r$. The propensity $\pi(l)=P(R=1_d\mid l)$ is identifiable and has the following recursive form:
\begin{align}\label{eqn:chen-ipw}
\pi(l)&=\frac{1}{\sum_{r\in\calR}Q^r(l)};\\
Q^r(l)&=O^r(\lob)\times \sum_{s\in\PAr}Q^s(l)
=\sum_{\Xi\in\Pi_r}\prod_{s\in\Xi}O^s(l^{[s]})\nonumber\ .
\end{align}

\subsection{Generalization of missing data assumptions encoded in pattern graph}
Note that, the right hand side of \eqref{eqn:chen-density} can be rewritten as
\begin{equation}
\sum_{s\in\PAr}\frac{P(R=s\mid\lob)}{P(R\in\PAr\mid\lob)}p(\lms\mid\lob,R=s).
\label{eqn:chen_mix}
\end{equation}

In other words, the missing variable density is assumed to be a mixture density of that for parent patterns, where the mixture coefficients are $P(R=s\mid\lob)/P(R\in\PAr\mid\lob)$. It is possible to generalize the choice of the mixture coefficients, extending the work of \citet{chen2022pattern}. We propose the following generalization of mixture density:
\begin{align}\label{identifying-conditional-density}
P(l^{\bar{r}}\mid \lob,R=r)
=\sum_{s\in\PAr}C^{s,r}(\lob)P(l^{\bar{r}}\mid \lob,R=s),
\end{align}
where $C^{s,r}(\lob)$ is an identifiable function of observed variables $\lob$ under the constraint that $C^{s,r}(\lob)\ge0$ and $\sum_{s\in\PAr}C^{s,r}(\lob)=1$. Let $O^{s,r}(\lob)=P(R=r\mid\lob)/P(R=s\mid\lob)$. Therefore, we have
\begin{align}\label{identifying-Q}
Q^r(l)=\sum_{s\in\PAr}C^{s,r}(\lob)O^{s,r}(\lob)Q^s(l).
\end{align}
Gathering the assumptions for every missing pattern $r$, we claim the following theorems for identifiability.
\begin{theorem}
    Assume that the conditional density of missing variables is modeled as in \eqref{identifying-conditional-density} for every missing pattern $r$, then $p(l,r)$ is nonparametrically identifiable/saturated.
\end{theorem}

\begin{theorem}
    Assume that the propensity odds $Q^r(l)$ is modeled as in \eqref{identifying-Q} for every missing pattern $r$, then $\pi(l)$ is nonparametrically identifiable/saturated.
\end{theorem}

One may recognize that there are infinitely many possible choices for mixture coefficients. In this paper, we focus on the following three types:
\begin{align*}
\textit{Type 1}:\frac{P(R=s\mid \lob)}{P(R\in\PAr\mid \lob)};\quad
\textit{Type 2}:\frac{P(R=s)}{P(R\in\PAr)};\quad
\textit{Type 3}:\textit{known constant}\ .
\end{align*}

The first type corresponds to the assumptions in \citet{chen2022pattern}, where the conditional density of missing variables for pattern $r$ is matched with that for the group containing all the parent patterns of $r$. The second type and the third type are common mixture coefficients of pattern-mixture models \citep{little1993pattern}. They are usually used especially when the researchers have prior knowledge of how the mixture density is constructed.

We can incorporate information about mixture coefficients into the pattern graphs, allowing us to understand how the densities of missing variables are identified. Patterns that have only one parent maintain a constant “mixture coefficient”, specifically 1, and can be abbreviated in the graph.

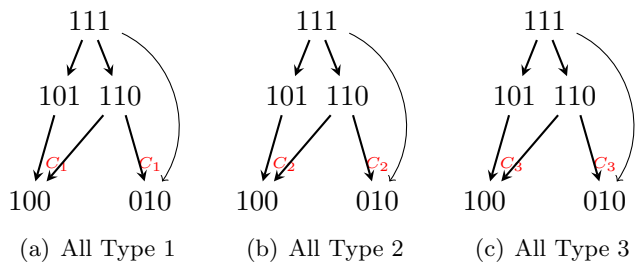
\begin{figure}[ht]
\centering
\subfigure[All Type 1]{
\begin{tikzpicture}[node distance=1cm, scale=0.5]
    \node(111){111};
    \node[below of = 111, xshift = 0.4cm](110){110};
    \node[below of = 111, xshift = -0.4cm](101){101};
    \node[below of = 110, xshift = 0.4cm, yshift = -0.4cm, label={[red]90:\tiny $C_1$}](010){010};
    \node[below of = 101, xshift = -0.4cm, yshift = -0.4cm, label={[red]75:\tiny $C_1$}](100){100};
    \path[line](111)--(110);
    \path[line](111)--(101);
    \path[line](110)--(010);
    \path[line](110)--(100);
    \path[line](101)--(100);
    \draw[->] (111) to[bend left = 50] (010);
\end{tikzpicture}
}
\subfigure[All Type 2]{
\begin{tikzpicture}[node distance=1cm, scale=0.5]
    \node(111){111};
    \node[below of = 111, xshift = 0.4cm](110){110};
    \node[below of = 111, xshift = -0.4cm](101){101};
    \node[below of = 110, xshift = 0.4cm, yshift = -0.4cm, label={[red]90:\tiny $C_2$}](010){010};
    \node[below of = 101, xshift = -0.4cm, yshift = -0.4cm, label={[red]75:\tiny $C_2$}](100){100};
    \path[line](111)--(110);
    \path[line](111)--(101);
    \path[line](110)--(010);
    \path[line](110)--(100);
    \path[line](101)--(100);
    \draw[->] (111) to[bend left = 50] (010);
\end{tikzpicture}
}
\subfigure[All Type 3]{
\begin{tikzpicture}[node distance=1cm, scale=0.5]
    \node(111){111};
    \node[below of = 111, xshift = 0.4cm](110){110};
    \node[below of = 111, xshift = -0.4cm](101){101};
    \node[below of = 110, xshift = 0.4cm, yshift = -0.4cm, label={[red]90:\tiny $C_3$}](010){010};
    \node[below of = 101, xshift = -0.4cm, yshift = -0.4cm, label={[red]75:\tiny $C_3$}](100){100};
    \path[line](111)--(110);
    \path[line](111)--(101);
    \path[line](110)--(010);
    \path[line](110)--(100);
    \path[line](101)--(100);
    \draw[->] (111) to[bend left = 50] (010);
\end{tikzpicture}
}
\caption{Mixture coefficients encoded in regular pattern graphs.}
\label{fig2-regular-graph}
\end{figure}

Our proposed method can handle the generalization if the mixture coefficients belong to the aforementioned types. For ease of exposition, we will first introduce the method with the particular choice of mixture coefficients as in \eqref{eqn:chen_mix}, while deferring the details related to the other two types to Appendix \ref{sec:method-general-assumptions}.

\section{The proposed method}
The proposed estimation of $\theta$ is motivated by \eqref{eqn:estimating}. The idea is to find appropriate weights $\hat{w}_i$'s and estimate $\theta$
by solving the weighted estimating equations:
\begin{align}\label{eqn:weighted_estimating_equations}
    \frac{1}{N}\sum_{i=1}^N \mathsf{1}_{R_i=1_d} \hat{w}_i \psi_\theta(L_i)=0.
\end{align}
From \eqref{eqn:estimating}, it is natural to choose the weight $\hat{w}_i$ as an estimator of $w_i := 1/\pi(L_i)$. The recursive form of inverse propensity \eqref{eqn:chen-ipw} allows us to construct an estimator of $w_i$ from estimators of $O^s$ for all missing patterns $s$. That is,
\begin{align}\label{eqn:recursive-weight}
	\hat{w}(L_i)
	=\sum_{r\in\calR}\hat{Q}^r(L_i)
	=\sum_{r\in\calR}\sum_{\Xi\in\Pi_r}\prod_{s\in\Xi}\hat{O}^s(L^{[s]}_i),
\end{align} 
where $\hat{Q}^r$ represents a generic estimator of $Q^r$ and $\hat{O}^s$ represents a generic estimator of $O^s$.

\subsection{Local estimation}
In this section, we focus on the appropriate estimator of $O^s$ and the corresponding weights described in \eqref{eqn:recursive-weight}. Recall that $O^r(l^{[r]})=P(R=r\mid l^{[r]})/P(R\in\PAr\mid l^{[r]})$, which focuses on pattern $r$ and its parents $\PAr$. The estimation can be done locally using the data from pattern $r$ and $\PAr$.

\subsubsection{Minimizing the entropy loss}
For instance, the estimation can be achieved by fitting a logistic regression \citep{chen2022pattern} with a binary outcome, where label $1$ refers to $R=r$ and label $0$ refers to $R\in\PAr$, and the feature/covariate as $\lob$. Fitting a logistic regression amounts to minimizing the entropy loss. Formally, for each missing pattern $r\neq 1_d$, we model the odds
\begin{align*}
O^r(\lob;\alphar)=\exp\left\{\Phir(\lob)\tr\alphar\right\}\ ,
\end{align*}
where $\Phir(\lob)=\left\{\phi^r_1(\lob),\ldots,\phi^r_{K_r}(\lob)\right\}$ are $K_r$ basis functions for the observed variables in pattern $r$. One may choose suitable basis functions depending on the observed variables and the number of observations in different patterns. The estimator of $\alphar$ is obtained by minimizing the empirical risk 
\begin{align*}
\frac{1}{N}\sum_{i=1}^N\left\{\mathsf{1}_{R_i\in\PAr}\log(1+O^r(\Lobi;\alphar))
+\mathsf{1}_{R_i=r}\log(1+O^r(\Lobi;-\alphar))\right\}\ .
\end{align*}
However, some propensity odds estimates may become extremely large and lead to an unstable estimation of $\theta$, even when an estimate of $P(R\in\PAr\mid l^{[r]})$ is small. 

\subsubsection{Minimizing the tailored loss}
An alternative approach is based on covariate balancing. One can show that for any measurable function $g$ of observed variables in pattern $r$,
\begin{align}\label{eqn:local-balance}
	\E\{\mathsf{1}_{R\in\PAr}O^r(\Lob)g(\Lob)\}
	=\E\{\mathsf{1}_{R=r}g(\Lob)\},
\end{align}
which constitute the balancing conditions. Covariate balancing approach achieves stable estimation by minimizing a choice of variability measure of $\hat{O}^r(L_i^{[r]})$ such that the empirical balancing conditions hold. 

For several standard problems, \citet{zhao2019covariate} shows that one can construct a tailored loss to encourage the empirical balancing conditions, and so the covariate balancing approach is essentially equivalent to (penalized) empirical risk minimization with respect to a tailored loss. We \citep{dong2024balancing} extend this idea and construct an approach to efficiently estimate $\theta$ under the CCMV assumption, which is a special case encoded in pattern graph Fig~\ref{fig1-CCMV}, where each missing pattern has only one parent pattern, $1_d$. 

The local part of a general pattern graph is very similar to the special one except the child pattern may have more than one part. So, it is natural to extend our previous work and derive the balancing weights in the following way. Note that the number of basis functions, $K_r$, is allowed to grow with sample size for flexible modeling. The estimator of $\alphar$ is obtained by minimizing the empirical tailored loss with penalization:
\begin{align}\label{penalized-tailored-loss}
\frac{1}{N}\sum_{i=1}^N\left\{\mathsf{1}_{R_i\in\PAr}O^r(\Lobi;\alphar)-\mathsf{1}_{R_i=r}\log O^r(\Lobi;\alphar)\right\}
+\lambda\sum_{k=1}^{K_r}t_k|\alphar_k|,
\end{align}
where the tuning parameter $\lambda\ge0$ controls the degree of penalization and can be chosen by a cross-validation procedure. The $l_1$-norm penalty is weighted by $t_k$ which represents the imbalance tolerance (or importance). Smaller $t_k$ should be assigned to the basis functions that are important to approximate the desired functions. 

Gathering the estimators of $O^s$ estimated by the tailored loss for all missing patterns, we can construct the weights $\hat{w}(L_i)$ through \eqref{eqn:recursive-weight}. By solving the weighted estimating equations \eqref{eqn:weighted_estimating_equations}, we achieve the estimator of $\theta$ and denote it as $\theta_{\mathrm{local}}$.

\textit{Remark.} We can also apply the similar strategy to the estimation of $\alphar$ when the entropy loss is used. So, we minimize the empirical entropy loss with penalization 
\begin{align}\label{penalized-entropy-loss}
	\frac{1}{N}\sum_{i=1}^N\left\{\mathsf{1}_{R_i\in\PAr}\log(1+O^r(\Lobi;\alphar))
	+\mathsf{1}_{R_i=r}\log(1+O^r(\Lobi;-\alphar))\right\}
	+\lambda\sum_{k=1}^{K_r}t_k|\alphar_k|\ .
\end{align}
Gathering the estimators of $O^s$ for all missing patterns, we can construct the weights, and solve the weighted estimating equations \eqref{eqn:weighted_estimating_equations}. We denote the estimator of $\theta$ using weights estimated by entropy loss as $\theta_{\mathrm{entropy}}$.

\subsection{Drawbacks of local estimation}
However, no matter which loss function is used, the local estimation has three drawbacks. Firstly, the errors could accumulate and escalate due to multiplication (See \eqref{eqn:recursive-weight}) when local estimator $\hat{O}^r$ are used to construct the weights $\hat{w}$. 

Secondly, only the evaluation of weights on the complete case, $\mathsf{1}_{R_i=1_d}\hat{w}(L_i)$, shows up in the weighted estimating equations and affects the final estimate. However, the aforementioned optimization problem is trained by data restricted to missing pattern $r$ and its parent set $\PAr$. If $1_d\notin\PAr$, we need to extrapolate the propensity odds model to achieve the evaluations $\mathsf{1}_{R_i=1_d}\hat{O}^r(\Lobi)$, which is assembled to construct $\mathsf{1}_{R_i=1_d}\hat{w}(L_i)$.The extrapolation process may introduce uncertainty and a higher risk of producing extremely large estimates. 

Lastly, \citet{chen2022pattern} claims that $\theta_{\mathrm{entropy}}$ is consistent and asymptotically normal. However, it does not achieve the asymptotic efficiency. The augmented method (AIPW) is efficient but requires repeated sampling and is computationally demanding. The consistency and asymptotic efficiency of $\theta_{\mathrm{local}}$ are established under the CCMV assumption \citep{dong2024balancing}, which is a special case of \eqref{eqn:chen-propensity} where any missing pattern has only one parent $1_d$.  It requires further study to extend the asymptotic properties to the general case of \eqref{eqn:chen-propensity}.

\subsection{Sequential estimation using balancing method}\label{sec:sequential}
To address these potential issues, we propose the sequential balancing approach. Instead of considering the balancing conditions \eqref{eqn:local-balance} by $O^r$, which focus locally on the balance between $r$ and $\PAr$, we examine the balancing conditions by $Q^r(\lob)=P(R=r\mid\lob)/P(R=1_d\mid\lob)$, which connects pattern $r$ with complete cases $1_d$. Recall the recursive form \eqref{eqn:chen-ipw} that $Q^r(l)=O^r(\lob)\times \sum_{s\in\PAr}Q^s(l)$. Therefore, $Q^r$ can be estimated sequentially. The recursive form encourages the following balancing conditions:
\begin{align}\label{eqn:sequential-balance}
	\E\{\mathsf{1}_{R=r}g(\Lob)\}
    =\E\{\mathsf{1}_{R=1_d}Q^r(L)g(\Lob)\}
	=\E\left\{\mathsf{1}_{R=1_d}O^r(\Lob)\left[\sum_{s\in\PAr}Q^s(L)\right]g(\Lob)\right\}\ .
\end{align}

Suppose that we have estimators $\hat{Q}^s(l)$ for $Q^s(l)$ for all $s\in\PAr$. Abbreviate the summation $\sum_{s\in\PAr}\hat{Q}^s(l)$ by $\hat{Q}^{\PAr}(l)$. To estimate $Q^r$, one can seek the estimator of $O^r$, denoted by $\hat{O}^r$, that encourages the empirical version of \eqref{eqn:sequential-balance}, which equate the empirical average over pattern $r$ and the reweighted average over complete cases $1_d$:
\begin{align}\label{eqn:empirical-reweighted-balance}
	\sum_{i=1}^N\mathsf{1}_{R_i=r}g(\Lobi)
    =\sum_{i=1}^N\mathsf{1}_{R_i=1_d}\hat{O}^r(\Lobi)\hat{Q}^{\PAr}(L_i)g(\Lobi)\ .
\end{align}

In Appendix \ref{sec:sequential-balance}, we proved that minimizing the following sequential balancing loss \eqref{seq_loss} imposes the empirical balance \eqref{eqn:empirical-reweighted-balance}. Let
\begin{align}
\calLr\{O^r(\lob;\alphar),R\}
=\mathsf{1}_{R=1_d}O^r(\lob;\alphar)\hat{Q}^{\PAr}(l)
-\mathsf{1}_{R=r}\log O^r(\lob;\alphar)\ .
\label{seq_loss}
\end{align}

The estimator of $\alphar$, denoted as $\hatalphar$, is obtained by minimizing the empirical sequential balancing loss with penalization:
\begin{align}\label{penalized-seq-loss}
\calLr_\lambda(\alphar)=
\frac{1}{N}\sum_{i=1}^N\calLr\{O^r(\Lobi;\alphar),R\}
+\lambda\sum_{k=1}^{K_r}t_k|\alphar_k|\ .
\end{align}

Then, gathering the estimators $\hatalphar$ for all missing patterns, we can construct the propensity odds estimates $O^r(\Lobi;\hatalphar)$ and $\hat{Q}^r(L_i)$, and weights $\hat{w}(L_i)=\sum_{r\in\calR}\hat{Q}^r(L_i)$. By solving the weighted estimating equations \eqref{eqn:weighted_estimating_equations}, we achieve the estimator of $\theta$ and denote it as $\theta_{\mathrm{seq}}$.

Formally, we propose the following sequential estimation algorithm. 
\begin{algorithm}[ht]
	\caption{Sequential estimation}
	\label{alg:seq}
	\begin{algorithmic}
		\STATE Note that $Q^{1_d}(l)=1$. Run the following steps for each $r\in\calR$ where $d_r=n-1$. Next, repeat the process for each $r\in\calR$ where $d_r=n-2$, and so on, until process each $r\in\calR$ where $d_r=1$. 
		\STATE {\bfseries Input:} The propensity odds estimates on complete cases, $\{\mathsf{1}_{R_i=1_d}\hat{Q}^{\PAr}(L_i)\}_{i=1}^N$.
		\STATE {\bfseries Step 1:} Solve the optimization using sequential loss function \eqref{penalized-seq-loss}, and obtain the model parameter $\hatalphar$.
		\STATE {\bfseries Step 2:} Obtain the estimates $\{\mathsf{1}_{R_i=1_d}O^r(L_i;\hatalphar)\}_{i=1}^N$.
		\STATE {\bfseries Step 3:} Construct the estimates $\{\mathsf{1}_{R_i=1_d}\hat{Q}^r(L_i)\}_{i=1}^N$ by recursive form \eqref{eqn:chen-ipw}. 
        \STATE {\bfseries Output:} When the above estimation is done, obtain $\theta_{\mathrm{seq}}$ by solving the weighted estimating equations with $\hat{w}(L_i)=\sum_{r\in\calR}\hat{Q}^r(\Lobi)$.
	\end{algorithmic}
\end{algorithm}

The advantage of the sequential estimation procedure is apparently in its structure. The estimators $\hat{Q}^s$ with $s\in\PAr$ are utilized for the estimation of $Q^r$. The multiplication terms are naturally controlled in the loss minimization procedure. Additionally, the extrapolation issue is addressed since the data in patterns $r$ and $1_d$ are used to fit the propensity model. So, we do not extrapolate the model for estimation. In the subsequent section, we will show that the proposed estimator of propensity odds is consistent and the resulting estimator of $\theta$ is consistent and asymptotically efficient.

\section{Asymptotic properties}
\label{s:asymptotic}
In this section, we first investigate the asymptotic variance lower bound for all regular estimators of $\theta$. We then develop the asymptotic normality and efficiency of the proposed estimators.

Recall that $\theta_0$ is the unique solution to $\E\{\psi_\theta(L)\}=0$. The concept ``regular estimator'' is defined according to \citet{begun1983information} and \citet{ibragimov2013statistical}. We require the following set of assumptions to establish the asymptotic theory.
\begin{assumption}\label{assump1}$ $ 
	\begin{enumerate}[label={\textbf{~\Alph*:}},ref={Assumption~\theassumption.\Alph*},leftmargin=1cm]
		\item\label{assump-1A}
		The estimating function $\psi(L,\theta)$ is differentiable with respect to $\theta$ with derivative $\dot{\psi}_\theta(L)$. Also, $\E\{\psi_\theta(L)\}$ has the unique root $\theta_0$ and is differentiable at $\theta_0$ with nonsingular derivative $D_{\theta_0}$.
		\item\label{assump-1B}
		There exists a constant $\delta_0>0$ such that $P(R=1_d\mid\lob)\ge\delta_0$ for any $r\in\calR$ and so $1_d\in\calR$.\ .
	\end{enumerate}
\end{assumption}
\ref{assump-1A} is a standard regularity assumption for Z-estimation. \ref{assump-1B} ensures that complete cases are available for analysis. Then, we claim the following efficiency bound under the proposed identifying assumptions \eqref{eqn:chen-propensity}.
\begin{theorem}\label{efficiency-bound}
	Under Assumption \ref{assump1}, the asymptotic variance lower bound for all regular estimators of $\theta_0$ is $D_{\theta_0}^{-1}V_{\theta_0}D_{\theta_0}^{-1\tr}$, where $V_\theta=\E\{F_\theta(L,R)F_\theta(L,R)\tr\}$ and
	\begin{align*}
		&F_\theta(L,R)
		=\mathsf{1}_{R=1_d}\left\{1+\sum_{\Xi\in\Pi}\prod_{s\in\Xi}O^s(L^{[s]})\right\}\psi_\theta(L)\\
		&+\sum_{1_d\neq r\in\calR}\sum_{\Xi\in\Pi_r}\sum_{s\in\Xi}F_\theta^{\Xi,s}(L,R)\ .
	\end{align*}
	Under the identifying assumptions \eqref{eqn:chen-propensity}, $F_\theta^{\Xi,s}(L,R)=\mathsf{1}_{R=s}u_\theta^s(l^{[s]})-\mathsf{1}_{R\in\PAs}O^s(l^{[s]})u_\theta^s(l^{[s]})$ where $u_\theta(\lob)=\E\{\psi_\theta(L)\mid\Lob=\lob,R=r\}$.
\end{theorem}
The detailed proof is in the Appendix \ref{sec:proof-efficiency-bound}.

Now, we consider the weights, $\hat{w}(L_i)=\sum_{r\in\calR}\hat{Q}^r(\Lobi)$ obtained from Algorithm \ref{alg:seq}, and construct the weighted estimator of $\E\{\psi_\theta(L)\}$: 
\begin{align*}
\hat{\bbP}_N\psi_\theta=\frac{1}{N}\sum_{i=1}^N\left\{\mathsf{1}_{R_i=1_d}\hat{w}(L_i)\psi_\theta(L_i)\right\},
\end{align*} 
The resulting estimator of $\theta_0$ is the solution to $\hat{\bbP}_N\psi_\theta=0$. Denote it by $\hat{\theta}_N$. Under mild conditions, we show that $O^r(\lob;\hatalphar)$ is consistent, $\hat{\bbP}_N\psi_\theta$ is asymptotically normal for each $\theta$ in a compact set $\Theta\subset\bbR^q$, and $\hat{\theta}_N$ is consistent and efficient. 
\begin{theorem}\label{theta}
Suppose that Assumptions \ref{assump1} and \ref{assump2}--\ref{assump4} hold. Then
\begin{align*}
\hat{\theta}_N\xrightarrow{P}\theta_0
\end{align*}
and
\begin{align*}
N^{\frac{1}{2}}(\hat{\theta}_N-\theta_0)\overset{d}{\to}N(0,D_{\theta_0}^{-1}V_{\theta_0}D_{\theta_0}^{-1\tr})\ ,
\end{align*}
where $D_{\theta_0}^{-1}V_{\theta_0}D_{\theta_0}^{-1\tr}$ is the asymptotic variance bound in Theorem \ref{efficiency-bound}. Therefore, $\hat{\theta}_N$ is semiparametrically efficient. 
\end{theorem}
The proof is given in the Appendix.

\section{Simulation}
\label{s:simulation}
A simulation study is conducted to evaluate the finite-sample performance of the proposed estimators. We designed a missing mechanism that can be represented by the following pattern graph (Figure \ref{graph-sim}). We simulated 1,000 independent data sets, each of size $N$=1,000, where $X_j,\ j=1,2,3,4$, are generated independently from a truncated standard normal distribution with support $[-3,3]$. We considered a logistic regression model $\mathrm{logit}\{P(Y=1\mid X)\}=\theta_0+\theta_1X_1+\theta_2X_2+\theta_3X_3+\theta_4X_4$ where the true coefficients $\theta_0=(3,-2,1,2,-1)$ are the parameters of interest. We generated eight non-monotone response patterns where any variable could be missing. Denote each response pattern by the corresponding binary vector. So, the set of all possible missing patterns is $\calR=\{11111, 01111, 10111, 11110, 11001, 10110, 11010, 11000\}$. The response patterns are generated from a multinomial distribution with the probabilities $P(R=r\mid l)$ calculated from the recursive form \eqref{eqn:chen-ipw} where propensity odds $O^r(\lob)$ are polynomials of observed variables with degrees up to four.

\textsc{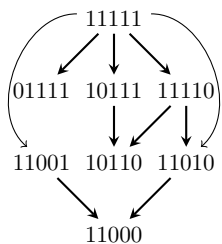
\begin{figure}[ht]
\centering
\begin{tikzpicture}[node distance=1.2cm, every node/.style={scale=0.8}]
    \node(1d){11111};
    \node[below of = 1d, xshift = -1.2cm](r1){01111};
    \node[below of = 1d](r2){10111};
    \node[below of = 1d, xshift = 1.2cm](r3){11110};
    \node[below of = r2, xshift = -1.2cm](s1){11001};
    \node[below of = r2](s2){10110};
    \node[below of = r2, xshift = 1.2cm](s3){11010};
    \node[below of = s2](t){11000};
    \path[line](1d)--(r1);
    \path[line](1d)--(r2);
    \path[line](1d)--(r3);
    \path[line](r2)--(s2);
    \path[line](r3)--(s2);
    \path[line](r3)--(s3);
    \path[line](s1)--(t);
    \path[line](s3)--(t);
    \draw[->] (1d) to[bend right = 70] (s1);
    \draw[->] (1d) to[bend left = 70] (s3);
\end{tikzpicture}
\caption{A regular pattern graph for simulation}
\label{graph-sim}
\end{figure}}

We first analyzed the simulated data with the full dataset (Full), which is the ideal case with no missingness. We then analyzed the data in the complete case pattern (Complete-case), for which data in all missing patterns $r\neq 1_d$ are discarded, and an unweighted analysis is used for the remaining data. Next, we considered the inverse propensity weighting methods with the true inverse propensity weights (True-weight). We also examined the performance of the estimators based on the estimated propensity odds using the different loss functions (Entropy, Local, Sequential). We model the propensity odds with basis functions $\Phir(\lob)$ where six splines of degrees up to four are chosen for each continuous variable, and a binary indicator function is chosen for discrete variables. Given the propensity odds estimators obtained from minimizing the penalized empirical loss \eqref{penalized-entropy-loss}, \eqref{penalized-tailored-loss} and \eqref{penalized-seq-loss}, we construct the estimators $\hat{\pi}(l)$ for $\pi(l)=P(R=1_d\mid l)$. So, $\theta_{\mathrm{entropy}}$, $\theta_{\mathrm{local}}$ and $\theta_{\mathrm{seq}}$ can be obtained by solving the weighted estimating equation: $N^{-1}\sum_{i=1}^N \mathsf{1}_{R_i=1_d}\psi_\theta(L_i)/\hat{\pi}(L_i)=0$ with corresponding loss.
  
The biases and mean squared errors of each coefficient are shown in Table \ref{simulation}. We notice that the local estimations (both Entropy and Local) fail in around $5\%$ dataset under the above setting. The sequential estimation provides smaller errors than the other two IPW estimations. It is expected since sequential estimation not only encourages the balance of observed variables but also alleviates the extrapolation issue. 

\begin{table*}[t]
	\caption{Results of the simulation study based on 1000 replications.}
	\label{simulation}
	\begin{center}
		\begin{tabular}{lcccccccccccc}
			\toprule
			\multicolumn{1}{l}{Method} & \multicolumn{6}{c}{Bias} & \multicolumn{6}{c}{MSE}\\
            \cmidrule(lr){2-7} \cmidrule(lr){8-13}
			& $\theta_1$ & $\theta_2$ & $\theta_3$ & $\theta_4$ & $\theta_5$ & $\|\cdot\|_1$ & $\theta_1$ & $\theta_2$ & $\theta_3$ & $\theta_4$ & $\theta_5$ & $\|\cdot\|_2$\\
			\hline
			Full & 0.04 & -0.03 & 0.02 & 0.03 & -0.01 & 0.67 
            & 0.05 & 0.03 & 0.02 & 0.03 & 0.02 & 0.15 \\
			CC & 0.79 & -0.14 & 0.03 & 0.07 & -0.34 & 2.09 
            & 0.90 & 0.19 & 0.12 & 0.15 & 0.25 & 1.61 \\ 
			True & 0.49 & -0.31 & 0.11 & 0.25 & -0.18 & 2.45 
            & 0.77 & 0.43 & 0.26 & 0.35 & 0.31 & 2.12 \\ 
			Entropy & 0.56 & -0.33 & 0.10 & 0.27 & -0.19 & 2.50 
            & 0.80 & 0.45 & 0.28 & 0.38 & 0.32 & 2.23 \\ 
			Local & 0.37 & -0.19 & 0.03 & 0.14 & -0.16 & 1.90 
            & 0.46 & 0.27 & 0.16 & 0.21 & 0.19 & 1.29 \\ 
			Seq & 0.32 & -0.17 & 0.02 & 0.14 & -0.17 & 1.82 
            & 0.39 & 0.24 & 0.15 & 0.20 & 0.19 & 1.17 \\ 
			\bottomrule 
		\end{tabular}
	\end{center}
\end{table*}

We also perform the sensitivity analysis based on identifying assumptions. Two misspecified pattern graphs are constructed (Figure \ref{graph-sensitivity}, where the first one corresponds to CCMV and the second one has one missing edge compared to the correct graph. Sequential estimation provides more robust results (See Table \ref{simulation-sensitivity}). While the estimation under the misspecified CCMV assumption provides smaller errors.
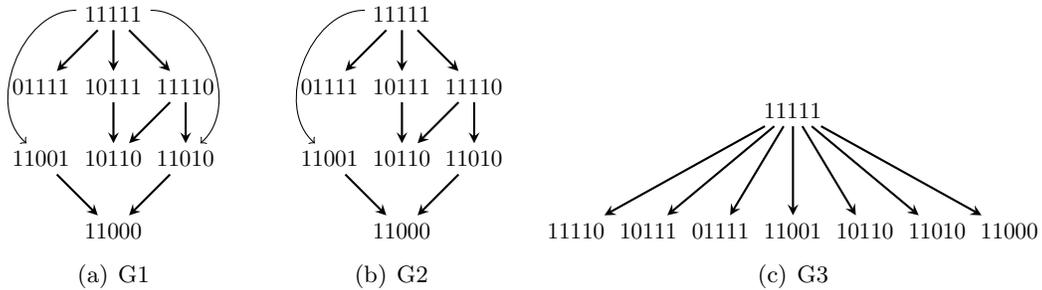
\begin{figure}
    \centering
    \subfigure[G1]{
    \begin{tikzpicture}[node distance=1.2cm, every node/.style={scale=0.8}]
        \node(1d){11111};
        \node[below of = 1d, xshift = -1.2cm](r1){01111};
        \node[below of = 1d](r2){10111};
        \node[below of = 1d, xshift = 1.2cm](r3){11110};
        \node[below of = r2, xshift = -1.2cm](s1){11001};
        \node[below of = r2](s2){10110};
        \node[below of = r2, xshift = 1.2cm](s3){11010};
        \node[below of = s2](t){11000};
        \path[line](1d)--(r1);
        \path[line](1d)--(r2);
        \path[line](1d)--(r3);
        \path[line](r2)--(s2);
        \path[line](r3)--(s2);
        \path[line](r3)--(s3);
        \path[line](s1)--(t);
        \path[line](s3)--(t);
        \draw[->] (1d) to[bend right = 70] (s1);
        \draw[->] (1d) to[bend left = 70] (s3);
    \end{tikzpicture}
    }
    \subfigure[G2]{
    \begin{tikzpicture}[node distance=1.2cm, every node/.style={scale=0.8}]
        \node(1d){11111};
        \node[below of = 1d, xshift = -1.2cm](r1){01111};
        \node[below of = 1d](r2){10111};
        \node[below of = 1d, xshift = 1.2cm](r3){11110};
        \node[below of = r2, xshift = -1.2cm](s1){11001};
        \node[below of = r2](s2){10110};
        \node[below of = r2, xshift = 1.2cm](s3){11010};
        \node[below of = s2](t){11000};
        \path[line](1d)--(r1);
        \path[line](1d)--(r2);
        \path[line](1d)--(r3);
        \path[line](r2)--(s2);
        \path[line](r3)--(s2);
        \path[line](r3)--(s3);
        \path[line](s1)--(t);
        \path[line](s3)--(t);
        \draw[->] (1d) to[bend right = 70] (s1);
    \end{tikzpicture}
    }
    \subfigure[G3]{
    \begin{tikzpicture}[node distance=1cm, every node/.style={scale=0.8}]
        \node(1d){11111};
        \node[below of = 1d, xshift = -1.2cm, yshift = -1cm](r1){01111};
        \node[below of = 1d, xshift = -2.4cm, yshift = -1cm](r2){10111};
        \node[below of = 1d, xshift = -3.6cm, yshift = -1cm](r3){11110};
        \node[below of = 1d, yshift = -1cm](s1){11001};
        \node[below of = 1d, xshift = 1.2cm, yshift = -1cm](s2){10110};
        \node[below of = 1d, xshift = 2.4cm, yshift = -1cm](s3){11010};
        \node[below of = 1d, xshift = 3.6cm, yshift = -1cm](t){11000};
        \path[line](1d)--(r1);
        \path[line](1d)--(r2);
        \path[line](1d)--(r3);
        \path[line](1d)--(s1);
        \path[line](1d)--(s2);
        \path[line](1d)--(s3);
        \path[line](1d)--(t);
    \end{tikzpicture}
    }
\caption{A regular pattern graph for simulation}
\label{graph-sensitivity}
\end{figure}

\begin{table*}[t]
	\caption{Results of the simulation study based on 1000 replications.}
	\label{simulation-sensitivity}
	\begin{center}
		\begin{tabular}{lcccccccccccc}
			\toprule
			\multicolumn{1}{l}{Method} & \multicolumn{6}{c}{Bias} & \multicolumn{6}{c}{MSE}\\
            \cmidrule(lr){2-7} \cmidrule(lr){8-13}
			& $\theta_1$ & $\theta_2$ & $\theta_3$ & $\theta_4$ & $\theta_5$ & $\|\cdot\|_1$ & $\theta_1$ & $\theta_2$ & $\theta_3$ & $\theta_4$ & $\theta_5$ & $\|\cdot\|_2$\\
			\hline
			Entropy(G1) & 0.56 & -0.33 & 0.10 & 0.27 & -0.19 & 2.50 
            & 0.80 & 0.45 & 0.28 & 0.38 & 0.32 & 2.23 \\  
			Entropy(G2)  & 0.72 & -0.42 & 0.12 & 0.33 & -0.21 & 2.83 
            & 1.12 & 0.59 & 0.34 & 0.48 & 0.37 & 2.89 \\ 
            Entropy(G3) & 0.43 & -0.25 & 0.06 & 0.22 & -0.18 & 2.17 
            & 0.57 & 0.34 & 0.21 & 0.28 & 0.26 & 1.67 \\ 
			Local(G1) & 0.37 & -0.19 & 0.03 & 0.14 & -0.16 & 1.90 
            & 0.46 & 0.27 & 0.16 & 0.21 & 0.19 & 1.29 \\ 
			Local(G2) & 0.48 & -0.23 & 0.04 & 0.16 & -0.16 & 2.03 
            & 0.59 & 0.31 & 0.18 & 0.23 & 0.20 & 1.50 \\ 
			Sequential(G1)  & 0.32 & -0.17 & 0.02 & 0.14 & -0.17 & 1.82
            & 0.39 & 0.24 & 0.15 & 0.20 & 0.19 & 1.17 \\ 
			Sequential(G2) & 0.36 & -0.19 & 0.02 & 0.15 & -0.16 & 1.89 
            & 0.44 & 0.26 & 0.16 & 0.21 & 0.19 & 1.27 \\ 
			Sequential(G3) & 0.31 & -0.16 & 0.02 & 0.13 & -0.19 & 1.77 
            & 0.37 & 0.23 & 0.14 & 0.19 & 0.19 & 1.11 \\ 
            \bottomrule 
		\end{tabular}
	\end{center}
\end{table*}

\section{Real data analysis}
This section presents a real-world example to illustrate the proposed methodology, using data from a survey on public responses to the economic crisis \citep{burns2012risk}. Risk perceptions can vary widely among individuals, influenced by personal characteristics and emotions. The key variables considered include age, gender, income, and attitudes toward risk in both investments and jobs. The focus of this analysis is on the coefficients of a logistic regression model, where these five variables serve as predictors, and the outcome of interest is whether participants made riskier investments in the week before completing the questionnaire. We focus on the seventh wave in the serial survey and remove data from participants who are not in this survey. Eight response patterns are observed.

The choice of missing mechanisms, represented by different pattern graphs, is essential for unbiased estimation. For the sensitivity analysis, we examine three missing mechanisms, each illustrated by pattern graphs in Figure \ref{graph-realdata}. The first mechanism follows the CCMV assumption, while the other two adopt a hierarchical structure. In these two cases, each parent set is selected from the patterns one layer above, based on the idea that missing patterns differing by only one observed variable should exhibit greater similarity.

We present the parameter estimates and p-values in Table~\ref{realdata}. Our proposed estimators yield consistent results across three different missing mechanisms. A key observation is that the coefficient for JOB is marginally significant, suggesting that risk perception related to one's job plays an important role in decision-making. However, a notable difference is that the coefficients and their p-values for AGE and INCOME vary substantially across the different missing mechanisms, indicating that complete case analysis may lead to biased estimates.

\begin{figure}[ht]
	\centering
    \subfigure[CCMV]{
	\begin{tikzpicture}[node distance=1cm, every node/.style={scale=0.8}]
	\node(1d){11111};
	\node[below of = 1d, xshift = -1.2cm, yshift = -1cm](r1){111101};
	\node[below of = 1d, xshift = -2.4cm, yshift = -1cm](r2){111011};
	\node[below of = 1d, xshift = -3.6cm, yshift = -1cm](r3){110111};
	\node[below of = 1d, yshift = -1cm](s1){111001};
	\node[below of = 1d, xshift = 1.2cm, yshift = -1cm](s2){110011};
	\node[below of = 1d, xshift = 2.4cm, yshift = -1cm](t1){111000};
	\node[below of = 1d, xshift = 3.6cm, yshift = -1cm](t2){110001};
	\path[line](1d)--(r1);
	\path[line](1d)--(r2);
	\path[line](1d)--(r3);
	\path[line](1d)--(s1);
	\path[line](1d)--(s2);
	\path[line](1d)--(t1);
	\path[line](1d)--(t2);
	\end{tikzpicture}
    }
    \subfigure[Graph1]{
	\begin{tikzpicture}[node distance=1.2cm, every node/.style={scale=0.8}]
		\node(1d){111111};
		\node[below of = 1d, xshift = -1.2cm](r1){111101};
		\node[below of = 1d](r2){111011};
		\node[below of = 1d, xshift = 1.2cm](r3){110111};
		\node[below of = r2, xshift = -1cm](s1){111001};
		\node[below of = r2, xshift = 1cm](s2){110011};
		\node[below of = s1](t1){111000};
		\node[below of = s2](t2){110001};
		\path[line](1d)--(r1);
		\path[line](1d)--(r2);
		\path[line](1d)--(r3);
		\path[line](r1)--(s1);
		\path[line](r2)--(s2);
		\path[line](r3)--(s2);
		\path[line](s1)--(t1);
		\path[line](s2)--(t2);
	\end{tikzpicture}
    }
    \subfigure[Graph2]{
	\begin{tikzpicture}[node distance=1.2cm, every node/.style={scale=0.8}]
	\node(1d){111111};
	\node[below of = 1d, xshift = -1.2cm](r1){111101};
	\node[below of = 1d](r2){111011};
	\node[below of = 1d, xshift = 1.2cm](r3){110111};
	\node[below of = r2, xshift = -1cm](s1){111001};
	\node[below of = r2, xshift = 1cm](s2){110011};
	\node[below of = s1](t1){111000};
	\node[below of = s2](t2){110001};
	\path[line](1d)--(r1);
	\path[line](1d)--(r2);
	\path[line](1d)--(r3);
	\path[line](r1)--(s1);
	\path[line](r2)--(s1);
	\path[line](r2)--(s2);
	\path[line](r3)--(s2);
	\path[line](s1)--(t1);
	\path[line](s1)--(t2);
	\path[line](s2)--(t2);
	\end{tikzpicture}
    }
	\caption{A regular pattern graph for real data analysis}
	\label{graph-realdata}
\end{figure}
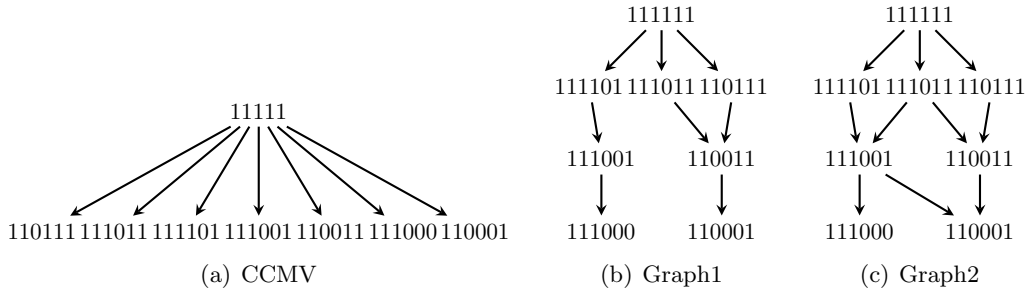

\begin{table*}[t]
	\caption{Results of the Financial Crisis Data analysis: Estimates and p-values.} 
	\label{realdata}
	\begin{center}
		\begin{tabular}{lcccccccc}
			\toprule
			\multicolumn{1}{l}{ } & \multicolumn{2}{c}{Complete-case} & \multicolumn{2}{c}{CCMV}
			& \multicolumn{2}{c}{Graph 1} & \multicolumn{2}{c}{Graph 2}\\ 
			\cmidrule(lr){2-3}
			\cmidrule(lr){4-5}
			\cmidrule(lr){6-7}
			\cmidrule(lr){8-9}
			Parameters & Estimate & p-value & Estimate & p-value & Estimate & p-value & Estimate & p-value  \\ 
			\hline
			AGE & -0.02 & 0.37 & -0.02 & 0.40 & -0.01 & 0.56 & -0.01 & 0.58 \\  
			GENDER & -0.22 & 0.63 & 0.02 & 0.97 & 0.14 & 0.80 & 0.13 & 0.81 \\   
			INCOME & -0.02 & 0.92 & 0.08 & 0.58 & 0.14 & 0.40 & 0.15 & 0.42 \\   
			INVESTMT & -0.18 & 0.46 & -0.24 & 0.28 & -0.23 & 0.30 & -0.22 & 0.31 \\   
			JOB & 0.32 & 0.15 & 0.37 & 0.06 & 0.35 & 0.07 & 0.37 & 0.06 \\ 
			\bottomrule
		\end{tabular}
	\end{center}
\end{table*}


\section*{Acknowledgements}
This work was based on part of the PhD dissertation of the first author,
and was completed while the first author was affiliated with Texas A\&M University.

\begin{appendices}
\section{Proof of Sequential Balance}\label{sec:sequential-balance}
\begin{proof}
Recall the sequential balancing loss function is:
\begin{align*}
\calLr\{O^r(\lob;\alphar),R\}
=\mathsf{1}_{R=1_d}O^r(\lob;\alphar)\hat{Q}^{\PAr}(l)
-\mathsf{1}_{R=r}\log O^r(\lob;\alphar)\ .
\end{align*}
Define the average loss:
\begin{align*}
\calLr_N(\alphar)
=\frac{1}{N}\sum_{i=1}^N\calLr\{O^r(\Lobi;\alphar),R_i\}\ .
\end{align*}
The derivative of $\calLr_N(\alphar)$ with respect to $\alphar$ is:
\begin{align*}
\nabla\calLr_N(\alphar)
=\sum_{i=1}^N\mathsf{1}_{R_i=1_d}w_i\hat{Q}^{\PAr}(l)\Phir(\Lobi)
-\sum_{i=1}^N\mathsf{1}_{R_i=r}\Phir(\Lobi)\ .
\end{align*}
Denote the minimizer of average loss by $\hat{\alpha}^{[r]}$. So, $w_i=O^r(\lob;\hat{\alpha}^{[r]})$. The minimizer satisfies $\nabla\calLr_N(\alphar)=0$, which can be rewritten as:
\begin{align*}
	\sum_{i=1}^N\mathsf{1}_{R_i=r}\Phir(\Lobi)
	=\sum_{i=1}^N\mathsf{1}_{R_i=1_d}w_i\hat{Q}^{\PAr}(l)\Phir(\Lobi)\ .
\end{align*}
Therefore, the proposed sequential balancing loss function encourages the empirical balance between pattern $r$ and $1_d$. The balancing condition in Section \ref{sec:sequential} holds for a general function $g$, instead of the basis functions $\Phir$. To achieve the balance of a desired function, one wants to cautiously choose the basis functions.
\end{proof}

\section{The proposed method under generalized missing data assumptions encoded in pattern graph}\label{sec:method-general-assumptions}

\section{Proof of Theorem \ref{efficiency-bound}}\label{sec:proof-efficiency-bound}
{\bf Proof sketch:}
To show $D_{\theta_0}^{-1}V_{\theta_0}D_{\theta_0}^{-1\tr}$ is the efficiency bound, we closely follow the structure of semiparametric efficiency bound derivation of \cite{newey1990semiparametric}, \cite{bickel1993efficient} and \cite{chen2008semiparametric}. Briefly speaking, we want to utilize Theorem 3.1 in \cite{newey1990semiparametric} to calculate the efficiency bound.

Firstly, pathwise differentiability follows if we can find an influence function satisfying \eqref{path-differentiability} for all regular parametric submodels. Calculation of the tangent set is typically straightforward. $\calT$ is defined as the mean square closure of all $q$-dimensional linear combinations of scores for all smooth functions. Calculation of the projection can be difficult. However, the influence functions we found in the previous step are in the tangent set for the three cases mentioned in this paper, which completes the proof.
\begin{proof}
Consider an \textbf{arbitrary parametric} submodel for the joint density of the $(\lob,R)$ with parameter $\beta$:
\begin{align*}
f_\beta(\lob,r)
=\prod_{s\in\calR}\left\{P_\beta(R=s)f_\beta(l^{[s]}\mid R=s)\right\}^{\mathsf{1}_{r=s}}
\end{align*}
where $\beta_0$ gives the true distribution. The resulting score is given by
\begin{align}\label{parametric-score}
S_\beta(l,r)
=\sum_{s\in\calR}\mathsf{1}_{r=s}S_\beta(l^{[s]}\mid R=s)
+\sum_{s\in\calR}\mathsf{1}_{r=s}\frac{\dot{P}_\beta(R=s)}{P_\beta(R=s)}
\end{align}
where $S_\beta(l^{[s]}\mid R=s)=\partial\log f_\beta(l^{[s]}\mid R=s)/\partial\beta$ satisfies $\int S_\beta(l^{[s]}\mid R=s)f_\beta(l^{[s]}\mid R=s)dl^{[s]}=0$ for $s\in\calR$. Besides, $\sum_{s\in\calR}\E\{\mathsf{1}_{R=s}\}\dot{P}_\beta(R=s)/P_\beta(R=s)=0$. 

Recall that the parameter of interest $\theta_0$ is the solution to $\E\{\psi_\theta(L)\}=0$ and thus is a function of $\beta$, denoted by $\theta_0(\beta)$. To apply Theorem 3.1 in \cite{newey1990semiparametric}, we firstly prove that $\theta_0(\beta)$ is differentiable. Pathwise differentiability follows if we can find an influence function $\zeta(L,R)$ for all regular parametric submodels such that
\begin{align}\label{path-differentiability}
\frac{\partial\theta_0(\beta_0)}{\partial\beta}
=\E\{\zeta(L,R)S_{\beta_0}(L,R)\}\ .
\end{align}
To save notations, $\beta$ is also used as the true parameter value $\beta_0$. Chain rule and Leibniz integral rule (differentiating under the integral) gives
\begin{align*}
\frac{\partial\E\{\psi_\theta(L)\}}{\partial\beta}
&=\int\frac{\partial\psi_\theta(l)f_\beta(l)}{\partial\beta}dl
=\int\left\{\frac{\partial\psi_\theta(l)}{\partial\theta}\frac{\partial\theta(\beta)}{\partial\beta}f_\beta(l)
+\psi_\theta(l)\frac{\partial f_\beta(l)}{\partial\beta}\right\}dl\\
&=\frac{\partial\theta(\beta)}{\partial\beta}\int\frac{\partial\psi_\theta(l)}{\partial\theta}f_\beta(l)dl
+\int\psi_\theta(l)\frac{\partial\log f_\beta(l)}{\partial\beta}f_\beta(l)dl\\
&=\frac{\partial\theta(\beta)}{\partial\beta}\frac{\partial\E\{\psi_\theta(L)\}}{\partial\theta}
+\E\left\{\psi_\theta(L)\frac{\partial\log f_\beta(L)}{\partial\beta}\right\}\ .
\end{align*}
Therefore, by the fact that $\E\{\psi_\theta(L)\}=0$,
\begin{align*}
\frac{\partial\theta_0(\beta)}{\partial\beta}
=-\left[\left.\frac{\partial\E\{\psi_\theta(L)\}}{\partial\theta}\right\vert_{\theta_0}\right]^{-1}\E\left\{\psi_{\theta_0}(L)\frac{\partial\log f_\beta(L)}{\partial\beta}\right\}\ .
\end{align*}
The marginal density of $L$ is
\begin{align*}
f_\beta(l)
=\sum_{r\in\calR}f_\beta(l,r)
=P_\beta(R=1_d)f_\beta(l\mid R=1_d)
+\sum_{1_d\neq r\in\calR}P_\beta(R=r)f_\beta(\lob\mid R=r)f_\beta(\lms\mid\lob,R=r)\ .
\end{align*}
Then,
\begin{align}
\E\left\{\psi_\theta(L)\frac{\partial\log f_\beta(L)}{\partial\beta}\right\}
&=\int\psi_\theta(l)\frac{\partial P_\beta(R=1_d)f_\beta(l\mid R=1_d)}{\partial\beta}dl\label{EIF1}\\
&+\sum_{1_d\neq r\in\calR}\int\psi_\theta(l)\frac{\partial P_\beta(R=r)f_\beta(\lob\mid R=r)f_\beta(\lms\mid\lob,R=r)}{\partial\beta}dl\label{EIF2}\ .
\end{align}
The first term on the right hand side of the equation \eqref{EIF1} is 
\begin{align}
&\quads\int\psi_\theta(l)\frac{\partial P_\beta(R=1_d)f_\beta(l\mid R=1_d)}{\partial\beta}dl\label{ctrb1}\\
&=\dot{P}_\beta(R=1_d)\int\psi_\theta(l)f_\beta(l\mid R=1_d)dl
+\int\psi_\theta(l)P_\beta(R=1_d)S_\beta(l\mid R=1_d)f_\beta(l\mid R=1_d)dl\nonumber\\
&=\frac{\E\left\{\mathsf{1}_{R=1_d}\right\}}{P_\beta(R=1_d)}\dot{P}_\beta(R=1_d)\E\left\{\psi_\theta(L)\mid R=1_d\right\}
+\int\psi_\theta(l)S_\beta(l\mid R=1_d)f_\beta(l, R=1_d)dl\nonumber\\
&=\E\left[\mathsf{1}_{R=1_d}\E\left\{\psi_\theta(L)\mid R=1_d\right\}\frac{\dot{P}_\beta(R=1_d)}{P_\beta(R=1_d)}\right]
+\E\left\{\mathsf{1}_{R=1_d}\psi_\theta(L)S_\beta(L\mid R=1_d)\right\}\nonumber\\
&=\E\left[\mathsf{1}_{R=1_d}\frac{\E\{\mathsf{1}_{R=1_d}\psi_\theta(L)\}}{P_\beta(R=1_d)}\frac{\dot{P}_\beta(R=1_d)}{P_\beta(R=1_d)}\right]
+\E\left[\mathsf{1}_{R=1_d}\left[\psi_\theta(L)-\E\left\{\psi_\theta(L)\mid R=1_d\right\}\right]S_\beta(L\mid R=1_d)\right]\nonumber,
\end{align}
since for any constant $C$,
\begin{align*}
\E\left\{\mathsf{1}_{R=1_d}CS_\beta(L\mid R=1_d)\right\}
=\E\left[\mathsf{1}_{R=1_d}C\E\left\{S_\beta(L\mid R=1_d)\right\}\right]
=0\ .
\end{align*}
Note that $\psi_\theta(l)-\E\{\psi_\theta(L)\mid R=1_d\}$ satisfies
\begin{align*}
\int\left[\psi_\theta(l)-\E\left\{\psi_\theta(L)\mid R=1_d\right\}\right]f_\beta(l\mid R=1_d)dl
=0\ .
\end{align*}
Now, we consider each term in \eqref{EIF2}. For each missing pattern $r\neq 1_d$, by the identification assumption  $f_\beta(\lms\mid\lob,R=r)
=\sum_{s\in\PAr}C^{s,r}(\lob)f_\beta(\lms\mid\lob,R=s)$,
the marginal density $f_\beta(l,r)$ has the recursive form
\begin{align*}
f_\beta(l,r)
&=P_\beta(R=r)f_\beta(\lob\mid R=r)f_\beta(\lms\mid\lob,R=r)\nonumber\\
&=\sum_{s\in\PAr}\frac{f_\beta(\lob,r)}{f_\beta(\lob,s)}C^{s,r}(\lob)f_\beta(l,s)
=\sum_{s\in\PAr}O^{s,r}(\lob)C^{s,r}(\lob)f_\beta(l,s)\\
&=\cdots
=P_\beta(R=1_d)f_\beta(l\mid R=1_d)\sum_{\Xi\in\Pi_r}\prod_{j=2}^{|\Xi|}O^{s_{j-1},s_j}(l^{s_j})C_{s_{j-1},s_j}(l^{s_j})\ .
\end{align*}
\textbf{Notation:} 
For a path $\Xi=\{1_d=s_1,\cdots,s_{|\Xi|}=r\}\in\Pi_r$, abbreviate $O^{s_{j-1},s_j}(l^{s_j})$ and $ C_{s_{j-1},s_j}(l^{s_j})$ as $O^{j-1,j}(l^{s_j})$ and $C_{j-1,j}(l^{s_j})$ respectively. Also denote the product $O^{j-1,j}(l^{s_j})C_{j-1,j}(l^{s_j})$ as $V_{j-1,j}(l^{s_j})$. With a little bit abuse of notation, define $V_{0,1}(l)=1$. When the mixture coefficients are type 2, abbreviate $O^{s_j}(l^{s_j})$ as $O^j(l^{s_j})$. 

Then, the derivative of $f_\beta(l,r)$ is
\begin{align}
\frac{\partial f_\beta(l,r)}{\partial\beta}
&=\left\{\frac{\dot{P_\beta}(R=1_d)}{P_\beta(R=1_d)}+S_\beta(l\mid R=1_d)\right\}f_\beta(l,r)\label{ctrb2}\\
&+P_\beta(R=1_d)f_\beta(l\mid R=1_d)\sum_{\Xi\in\Pi_r}\sum_{k=2}^{|\Xi|}\frac{\partial V_{k-1,k}(l^{s_k})/\partial\beta}{V_{k-1,k}(l^{s_k})}\prod_{j=2}^{|\Xi|}V_{j-1,j}(l^{s_j})\label{ctrb3}\ .
\end{align}
Similar to \eqref{ctrb1}, for each missing pattern $1_d\neq r\in\calR$, the first two terms of $\partial f_\beta(l,r)/\partial\beta$ contributes to $\E\{\psi_\theta(L)\partial\log f_\beta(L)/\partial\beta\}$ with
\begin{align*}
\frac{\dot{P_\beta}(R=1_d)}{P_\beta(R=1_d)}\int\psi_\theta(l)f_\beta(l,r)dl
&=\E\left[\mathsf{1}_{R=1_d}\frac{\E\{\mathsf{1}_{R=r}\psi_\theta(L)\}}{P_\beta(R=1_d)}\frac{\dot{P}_\beta(R=1_d)}{P_\beta(R=1_d)}\right]
\end{align*}
and
\begin{align*}
&\quads\int\psi_\theta(l)S_\beta(l\mid R=1_d)f_\beta(l,r)dl\\
&=\int\psi_\theta(l)\frac{f_\beta(l,r)}{f_\beta(l,1_d)}S_\beta(l\mid R=1_d)f_\beta(l,1_d)dl\nonumber\\
&=\E\left\{\mathsf{1}_{R=1_d}\psi_\theta(L)Q_r(L)S_\beta(L\mid R=1_d)\right\}\nonumber\\
&=\E\left[\mathsf{1}_{R=1_d}\left[\psi_\theta(L)Q_r(L)-\E\left\{\psi_\theta(L)Q_r(L)\mid R=1_d\right\}\right]S_\beta(L\mid R=1_d)\right]\nonumber,
\end{align*}
since for any constant $C$,
\begin{align*}
\E\left\{\mathsf{1}_{R=1_d}CS_\beta(L\mid R=1_d)\right\}
=\E\left[\mathsf{1}_{R=1_d}C\E\left\{S_\beta(L\mid R=1_d)\right\}\right]
=0\ .
\end{align*}
Use the fact that $Q_{1_d}(l)=1$ and 
\begin{align*}
\E\left\{\psi_\theta(L)Q_r(L)\mid R=1_d\right\}
=\int\psi_\theta(l)Q_r(l)\frac{f_\beta(l,1_d)}{P_\beta(R=1_d)}dl
=\frac{\E\{\mathsf{1}_{R=r}\psi_\theta(L)\}}{P_\beta(R=1_d)}\ .
\end{align*}
The components \eqref{ctrb1} and \eqref{ctrb2} collectively contribute to the influence function with term
\begin{align}\label{term1}
\mathsf{1}_{R=1_d}\sum_{r\in\calR}\frac{\E\{\mathsf{1}_{R=r}\psi_\theta(L)\}}{P_\beta(R=1_d)}
=\mathsf{1}_{R=1_d}\frac{\E\{\psi_\theta(L)\}}{P_\beta(R=1_d)}
\ ,
\end{align}
which is related to $\mathsf{1}_{R=1_d}\dot{P}_\beta(R=1_d)/P_\beta(R=1_d)$, and term
\begin{align}\label{term2}
\mathsf{1}_{R=1_d}\sum_{r\in\calR}\left[\psi_\theta(l)Q_r(l)-\frac{\E\{\mathsf{1}_{R=r}\psi_\theta(L)\}}{P_\beta(R=1_d)}\right]
=\mathsf{1}_{R=1_d}\frac{\psi_\theta(l)}{P(R=1_d\mid l)}
-\mathsf{1}_{R=1_d}\frac{\E\{\psi_\theta(L)\}}{P_\beta(R=1_d)}
\ ,
\end{align}
which is related to $\mathsf{1}_{R=1_d}S_\beta(l\mid R=1_d)$. It worth noting that $\E\{\psi_\theta(L)\}$ equals to 0 if one plugs in the true $\theta_0$. 

For the rest terms in \eqref{ctrb3}, it is natural to consider the contribution of each vertex on each path. Consider the path $\Xi=\{1_d=s_1,\cdots,s_{|\Xi|}=r\}$. For $2\le k\le|\Xi|$, the contribution to $\E\{\psi_\theta(L)\partial\log f_\beta(L)/\partial\beta\}$ that is related to $s_k$ is
\begin{align*}
&\quads\int\psi_\theta(l)\frac{\partial V_{k-1,k}(l^{s_k})/\partial\beta}{V_{k-1,k}(l^{s_k})}\prod_{j=2}^{|\Xi|}V_{j-1,j}(l^{s_j})f_\beta(l,1_d)dl\\
&=\int\int\psi_\theta(l)\prod_{j=1}^{k-1}V_{j-1,j}(l^{s_j})f_\beta(l^{\overline{s_k}},1_d\mid l^{s_k})dl^{\overline{s_k}}
\frac{\partial V_{k-1,k}(l^{s_k})/\partial\beta}{V_{k-1,k}(l^{s_k})}\prod_{j=k}^{|\Xi|}V_{j-1,j}(l^{s_j})f_\beta(l^{s_k})dl^{s_k}\\
&=\int m_{\Xi,k}(l^{s_k})\frac{\partial V_{k-1,k}(l^{s_k})/\partial\beta}{V_{k-1,k}(l^{s_k})}\prod_{j=k}^{|\Xi|}V_{j-1,j}(l^{s_j})f_\beta(l^{s_k})dl^{s_k}
\end{align*}
where
\begin{align*}
m_{\Xi,k}(l^{s_k})
&=\int\psi_\theta(l)\prod_{j=1}^{k-1}V_{j-1,j}(l^{s_j})f_\beta(l^{\overline{s_k}},1_d\mid l^{s_k})dl^{\overline{s_k}}
=\E\left\{\mathsf{1}_{R=1_d}\psi_\theta(L)\prod_{j=1}^{k-1}V_{j-1,j}(L^{s_j})\mid L^{s_k}=l^{s_k}\right\}\ .
\end{align*}
\textbf{Notations:} For any pattern $r$ and any pattern $s\in \mathrm{PA}_r$, define $S_\beta(\lob,r):=\partial\log f_\beta(\lob,r)/\partial\beta$, $S_\beta(\lob,s):=\partial\log f_\beta(\lob,s)/\partial\beta$ and $S_\beta(l^{s-r}\mid\lob,R=s):=\partial\log f_\beta(l^{s-r}\mid\lob,R=s)/\partial\beta$. Then,
\begin{align*}
S_\beta(\lob,r)
=S_\beta(\lob\mid R=r)+\frac{\dot{P_\beta}(R=r)}{P_\beta(R=r)},
\end{align*}
and,
\begin{align*}
S_\beta(l^{[s]},s)
=S_\beta(\lob,s)+S_\beta(l^{s-r}\mid\lob,R=s)\ .
\end{align*}
Also define $S_\beta(\lob,\PAr):=\partial\log f_\beta(\lob,\PAr)/\partial\beta$. Then,
\begin{align*}
S_\beta(\lob,\PAr)f_\beta(\lob,\PAr)
=\frac{\partial f_\beta(\lob,\PAr)}{\partial\beta}
=\sum_{s\in\PAr}\frac{\partial f_\beta(\lob,s)}{\partial\beta}
=\sum_{s\in\PAr}S_\beta(\lob,s)f_\beta(\lob,s)\ .
\end{align*}
\textbf{Consider the derivatives $\partial V_{k-1,k}(l^{s_k})/\partial\beta$ given three types of $C_{k-1,k}$.} 

\textbf{Type (1):} $C_{s_{k-1},s_k}(l^{s_k})=C_{s_{k-1},s_k}$. Then, 
\begin{align*}
\frac{\partial V_{k-1,k}(l^{s_k})}{\partial\beta}
=C_{s_{k-1},s_k}\frac{\partial O^{k-1,k}(l^{s_k})}{\partial\beta},
\end{align*}
and
\begin{align*}
\frac{\partial O^{k-1,k}(l^{s_k})}{\partial\beta}
&=\left[\frac{f_\beta(l^{s_k},s_k)}{f_\beta(l^{s_k},s_{k-1})}\right]'
=\frac{f_\beta'(l^{s_k},s_k)}{f_\beta(l^{s_k},s_{k-1})}-\frac{f_\beta(l^{s_k},s_k)f_\beta'(l^{s_k},s_{k-1})}{f^2_\beta(l^{s_k},s_{k-1})}\\
&=\frac{f_\beta'(l^{s_k},s_k)}{f_\beta(l^{s_k},s_k)}\frac{f_\beta(l^{s_k},s_k)}{f_\beta(l^{s_k},s_{k-1})}
-\frac{f_\beta'(l^{s_k},s_{k-1})}{f_\beta(l^{s_k},s_{k-1})}\frac{f_\beta(l^{s_k},s_k)}{f_\beta(l^{s_k},s_{k-1})}\\
&=O^{k-1,k}(l^{s_k})\frac{\partial\log f_\beta(l^{s_k},s_k)}{\partial\beta}
-O^{k-1,k}(l^{s_k})\frac{\partial\log f_\beta(l^{s_k},s_{k-1})}{\partial\beta}\ .
\end{align*}
Thus,
\begin{align*}
\frac{\partial V_{k-1,k}(l^{s_k})/\partial\beta}{V_{k-1,k}(l^{s_k})}
=S_\beta(l^{s_k},s_k)-S_\beta(l^{s_k},s_{k-1})\ .
\end{align*}
Therefore, the contribution is
\begin{align*}
&\quads\int m_{\Xi,k}(l^{s_k})\left\{S_\beta(l^{s_k},s_k)-S_\beta(l^{s_k},s_{k-1})\right\}\prod_{j=k}^{|\Xi|}V_{j-1,j}(l^{s_j})f_\beta(l^{s_k})dl^{s_k}\\
&=\int S_\beta(l^{s_k},s_k)m_{\Xi,k}(l^{s_k})C_{s_{k-1},s_k}\frac{P_\beta(R=s_k\mid l^{s_k})}{P_\beta(R=s_{k-1}\mid l^{s_k})}\prod_{j=k+1}^{|\Xi|}V_{j-1,j}(l^{s_j})f_\beta(l^{s_k})dl^{s_k}\\
&\quads-\int S_\beta(l^{s_k},s_{k-1})m_{\Xi,k}(l^{s_k})\frac{P_\beta(R=s_{k-1}\mid l^{s_k})}{P_\beta(R=s_{k-1}\mid l^{s_k})}V_{k-1,k}(l^{s_k})\prod_{j=k+1}^{|\Xi|}V_{j-1,j}(l^{s_j})f_\beta(l^{s_k})dl^{s_k}\\
&=\int S_\beta(l^{s_k},s_k)\frac{m_{\Xi,k}(l^{s_k})C_{s_{k-1},s_k}}{P_\beta(R=s_{k-1}\mid l^{s_k})}\prod_{j=k+1}^{|\Xi|}V_{j-1,j}(l^{s_j})f_\beta(l^{s_k},s_k)dl^{s_k}\\
&\quads-\int S_\beta(l^{s_k},s_{k-1})\frac{m_{\Xi,k}(l^{s_k})C_{s_{k-1},s_k}}{P_\beta(R=s_{k-1}\mid l^{s_k})}O^{k-1,k}(l^{s_k})\prod_{j=k+1}^{|\Xi|}V_{j-1,j}(l^{s_j})f_\beta(l^{s_k},s_{k-1})dl^{s_k}
\ .
\end{align*}
Define
\begin{align*}
\mu_{1,\Xi,k}(l^{s_k})=\frac{m_{\Xi,k}(l^{s_k})C_{s_{k-1},s_k}}{P_\beta(R=s_{k-1}\mid l^{s_k})}\prod_{j=k+1}^{|\Xi|}V_{j-1,j}(l^{s_j})\ .
\end{align*}
So, the above contribution can be written as 
\begin{align*}
&\quads\E\left\{\mathsf{1}_{R=s_k}\mu_{1,\Xi,k}(L^{s_k})S_\beta(L^{s_k},s_k)\right\}
-\E\left\{\mathsf{1}_{R=s_{k-1}}O^{k-1,k}(L^{s_k})\mu_{1,\Xi,k}(L^{s_k})S_\beta(L^{s_k},s_{k-1})\right\}\\
&=\E\left\{\mathsf{1}_{R=s_k}\mu_{1,\Xi,k}(L^{s_k})S_\beta(L^{s_k},s_k)\right\}
-\E\left\{\mathsf{1}_{R=s_{k-1}}O^{k-1,k}(L^{s_k})\mu_{1,\Xi,k}(L^{s_k})S_\beta(L^{s_{k-1}},s_{k-1})\right\}
\end{align*}
since 
\begin{align*}
S_\beta(l^{s_{k-1}},s_{k-1})
=S_\beta(l^{s_k},s_{k-1})+S_\beta(l^{s_{k-1}-s_k}\mid l^{s_k},R=s_{k-1}),
\end{align*}
and, for any function $g(l^{s_k})$,
\begin{align*}
&\quads\E\{\mathsf{1}_{R=s_{k-1}}g(L^{s_k})S_\beta(L^{s_{k-1}-s_k}\mid L^{s_k},R=s_{k-1})\}\\
&=\E[\mathsf{1}_{R=s_{k-1}}g(L^{s_k})\E\{S_\beta(L^{s_{k-1}-s_k}\mid L^{s_k},R=s_{k-1})\}]
=0\ .
\end{align*}

The above contribution can be further decomposed as
\begin{align*}
&\E\left\{\mathsf{1}_{R=s_k}\mu_{1,\Xi,k}(L^{s_k})S_\beta(L^{s_k}\mid R=s_k)\right\}
+\E\left\{\mathsf{1}_{R=s_k}\mu_{1,\Xi,k}(L^{s_k})\frac{\dot{P}_\beta(R=s_k)}{P_\beta(R=s_k)}\right\}\\
&-\E\left\{\mathsf{1}_{R=s_{k-1}}O^{k-1,k}(L^{s_k})\mu_{1,\Xi,k}(L^{s_k})S_\beta(L^{s_{k-1}}\mid R=s_{k-1})\right\}\\
&-\E\left\{\mathsf{1}_{R=s_{k-1}}O^{k-1,k}(L^{s_k})\mu_{1,\Xi,k}(L^{s_k})\frac{\dot{P}_\beta(R=s_{k-1})}{P_\beta(R=s_{k-1})}\right\},
\end{align*}
since
\begin{align*}
S_\beta(\lob,r)=S_\beta(\lob\mid R=r)+\frac{\dot{P_\beta}(R=r)}{P_\beta(R=r)}\ .
\end{align*}

For any constant $C$, 
\begin{align*}
\E\{\mathsf{1}_{R=r}CS_\beta(L\mid R=r)\}=\E[\mathsf{1}_{R=r}C\E\{S_\beta(L\mid R=r)\}]=0\ .
\end{align*}
Besides,
\begin{align*}
\E\left\{\mathsf{1}_{R=s_k}\mu_{1,\Xi,k}(L^{s_k})\frac{\dot{P}_\beta(R=s_k)}{P_\beta(R=s_k)}\right\}
=\E\left[\mathsf{1}_{R=s_k}\E\left\{\mu_{1,\Xi,k}(L^{s_k})\mid R=s_k\right\}\frac{\dot{P}_\beta(R=s_k)}{P_\beta(R=s_k)}\right],
\end{align*}
and
\begin{align*}
&\quads\E\left\{\mathsf{1}_{R=s_{k-1}}O^{k-1,k}(L^{s_k})\mu_{1,\Xi,k}(L^{s_k})\frac{\dot{P}_\beta(R=s_{k-1})}{P_\beta(R=s_{k-1})}\right\}\\
&=\E\left[\mathsf{1}_{R=s_{k-1}}\E\left\{O^{k-1,k}(L^{s_k})\mu_{1,\Xi,k}(L^{s_k})\right\}\frac{\dot{P}_\beta(R=s_{k-1})}{P_\beta(R=s_{k-1})}\right]\ .
\end{align*}
Therefore, the above four terms can be further simplified, which concludes that the contribution of $s_k$ to the influence function are:
\begin{align*}
\mathsf{1}_{R=s_k}\left[\mu_{1,\Xi,k}(l^{s_k})-\E\left\{\mu_{1,\Xi,k}(L^{s_k})\mid R=s_k\right\}\right],
\end{align*}
which is related to $\mathsf{1}_{R=s_k}S_\beta(l^{s_k}\mid R=s_k)$;
\begin{align*}
\mathsf{1}_{R=s_k}\E\left\{\mu_{1,\Xi,k}(L^{s_k})\mid R=s_k\right\},
\end{align*}
which is related to $\mathsf{1}_{R=s_k}\dot{P}_\beta(R=s_k)/P_\beta(R=s_k)$;
\begin{align*}
-\mathsf{1}_{R=s_{k-1}}\left[O^{k-1,k}(l^{s_k})\mu_{1,\Xi,k}(l^{s_k})-\E\left\{O^{k-1,k}(L^{s_k})\mu_{1,\Xi,k}(L^{s_k})\mid R=s_{k-1}\right\}\right],
\end{align*}
which is related to $\mathsf{1}_{R=s_{k-1}}S_\beta(l^{s_{k-1}}\mid R=s_{k-1})$;
\begin{align*}
-\mathsf{1}_{R=s_{k-1}}\E\left\{O^{k-1,k}(L^{s_k})\mu_{1,\Xi,k}(L^{s_k})\mid R=s_{k-1}\right\},
\end{align*}
which is related to $\mathsf{1}_{R=s_{k-1}}\dot{P}_\beta(R=s_{k-1})/P_\beta(R=s_{k-1})$. 

\textbf{It worth noting that}
\begin{align*}
\E\left\{\mu_{1,\Xi,k}(L^{s_k})\mid R=s_k\right\}
=\frac{\E\left\{\mathsf{1}_{R=1_d}\psi_\theta(L)\prod_{j=2}^{|\Xi|}V_{j-1,j}(L^{s_j})\right\}}{P(R=s_k)}\ ,
\end{align*}
and
\begin{align*}
\E\left\{O^{k-1,k}(L^{s_k})\mu_{1,\Xi,k}(L^{s_k})\mid R=s_{k-1}\right\}
=\frac{\E\left\{\mathsf{1}_{R=1_d}\psi_\theta(L)\prod_{j=2}^{|\Xi|}V_{j-1,j}(L^{s_j})\right\}}{P(R=s_{k-1})}\ .
\end{align*}
\textbf{Type (2):} $C_{s_{k-1},s_k}(l^{s_k})=P(R=s_{k-1}\mid l^{s_k})/P(R\in\PAsk\mid l^{s_k})$. Then, 
\begin{align*}
V_{k-1,k}(l^{s_k})
=\frac{P(R=s_k\mid l^{s_k})}{P(R\in\PAsk\mid l^{s_k})}
=O^k(l^{s_k}),
\end{align*}
and,
\begin{align*}
\frac{\partial V_{k-1,k}(l^{s_k})}{\partial\beta}
&=\left[\frac{f_\beta(l^{s_k},s_k)}{f_\beta(l^{s_k},\PAsk)}\right]'
=\frac{f_\beta'(l^{s_k},s_k)}{f_\beta(l^{s_k},\PAsk)}-\frac{f_\beta(l^{s_k},s_k)f_\beta'(l^{s_k},\PAsk)}{f^2_\beta(l^{s_k},\PAsk)}\\
&=\frac{f_\beta'(l^{s_k},s_k)}{f_\beta(l^{s_k},s_k)}\frac{f_\beta(l^{s_k},s_k)}{f_\beta(l^{s_k},\PAsk)}
-\frac{f_\beta'(l^{s_k},\PAsk)}{f_\beta(l^{s_k},\PAsk)}\frac{f_\beta(l^{s_k},s_k)}{f_\beta(l^{s_k},\PAsk)}\\
&=V_{k-1,k}(l^{s_k})\frac{\partial\log f_\beta(l^{s_k},s_k)}{\partial\beta}
-V_{k-1,k}(l^{s_k})\frac{\partial\log f_\beta(l^{s_k},\PAsk)}{\partial\beta}\ .
\end{align*}
Thus, 
\begin{align*}
\frac{\partial V_{k-1,k}(l^{s_k})/\partial\beta}{V_{k-1,k}(l^{s_k})}
=S_\beta(l^{s_k},s_k)-S_\beta(l^{s_k},\PAsk)\ .
\end{align*}
Therefore, the contribution is
\begin{align*}
&\quads\int m_{\Xi,k}(l^{s_k})\left\{S_\beta(l^{s_k},s_k)-S_\beta(l^{s_k},\PAsk)\right\}\prod_{j=k}^{|\Xi|}V_{j-1,j}(l^{s_j})f_\beta(l^{s_k})dl^{s_k}\\
&=\int S_\beta(l^{s_k},s_k)m_{\Xi,k}(l^{s_k})\frac{P_\beta(R=s_k\mid l^{s_k})}{P_\beta(R\in\PAsk\mid l^{s_k})}\prod_{j=k+1}^{|\Xi|}V_{j-1,j}(l^{s_j})f_\beta(l^{s_k})dl^{s_k}\\
&\quads-\int S_\beta(l^{s_k},\PAsk)m_{\Xi,k}(l^{s_k})\frac{P_\beta(R\in\PAsk\mid l^{s_k})}{P_\beta(R\in\PAsk\mid l^{s_k})}V_{k-1,k}(l^{s_k})\prod_{j=k+1}^{|\Xi|}V_{j-1,j}(l^{s_j})f_\beta(l^{s_k})dl^{s_k}\\
&=\int S_\beta(l^{s_k},s_k)\frac{m_{\Xi,k}(l^{s_k})}{P_\beta(R\in\PAsk\mid l^{s_k})}\prod_{j=k+1}^{|\Xi|}V_{j-1,j}(l^{s_j})f_\beta(l^{s_k},s_k)dl^{s_k}\\
&\quads-\sum_{t\in\PAsk}\int S_\beta(l^{s_k},t)\frac{m_{\Xi,k}(l^{s_k})}{P_\beta(R\in\PAsk\mid l^{s_k})}V_{k-1,k}(l^{s_k})\prod_{j=k+1}^{|\Xi|}V_{j-1,j}(l^{s_j})f_\beta(l^{s_k},t)dl^{s_k},
\end{align*}
since
\begin{align*}
&\quads S_\beta(l^{s_k},\PAsk)P_\beta(R\in\PAsk\mid l^{s_k})f_\beta(l^{s_k})\\
&=S_\beta(l^{s_k},\PAsk)f_\beta(l^{s_k},\PAsk)
=\sum_{s\in\PAr}S_\beta(\lob,s)f_\beta(\lob,s)\ .
\end{align*}
Define
\begin{align*}
\mu_{2,\Xi,k}(l^{s_k})
&=\frac{m_{\Xi,k}(l^{s_k})}{P_\beta(R\in\PAsk\mid l^{s_k})}\prod_{j=k+1}^{|\Xi|}V_{j-1,j}(l^{s_j})
\ .
\end{align*}

Similarly, the above contribution can be written as 
\begin{align*}
\E\left\{\mathsf{1}_{R=s_k}\mu_{2,\Xi,k}(L^{s_k})S_\beta(L^{s_k},s_k)\right\}
-\sum_{t\in\PAsk}\E\left\{\mathsf{1}_{R=t}O^k(L^{s_k})\mu_{2,\Xi,k}(L^{s_k})S_\beta(L^t,t)\right\}\ ,
\end{align*}
which includes the following terms:
\begin{align*}
\mathsf{1}_{R=s_k}\left[\mu_{2,\Xi,k}(l^{s_k})-\E\left\{\mu_{2,\Xi,k}(L^{s_k})\mid R=s_k\right\}\right],
\end{align*}
which is related to $\mathsf{1}_{R=s_k}S_\beta(l^{s_k}\mid R=s_k)$;
\begin{align*}
\mathsf{1}_{R=s_k}\E\left\{\mu_{2,\Xi,k}(L^{s_k})\mid R=s_k\right\},
\end{align*}
which is related to $\mathsf{1}_{R=s_k}\dot{P}_\beta(R=s_k)/P_\beta(R=s_k)$;
\begin{align*}
-\mathsf{1}_{R=t}\left[O^k(l^{s_k})\mu_{2,\Xi,k}(l^{s_k})-\E\left\{O^k(L^{s_k})\mu_{2,\Xi,k}(L^{s_k})\mid R=t\right\}\right],
\end{align*}
which is related to $\mathsf{1}_{R=t}S_\beta(l^t\mid R=t)$ for each $t\in\PAsk$;
\begin{align*}
-\mathsf{1}_{R=t}\E\left\{O^k(L^{s_k})\mu_{2,\Xi,k}(L^{s_k})\mid R=t\right\},
\end{align*}
which is related to $\mathsf{1}_{R=t}\dot{P}_\beta(R=t)/P_\beta(R=t)$ for each $t\in\PAsk$. 

\textbf{It worth noting that}
\begin{align*}
\E\left\{\mu_{2,\Xi,k}(L^{s_k})\mid R=s_k\right\}
=\frac{\E\left\{\mathsf{1}_{R=1_d}\psi_\theta(L)\prod_{j=2}^{|\Xi|}V_{j-1,j}(L^{s_j})\right\}}{P(R=s_k)}\ ,
\end{align*}
and
\begin{align*}
\E\left\{O^k(L^{s_k})\mu_{2,\Xi,k}(L^{s_k})\mid R=t\right\}
=\frac{\E\left\{\mathsf{1}_{R=1_d}\psi_\theta(L)
\frac{P(R=t\mid l^{s_k})}{P(R\in\PAsk\mid l^{s_k})}\prod_{j=2}^{|\Xi|}V_{j-1,j}(L^{s_j})\right\}}{P(R=t)}\ .
\end{align*}
\textbf{Type (3):} $C_{s_{k-1},s_k}(l^{s_k})=P_\beta(R=s_{k-1})/P_\beta(R\in\PAsk)$. Then, 
\begin{align*}
V_{k-1,k}(l^{s_k})
=\frac{P_\beta(R=s_k\mid l^{s_k})}{P_\beta(R=s_{k-1}\mid l^{s_k})}\frac{P_\beta(R=s_{k-1})}{P_\beta(R\in\PAsk)}
=\frac{f_\beta(l^{s_k}\mid R=s_k)P_\beta(R=s_k)}{f_\beta(l^{s_k}\mid R=s_{k-1})P_\beta(R\in\PAsk)}
\end{align*}
and
\begin{align*}
&\quads\frac{\partial V_{k-1,k}(l^{s_k})}{\partial\beta}\\
&=\left[\frac{f_\beta(l^{s_k}\mid R=s_k)}{f_\beta(l^{s_k}\mid R=s_{k-1})}\right]'\frac{P_\beta(R=s_k)}{P_\beta(R\in\PAsk)}
+\frac{f_\beta(l^{s_k}\mid R=s_k)}{f_\beta(l^{s_k}\mid R=s_{k-1})}\left[\frac{P_\beta(R=s_k)}{P_\beta(R\in\PAsk)}\right]'\\
&=\left[\frac{f_\beta'(l^{s_k}\mid R=s_k)}{f_\beta(l^{s_k}\mid R=s_{k-1})}-\frac{f_\beta(l^{s_k}\mid R=s_k)f_\beta'(l^{s_k}\mid R=s_{k-1})}{f^2_\beta(l^{s_k}\mid R=s_{k-1})}\right]'\frac{P_\beta(R=s_k)}{P_\beta(R\in\PAsk)}\\
&\quads+\frac{f_\beta(l^{s_k}\mid R=s_k)}{f_\beta(l^{s_k}\mid R=s_{k-1})}\left[\frac{\dot{P}_\beta(R=s_k)}{P_\beta(R\in\PAsk)}-\frac{P_\beta(R=s_k)\dot{P}_\beta(R\in\PAsk)}{P^2_\beta(R\in\PAsk)}\right]\\
&=\left\{S_\beta(l^{s_k}\mid R=s_k)-S_\beta(l^{s_k}\mid R=s_{k-1})+\frac{\dot{P}_\beta(R=s_k)}{P_\beta(R=s_k)}-\frac{\dot{P}_\beta(R\in\PAsk)}{P_\beta(R\in\PAsk)}\right\}V_{k-1,k}(l^{s_k})
\ .
\end{align*}
Thus, the contribution is
\begin{align*}
&\int S_\beta(l^{s_k}\mid R=s_k)m_{\Xi,k}(l^{s_k})\frac{P_\beta(R=s_k\mid l^{s_k})}{P_\beta(R=s_{k-1}\mid l^{s_k})}\frac{P_\beta(R=s_{k-1})}{P_\beta(R\in\PAsk)}\prod_{j=k+1}^{|\Xi|}V_{j-1,j}(l^{s_j})f_\beta(l^{s_k})dl^{s_k}\\
-&\int S_\beta(l^{s_k}\mid R=s_{k-1})m_{\Xi,k}(l^{s_k})\frac{P_\beta(R=s_{k-1}\mid l^{s_k})}{P_\beta(R=s_{k-1}\mid l^{s_k})}V_{k-1,k}(l^{s_j})\prod_{j=k+1}^{|\Xi|}V_{j-1,j}(l^{s_j})f_\beta(l^{s_k})dl^{s_k}\\
+&\left\{\frac{\dot{P}_\beta(R=s_k)}{P_\beta(R=s_k)}-\frac{\dot{P}_\beta(R\in\PAsk)}{P_\beta(R\in\PAsk)}\right\}
\int m_{\Xi,k}(l^{s_k})\prod_{j=k}^{|\Xi|}V_{j-1,j}(l^{s_j})f_\beta(l^{s_k})dl^{s_k}
\ .
\end{align*}
Define
\begin{align*}   
\mu_{3,\Xi,k}(l^{s_k})
&=\frac{m_{\Xi,k}(l^{s_k})}{P_\beta(R=s_{k-1}\mid l^{s_k})}\frac{P_\beta(R=s_{k-1})}{P_\beta(R\in\PAsk)}\prod_{j=k+1}^{|\Xi|}V_{j-1,j}(l^{s_j}),\\
c_{\Xi,k} 
&=\int m_{\Xi,k}(l^{s_k})\prod_{j=k}^{|\Xi|}V_{j-1,j}(l^{s_j})f_\beta(l^{s_k})dl^{s_k}
=\E\left\{\mathsf{1}_{R=1_d}\psi_\theta(L)\prod_{j=2}^{|\Xi|}V_{j-1,j}(L^{s_j})\right\}
\ .
\end{align*}

Similarly, the above contribution can be written as 
\begin{align*}
&\E\left\{\mathsf{1}_{R=s_k}\mu_{3,\Xi,k}(L^{s_k})S_\beta(L^{s_k}\mid R=s_k)\right\}-\E\left\{\mathsf{1}_{R=s_{k-1}}O^{k-1,k}(L^{s_k})\mu_{3,\Xi,k}(L^{s_k})S_\beta(L^{s_k}\mid R=s_{k-1})\right\}\\
&+c_{\Xi,k}\frac{\E\left\{\mathsf{1}_{R=s_k}\right\}}{P_\beta(R=s_k)}\frac{\dot{P}_\beta(R=s_k)}{P_\beta(R=s_k)}
-\frac{c_{\Xi,k}}{P_\beta(R\in\PAsk)}\sum_{t\in\PAsk}\frac{\E\left\{\mathsf{1}_{R=t}\right\}}{P_\beta(R=t)}\dot{P}_\beta(R=t)\ ,
\end{align*}
which includes the following terms:
\begin{align*}
\mathsf{1}_{R=s_k}\left[\mu_{3,\Xi,k}(l^{s_k})-\E\left\{\mu_{3,\Xi,k}(L^{s_k})\mid R=s_k\right\}\right],
\end{align*}
which is related to $\mathsf{1}_{R=s_k}S_\beta(l^{s_k}\mid R=s_k)$;
\begin{align*}
\mathsf{1}_{R=s_k}\frac{c_{\Xi,k}}{P_\beta(R=s_k)},
\end{align*}
which is related to $\mathsf{1}_{R=s_k}\dot{P}_\beta(R=s_k)/P_\beta(R=s_k)$;
\begin{align*}
-\mathsf{1}_{R=s_{k-1}}\left[O^{k-1,k}(l^{s_k})\mu_{3,\Xi,k}(l^{s_k})-\E\left\{O^{k-1,k}(L^{s_k})\mu_{3,\Xi,k}(L^{s_k})\mid R=s_{k-1}\right\}\right],
\end{align*}
which is related to $\mathsf{1}_{R=s_{k-1}}S_\beta(l^{s_{k-1}}\mid R=s_{k-1})$;
\begin{align*}
-\mathsf{1}_{R=t}\frac{c_{\Xi,k}}{P_\beta(R\in\PAsk)},
\end{align*}
which is related to $\mathsf{1}_{R=t}\dot{P}_\beta(R=t)/P_\beta(R=t)$ for each $t\in\PAsk$. 

\textbf{It worth noting that}
\begin{align*}
\E\left\{\mu_{3,\Xi,k}(L^{s_k})\mid R=s_k\right\}
=\frac{c_{\Xi,k}}{P(R=s_k)}\ , 
\end{align*}
and
\begin{align*}
\E\left\{O^{k-1,k}(L^{s_k})\mu_{3,\Xi,k}(L^{s_k})\mid R=s_{k-1}\right\}
=\frac{c_{\Xi,k}}{P(R=s_{k-1})}\ .
\end{align*}
\textbf{Summary:} 
Notice that for all $j=1,2,3$, $\mu_{j,\Xi,k}(l^{s_k})$ has the uniform expression:
\begin{align*}
\mu_{\Xi,k}(l^{s_k})
=\frac{m_{\Xi,k}(l^{s_k})}{P_\beta(R=s_{k-1}\mid l^{s_k})}C_{s_{k-1},s_k}(l^{s_k})\prod_{j=k+1}^{|\Xi|}V_{j-1,j}(l^{s_j})\ .
\end{align*}
By the summation of those terms in each case, the contribution of $s_k$ on the path $\Xi$ to $\E\{\psi_\theta(L)\partial\log f_\beta(L)/\partial\beta\}$ can be written as $\E\{F_\theta^{\Xi,k}(L,R)S_\beta(L,R)\}$, where $F_\theta^{\Xi,k}(L,R)$
\begin{align*}
=
\begin{cases}
\mathsf{1}_{R=s_k}\mu_{\Xi,k}(L^{s_k})
-\mathsf{1}_{R=s_{k-1}}O^{k-1,k}(L^{s_k})\mu_{\Xi,k},
& \textit{Type(1) } , \\
\mathsf{1}_{R=s_k}\mu_{\Xi,k}(L^{s_k})
-\mathsf{1}_{R\in\PAsk}O^k(L^{s_k})\mu_{\Xi,k},
& \textit{Type(2) } , \\
\mathsf{1}_{R=s_k}\mu_{\Xi,k}(L^{s_k})
-\mathsf{1}_{R=s_{k-1}}O^{k-1,k}(L^{s_k})\mu_{\Xi,k}
+\mathsf{1}_{R=s_{k-1}}\frac{c_{\Xi,k}}{P(R=s_{k-1})}
-\mathsf{1}_{R\in\PAsk}\frac{c_{\Xi,k}}{P_\beta(R\in\PAsk)}
& \textit{Type(3) } \ .
\end{cases}  
\end{align*}
Notice that $\E\{\psi_{\theta_0}(L)\}=0$. Recall that we denote the derivative $\partial\E\{\psi_\theta(L)\}/\partial\theta$ at $\theta_0$ as $D_{\theta_0}$. Let the influence function $\zeta(L,R)=
-D_{\theta_0}^{-1}F_{\theta_0}(L,R)$ where
\begin{align*}
F_\theta(L,R)
=\mathsf{1}_{R=1_d}\left\{1+\sum_{\Xi\in\Pi}\prod_{s\in\Xi}O^s(L^{[s]})\right\}\psi_\theta(L)
+\sum_{1_d\neq r\in\calR}\sum_{\Xi\in\Pi_r}\sum_{k=2}^{|\Xi|}F_\theta^{\Xi,k}(L,R)
\ .
\end{align*}
Equation \eqref{path-differentiability} is satisfied for true parameters, since
\begin{align*}
&\quads\E\left\{\psi_\theta(L)\frac{\partial\log f_\beta(L)}{\partial\beta}\right\}
=\E\{F_\theta(L,R)S_\beta(L,R)\}\\
&=\E\left[\mathsf{1}_{R=1_d}\frac{\psi_\theta(L)}{P(R=1_d\mid L)}S_\beta(L,R)\right]
+\sum_{1_d\neq r\in\calR}\sum_{\Xi\in\Pi_r}\sum_{k=2}^{|\Xi|}
\E\left\{F_\theta^{\Xi,k}(L,R)S_\beta(L,R)\right\}
\ .
\end{align*}
\textbf{We also need to verify that $F_\theta$ has mean 0.} We have shown that the terms related to each $\mathsf{1}_{R=s}S_\beta(l^{[s]}\mid R=s)$ have mean 0. We need to show that the summation of all terms related to $\mathsf{1}_{R=s}\dot{P}_\beta(R=s)/P_\beta(R=s)$ over all $s\in\calR$ has mean 0, which is similar to the property $\sum_{s\in\calR}\E\{\mathsf{1}_{R=s}\}\dot{P}_\beta(R=s)/P_\beta(R=s)=0$.

Notice that the contribution of $s_k$ on the path $\Xi$ includes positive terms related to $\mathsf{1}_{R=s_k}\dot{P}_\beta(R=s_k)/P_\beta(R=s_k)$ and negative terms related to a specific parent $s_{k-1}$ or all parents $\PAsk$, i.e., $\mathsf{1}_{R=s_{k-1}}\dot{P}_\beta(R=s_{k-1})/P_\beta(R=s_{k-1})$ or $\mathsf{1}_{R=t}\dot{P}_\beta(R=t)/P_\beta(R=t)$ for each $t\in\PAsk$. It is tedious to first calculate the summation of those terms related to a specific pattern $s$, and then calculate the summation over all patterns $s\in\calR$. Instead, it suffices to show that for each $s_k$ on the path $\Xi$, those terms cancel out. More precisely,

\textbf{Type (1):}
\begin{align*}
&\quads\E\left[\mathsf{1}_{R=s_k}\E\left\{\mu_{1,\Xi,k}(L^{s_k})\mid R=s_k\right\}\right]
-\E\left[\mathsf{1}_{R=s_{k-1}}\E\left\{O^{k-1,k}(L^{s_k})\mu_{1,\Xi,k}(L^{s_k})\mid R=s_{k-1}\right\}\right]\\
&=P(R=s_k)\frac{\E\left\{\mathsf{1}_{R=1_d}\psi_\theta(L)\prod_{j=2}^{|\Xi|}V_{j-1,j}(L^{s_j})\right\}}{P(R=s_k)}
-P(R=s_{k-1})\frac{\E\left\{\mathsf{1}_{R=1_d}\psi_\theta(L)\prod_{j=2}^{|\Xi|}V_{j-1,j}(L^{s_j})\right\}}{P(R=s_{k-1})}\\
&=0\ .
\end{align*}

\textbf{Type (2):}
\begin{align*}
&\quads\E\left[\mathsf{1}_{R=s_k}\E\left\{\mu_{2,\Xi,k}(L^{s_k})\mid R=s_k\right\}\right]
-\sum_{t\in\PAsk}\E\left[\mathsf{1}_{R=t}\E\left\{O^k(L^{s_k})\mu_{2,\Xi,k}(L^{s_k})\mid R=t\right\}\right]\\
&=P(R=s_k)\frac{\E\left\{\mathsf{1}_{R=1_d}\psi_\theta(L)\prod_{j=2}^{|\Xi|}V_{j-1,j}(L^{s_j})\right\}}{P(R=s_k)}\\
&\quads-\sum_{t\in\PAsk}P(R=t)
\frac{\E\left\{\mathsf{1}_{R=1_d}\psi_\theta(L)C_{t,s_k}(L^{s_k})\prod_{j=2}^{|\Xi|}V_{j-1,j}(L^{s_j})\right\}}{P(R=t)}\\
&=\E\left\{\mathsf{1}_{R=1_d}\psi_\theta(L)\prod_{j=2}^{|\Xi|}V_{j-1,j}(L^{s_j})\right\}
-\E\left[\mathsf{1}_{R=1_d}\psi_\theta(L)\left\{\sum_{t\in\PAsk}C_{t,s_k}(L^{s_k})\right\}\prod_{j=2}^{|\Xi|}V_{j-1,j}(L^{s_j})\right]\\
&=0\ .
\end{align*}

\textbf{Type (3):}
\begin{align*}
\E\left\{\mathsf{1}_{R=s_k}\frac{c_{\Xi,k}}{P_\beta(R=s_k)}\right\}
-\sum_{t\in\PAsk}\E\left\{\mathsf{1}_{R=t}\frac{c_{\Xi,k}}{P_\beta(R\in\PAsk)}\right\}
=0\ .
\end{align*}
Therefore, $F_\theta$ has mean 0. 

The tangent set $\calT$ is defined as the mean square closure of all $q$-dimensional linear combinations of scores $S_\beta$ for smooth parametric submodels as in \ref{parametric-score}. That is,
\begin{align*}
\calT=\left\{h(L,R)\in\bbR^q:\E\{\|h\|^2\}\le\infty,
\exists A_jS_{\beta j}\textrm{ with }\lim_{j\to\infty}\E\{\|h-A_jS_{\beta j}\|\}=0
\right\}
\end{align*}
where $A_j$ is a constant matrix with $q$ rows. It can be verified by similar arguments as in \cite{newey1990semiparametric}. We have shown that $F_\theta$ is a linear combination of the components of score \eqref{parametric-score}. It is easy to see that $\zeta$ belongs to the tangent space $\calT$. 

Therefore, $\theta(\beta)$ is pathwise differentiable. All the conditions of Theorem 3.1 in \cite{newey1990semiparametric} hold, so the efficiency bound for regular estimators of the parameter $\theta$ is given by $D_{\theta_0}^{-1}V_{\theta_0}D_{\theta_0}^{-1\tr}$ where $V_{\theta_0}=\E\{F_{\theta_0}(L,R)F_{\theta_0}(L,R)\tr\}$.

\end{proof}

\section{Additional assumptions for asymptotic properties}\label{sec:additional-assump}
We require the following set of assumptions to establish the asymptotic theory. Under mild conditions, we show that $O^r(\lob;\hatalphar)$ is consistent, $\hat{\bbP}_N\psi_\theta$ is asymptotically normal for each $\theta$ in a compact set $\Theta\subset\bbR^q$, and $\hat{\theta}_N$ is consistent and efficient. We require the following set of assumptions to establish the consistency of the proposed estimator.

\begin{assumption}\label{assump2} 
The following conditions hold for each missing pattern $r\in\calR$:
\begin{enumerate}[label={\textbf{~\Alph*:}}, ref={Assumption~\theassumption.\Alph*},leftmargin=1cm]
\item\label{assump-2A}
There exist constants $0<c_0<C_0$ such that $c_0\le O^r(\lob)\le C_0$ for all $\lob\in\domr$.

\item\label{assump-2B}
The optimization ${\min_{\alphar}} \E[\calLr\{O^r(\Lob;\alphar),R\}]$ using the sequential loss with known $Q^{\PAr}(l)$, 
\begin{align*}
	\mathsf{1}_{R=1_d}O^r(\lob;\alphar)Q^{\PAr}(l)
    -\mathsf{1}_{R=r}\log O^r(\lob;\alphar),
\end{align*}
has a unique solution $\alphar_0\in\bbR^{K_r}$.

\item\label{assump-2C}
The total number of basis functions $K_r$ grows as the sample size increases and satisfies $K_r^2=o(N_r)$ where $N_r$ is the number of observations in patterns $r$ and $1_d$. 

\item\label{assump-2D}
There exist constants $C_1>0$ and $\mu_1>1/2$ such that for any positive integer $K_r$, there exist $\alpha_{K_r}^*\in\bbR^{K_r}$ satisfying 
$$
\underset{\lob\in\domr}{\sup}\big|O^r(\lob)-O^r(\lob;\alpha_{K_r}^*)\big|\le C_1K_r^{-\mu_1}.
$$

\item\label{assump-2E}
The Euclidean norm of the basis functions satisfies $\sup_{\lob\in\domr}\|\Phir(\lob)\|_2=O(K_r^{1/2})$. 

\item\label{assump-2F}
Let $\lambda_1,\ldots,\lambda_{K_r}$ be the eigenvalues of $\E\{\Phir(\Lob)\Phir(\Lob)\tr\}$ in the non-decreasing order. There exist constants $\lambda_{\min}^*$ and $\lambda_{\max}^*$ such that $0<\lambda_{\min}^*\le\lambda_1\le\lambda_{K_r}\le\lambda_{\max}^*$.

\item\label{assump-2G}
The tuning parameter $\lambda$ satisfies $\lambda=o(1/\sqrt{K_rN_r})$.
\end{enumerate}
\end{assumption}
\ref{assump-2A} is the boundedness assumption commonly used in missing data and causal inference. It is equivalent to that $P(R=r\mid\lob,R\in\{\PAr,r\})$ is strictly bounded away from 0 and 1. The domain $\domr$ is usually assumed to be compact, so it becomes possible to approximate $O^r$ with compactly supported functions. \ref{assump-2B} is a standard condition for consistency of minimum loss estimators of $\alphar_0$. It is well known that the uniform approximation error is related to the number of basis functions. Thus, we allow $K_r$ to increase with sample size under certain restrictions in \ref{assump-2C}. The uniform approximation rate $\mu_1$ in \ref{assump-2D} is related to the true propensity odds $O^r$ and the choice of basis functions. For instance, the rate $\mu_1=s/d$ for power series and splines if $O^r$ is continuously differentiable of order $s$ on $[0,1]^d$ under mild assumptions; see \citet{newey1997convergence} and \citet{fan2022optimal} for details. The restriction $\mu> 1/2$ is a technical condition such that the estimator of propensity odds is consistent. \ref{assump-2E} and \ref{assump-2F} are standard conditions for controlling the magnitude of the basis functions. The Euclidean norm of the basis function vector can increase as the spanned space extends, but its growth rate cannot be too fast. These assumptions are satisfied by many bases such as the regression spline, trigonometric polynomial, and wavelet bases; see, e.g., \citet{newey1997convergence, horowitz2004nonparametric, chen2007large} and \citet{fan2022optimal}. \ref{assump-2G} is a technical assumption of the tuning parameter $\lambda$ for the maintenance of consistency of weights. We now establish the consistency of the estimated odds.
\begin{theorem}\label{odds}
Under Assumptions \ref{assump1} and \ref{assump2}, for each missing pattern $r$, we have
\begin{align*}
\left\|O^r(\ \cdot\ ;\hatalphar)-O^r\right\|_\infty
&=O_p\left(\sqrt{\frac{K_r^2}{N_r}}+K_r^{\frac{1}{2}-\mu_1}\right)=o_p(1)\ ,\\
\left\|O^r(\ \cdot\ ;\hatalphar)-O^r\right\|_{P,2}
&=O_p\left(\sqrt{\frac{K_r}{N_r}}+K_r^{-\mu_1}\right)=o_p(1)
\end{align*}
where $\|X\|_{P,2}^2=\int X^2dP$ is the second moment of a random variable.
\end{theorem}

Next, we establish the asymptotic normality of the empirical weighted estimating function $\hat{\bbP}_N\psi_\theta$ for each $\theta$. Let $u_\theta^r(\lob)$ be the conditional expectation of the estimating function given variables $\Lob$, \emph{i.e.} $\E\{\psi_\theta(L)\mid\Lob=\lob,R=r\}$, which is equal to $\E\{\psi_\theta(L)\mid\Lob=\lob,R=1_d\}$ under identifying assumptions \eqref{eqn:chen-propensity}. Note that the estimating function $\psi_\theta$ is a $q$-dimensional vector-valued function. We only need to consider each entry separately. Denote the $j$-th entry of $\psi_\theta$ and $u_\theta^r$ as $\psi_{\theta,j}$ and $u_{\theta,j}^r$ respectively. Let $n_{[ \ ]}\{\epsilon,\calF,\cdot\}$ denote the bracketing number of the set $\calF$ by $\epsilon$-brackets with respect to a specific norm. We will need the following additional conditions. 
\begin{assumption}\label{assump3} 
The following conditions hold for all missing pattern $r$ and all $\theta\in\Theta$ where $\Theta$ is a compact set:
\begin{enumerate}[label={\textbf{~\Alph*:}},ref={Assumption~\theassumption.\Alph*},leftmargin=1cm]

\item\label{assump-3A}
There exist constants $C_2>0$ and $\mu_2>1/2$ such that for any $\theta$ and each missing pattern $r$, there exists a parameter $\beta_\theta^r$ satisfying $\sup_{\lob\in\domr}|u_\theta^r(\lob)-\Phir(\lob)\tr\beta_\theta^r|\le C_2K_r^{-\mu_2}$. 

\item\label{assump-3B}
Each of the true propensity odds, $O^r$, is contained in a set of smooth functions $\calM^r$. There exists constants $C_\calM>0$ and $d_\calM>1/2$ such that $\log n_{[ \ ]}\{\epsilon,\calM^r,L^\infty\}\le C_\calM(1/\epsilon)^{1/d_\calM}$.

\item\label{assump-3C}
The sets $\Psi:=\{\psi_{\theta,j}:\theta\in\Theta,j=1,\ldots,q\}$ are contained in a function class $\calH$ such that there exists constants $C_\calH>0$ and $d_\calH>1/2$ such that $\log n_{[ \ ]}\{\epsilon,\calH,L_2(P)\}\le C_\calH(1/\epsilon)^{1/d_\calH}$.

\item\label{assump-3D}
There exists a constant $C_3$ such that for all $j=1,\ldots,q$, 
\begin{align*}
\E\left\{\psi_{\theta,j}(L)-u_{\theta,j}^r(\Lob)\right\}^2
\le \E\left[\underset{\theta}{\sup}\left\{\psi_{\theta,j}(L)-u_{\theta,j}^r(\Lob)\right\}^2\right]\le C_3^2\ .
\end{align*}

\item\label{assump-3E}
$N_r^{1/\{2(\mu_1+\mu_2)\}}=o(K_r)$, which means that the growth rate of the number of basis functions has a lower bound.
\end{enumerate}
\end{assumption}
\ref{assump-3A} is a requirement similar in spirit to \ref{assump-2D} such that the conditional expectation $u_\theta^r(\lob)$ can be well approximated as we extend the space spanned by the basis functions. \ref{assump-3B} and \ref{assump-3C} are conditions on the complexity of the function classes $\calM^r$ and $\calH$ to ensure uniform convergence over $\theta$. These assumptions are satisfied for many function classes. For instance, if $\calM^r$ is a Sobolev class of functions $f:[0,1]\mapsto\bbR$ such that $\|f\|_\infty\le1$ and the ($s-1$)-th derivative is absolutely continuous with $\int(f^{(s)})^2(x)dx\le1$ for some fixed $s\in\mathbb{N}$, then $\log n_{[ \ ]}\{\epsilon,\calM^r,L^\infty)\}\le C(1/\epsilon)^{1/s}$ by Example 19.10 of \citet{van2000asymptotic}. The condition $d_\calM>1/2$ is satisfied for all $s\ge1$. A H\"older class of functions also satisfies this condition \citep{fan2022optimal}. \ref{assump-3D} is a technical condition related to the envelope function such that we can apply the maximal inequality via bracketing. \ref{assump-3E} requires the number of basis functions to grow such that the approximation error decreases in general.

\begin{theorem}\label{psi}
Suppose that Assumptions \ref{assump1}, \ref{assump2} and \ref{assump3} hold. For any $\theta\in\Theta$,
\begin{align*}
\sqrt{N}\left[\hat{\bbP}_N\psi_\theta-\E\{\psi_\theta(L)\}\right]\overset{d}{\to}N(0,V_\theta)
\end{align*}
where $V_\theta$ is defined in Theorem \ref{efficiency-bound}.
\end{theorem}
To prove the theorem, we utilize a few lemmas of the bracketing number $n_{[ \ ]}\{\epsilon,\calF,\cdot\}$. The proofs of the theorem and lemmas are given separately in Appendices.

With further assumptions, we show the consistency and asymptotical normality of $\hat{\theta}_N$ that solves $\hat{\bbP}_N\psi_\theta=0$.

\begin{assumption}\label{assump4} 
The following conditions hold for all missing pattern $r$ and all $\theta\in\Theta$:
\begin{enumerate}[label={\textbf{~\Alph*:}},ref={Assumption~\theassumption.\Alph*},leftmargin=1cm]
\item\label{assump-4A}
For any sequence $\{\theta_n\}\in\Theta$ , $\E\{\underset{1\le j\le q}{\max}|\psi_{\theta_n,j}(L)|\}\to 0$ implies 
$$
\|\theta_n-\theta_0\|_2\to 0.
$$

\item\label{assump-4B}
For each $j$-th entry and any $\delta>0$, there exists an envelop function $f_{\delta,j}$ such that $|\psi_{\theta,j}(l)-\psi_{\theta_0,j}(l)|\le f_{\delta,j}(l)$ for any $\theta$ such that $\|\theta-\theta_0\|_2\le\delta$. Besides, $\|f_{\delta,j}\|_{P,2}\to0$ when $\delta\to0$.
\end{enumerate}
\end{assumption}
\ref{assump-4A} is a standard regularity assumption for Z-estimation. \ref{assump-4B} corresponds to the continuity assumption on $\psi(l,\theta)$ with respect to $\theta$. For example, the Lipschitz class of functions $\{f_\theta:\theta\in\Theta\}$ satisfies this condition if $\Theta$ is compact. More precisely, there exists an uniform envelop function $f$ such that $|f_{\theta_1}(l)-f_{\theta_2}(l)|\le\|\theta_1-\theta_2\|_2f(l)$ for any $\theta_1,\theta_2\in\Theta$ where $\|f\|_{P,2}<\infty$. Now, we establish the theorem.

\section{Proof of Theorem \ref{odds}}
\label{sec:proof-odds}

{\bf Proof sketch:}
\ref{assump-2D} assumes that there is a close approximation of the true propensity odds. We will show our estimator is close to the approximation. With the help of a few inequalities, the distance of functions is proportionally bounded by the distance of coefficients. The key to the proof is to show the distance between two coefficients converges with a certain order. The problem is converted to the study of a quadratic form with random coefficients in Lemma \ref{alpha}. The quadratic coefficients form a symmetric random matrix. By the Weyl inequality, we can connect the random matrix with the magnitude of basis functions. So, we can apply the matrix Bernstein inequality to provide bounds on the spectral norm, \emph{i.e.}, the largest eigenvalue. Similarly, we can show that the linear coefficients are also bounded. Lemmas \ref{quadratic} and \ref{linear} provide the bound for the quadratic and linear coefficients respectively.

\begin{proof} 
By the triangle inequality and \ref{assump-2D},
{\small
\begin{align*}
&\quads\underset{\lob\in\domr}{\sup}\left| O^r(\lob;\hatalphar)-O^r(\lob)\right|\\
&\le\underset{\lob\in\domr}{\sup}\left| O^r(\lob;\hatalphar)-O^r(\lob;\alpha_{K_r}^*)\right|
+\underset{\lob\in\domr}{\sup}\left| O^r(\lob;\alpha_{K_r}^*)-O^r(\lob)\right|\\
&\le\underset{\lob\in\domr}{\sup}\left|\exp\left\{\Phir(\lob)\tr\hatalphar\right\}-\exp\left\{\Phir(\lob)\tr\alpha_{K_r}^*\right\}\right|
+C_1K_r^{-\mu_1}\ .
\end{align*}
}
Since the exponential function is locally Lipschitz continuous, $|e^x-e^y|=e^y|e^{x-y}-1|\le2e^y|x-y|$ if $|x-y|\le\ln2$. By the triangle inequality, \ref{assump-2A}, and \ref{assump-2D}, 
{\small
\begin{align*}
\underset{\lob\in\domr}{\sup}O^r(\lob;\alpha_{K_r}^*)
&\le\underset{\lob\in\domr}{\sup}O^r(\lob)
+\underset{\lob\in\domr}{\sup}\left|O^r(\lob;\alpha_{K_r}^*)-O^r(\lob)\right|\\
&\le C_0+C_1K_r^{-\mu_1}\ .
\end{align*}
}
Thus, there exists large enough $N^*$ such that $\sup_{\lob\in\domr}O^r(\lob;\alpha_{K_r}^*)
\le 2C_0$ for all $N\ge N^*$. Therefore,
\begin{align}\label{odds-diff}
|\exp\left\{\Phir(\lob)\tr\hatalphar\right\}-\exp\{\Phir(\lob)\tr\alpha_{K_r}^*\}|
\le4C_0|\Phir(\lob)\tr\hatalphar-\Phir(\lob)\tr\alpha_{K_r}^*|
\end{align}
if $|\Phir(\lob)\tr\hatalphar-\Phir(\lob)\tr\alpha_{K_r}^*|\le\ln2$. By the Cauchy inequality and \ref{assump-2E}, 
$|\Phir(\lob)\tr\hatalphar-\Phir(\lob)\tr\alpha_{K_r}^*|\le K_r^{1/2}\|\hatalphar-\alpha_{K_r}^*\|_2$ for any $\lob\in\domr$. By Lemma \ref{alpha}, 
$\|\hatalphar-\alpha_{K_r}^*\|_2=O_p(\sqrt{K_r/N_r}+K_r^{-\mu_1})$. More precisely, for any $\epsilon>0$, there exists a finite $M_\epsilon>0$ and $N_\epsilon>0$ such that 
\begin{align*}
P\left\{\|\hatalphar-\alpha_{K_r}^*\|_2
>M_\epsilon\left(\sqrt{\frac{K_r}{N_r}}+K_r^{-\mu_1}\right)
\right\}<\epsilon
\end{align*}
for any $N\ge N_\epsilon$. Considering the complementary event, we can find $N_\epsilon^*$ large enough such that $M_\epsilon(K_r/\sqrt{N_r}+K_r^{1/2-\mu_1})<\ln2$ which makes the inequality \eqref{odds-diff} hold for any $N\ge N_\epsilon^*$. Then,
{\footnotesize
\begin{align*}
P\left\{\underset{\lob\in\domr}{\sup}\left| O^r(\lob;\hatalphar)-O^r(\lob)\right|
\le4C_0M_\epsilon\left(\frac{K_r}{\sqrt{N_r}}+K_r^{\frac{1}{2}-\mu_1}\right)+C_1K_r^{-\mu_1}
\right\}\ge1-\epsilon
\end{align*}
}
for all $N\ge\max\{N^*,N_\epsilon,N_\epsilon^*\}$. In other words, 
\begin{align*}
\underset{\lob\in\domr}{\sup}|O^r(\lob;\hatalphar)-O^r(\lob)|=O_p\left(\frac{K_r}{\sqrt{N_r}}+K_r^{\frac{1}{2}-\mu_1}\right)\ .
\end{align*}
Now, we consider the $L_2(P)$ norm.
{\small
\begin{align*}
&\quads\|O^r(\Lob;\hatalphar)-O^r(\Lob)\|_{P,2}\\
&\le\|O^r(\Lob;\hatalphar)-O^r(\Lob;\alpha_{K_r}^*)\|_{P,2}+\|O^r(\Lob;\alpha_{K_r}^*)-O^r(\Lob)\|_{P,2}\\
&\le\|O^r(\Lob;\hatalphar)-O^r(\Lob;\alpha_{K_r}^*)\|_{P,2}+\underset{\lob\in\domr}{\sup}\left| O^r(\lob;\alpha_{K_r}^*)-O^r(\lob)\right|\ .
\end{align*}
}
Following the similar arguments, when $|\Phir(\lob)\tr\hatalphar-\Phir(\lob)\tr\alpha_{K_r}^*|\le\ln2$, we have
\begin{align*}
&\quads\|O^r(\Lob;\hatalphar)-O^r(\Lob;\alpha_{K_r}^*)\|_{P,2}^2\\
&\le16C_0^2\int\left\{\Phir(\Lob)\tr\hatalphar-\Phir(\Lob)\tr\alpha_{K_r}^*\right\}^2dP(L)\\
&\le16C_0^2\int(\hatalphar-\alpha_{K_r}^*)\tr\Phir(\Lob)\Phir(\Lob)\tr(\hatalphar-\alpha_{K_r}^*)dP(L)\\
&\le16C_0^2\underset{\lob\in\domr}{\sup}\lambda_{\max}\{\Phir(\lob)\Phir(\lob)\tr\}\int(\hatalphar-\alpha_{K_r}^*)\tr(\hatalphar-\alpha_{K_r}^*)dP(L)\\
&\le16C_0^2\lambda_{\max}^*\|\hatalphar-\alpha_{K_r}^*\|_2^2\ .
\end{align*}
Thus, $
\|O^r(\Lob;\hatalphar)-O^r(\Lob;\alpha_{K_r}^*)\|_{P,2}=O_p(\sqrt{K_r/N_r}+K_r^{-\mu_1})$. Therefore,
\begin{align*}
\|O^r(\Lob;\hatalphar)-O^r(\Lob)\|_{P,2}
=O_p\left(\sqrt{\frac{K_r}{N_r}}+K_r^{-\mu_1}\right)\ .
\end{align*}
\end{proof}

\section{Proof of Theorem \ref{psi}}
\label{sec:proof-normality}

{\bf Proof sketch:}
We further decompose the error terms and show that the first three terms converge to 0 with the rate faster than $1/\sqrt{N}$, and the last term contributes as the influence function. Since the components in the decomposition involve the estimator, they should be treated as random functions. So, we consider the uniform convergence over a set of functions and apply the theory in \citet{van2000asymptotic}. With the maximal inequality via bracketing, the problem is converted to the control of the entropy integral, which requires the calculation of bracketing numbers. 
Lemmas \ref{bracket1}--\ref{bracket4} are bracketing inequalities which could be of independent interest.

\begin{proof}
First, recall that $\hat{\bbP}_N\psi_\theta-\E\{\psi_\theta(L)\}$ has the following decomposition 
{\small
\begin{align*}
\hat{\bbP}_N\psi_\theta-\E\{\psi_\theta(L)\}
=\sum_{r\in\calR}\left[\frac{1}{N}\sum_{i=1}^N\mathsf{1}_{R_i=1_d}O^r(\Lobi;\hatalphar)\psi_\theta(L_i)-\E\{\mathsf{1}_{R=r}\psi_\theta(L)\}\right]\ .
\end{align*}
}
For each missing pattern $r$, denote $1/N\sum_{i=1}^N\mathsf{1}_{R_i=1_d}O^r(\Lobi;\hatalphar)\psi_\theta(L_i)$ as $\hat{\bbP}_N^r\psi_\theta$. Then, $\hat{\bbP}_N^r\psi_\theta-\E\{\mathsf{1}_{R=r}\psi_\theta(L)\}$ can be decomposed into 4 parts:
\begin{align*}
S_{\theta,1}^r
&=\frac{1}{N}\sum_{i=1}^N\mathsf{1}_{R_i=1_d}\left\{O^r(\Lobi;\hatalphar)-O^r(\Lob_i)\right\}\left\{\psi_\theta(L_i)-u^r_\theta(\Lob_i)\right\}\ ,\\
S_{\theta,2}^r
&=\frac{1}{N}\sum_{i=1}^N\left\{\mathsf{1}_{R_i=1_d}O^r(\Lobi;\hatalphar)-\mathsf{1}_{R_i=r}\right\}\left\{u^r_\theta(\Lob_i)-\Phir(\Lob_i)\tr\beta_\theta^r\right\}\ ,\\
S_{\theta,3}^r
&=\frac{1}{N}\sum_{i=1}^N\left\{\mathsf{1}_{R_i=1_d}O^r(\Lobi;\hatalphar)-\mathsf{1}_{R_i=r}\right\}\Phir(\Lob_i)\tr\beta_\theta^r\ ,\\
S_{\theta,4}^r
&=\frac{1}{N}\sum_{i=1}^N\mathsf{1}_{R_i=1_d}O^r(\Lob_i)\left\{\psi_\theta(L_i)-u^r_\theta(\Lob_i)\right\}\\
&\quads+\frac{1}{N}\sum_{i=1}^N\mathsf{1}_{R_i=r}u^r_\theta(\Lob_i)-\E\{\mathsf{1}_{R=r}\psi_\theta(L)\}\ .
\end{align*}
For any fixed $\theta\in\Theta$ and any missing pattern $r$, by Lemmas \ref{S1}, \ref{S2}, and \ref{S3}, $\sqrt{N}|S_{\theta,i}^r|=o_p(1), i=1,2,3$. It's easy to see that $\E(S_{\theta,4}^r)=0$. Therefore, by the central limit theorem,
\begin{align*}
\sqrt{N}\left[\hat{\bbP}_N\psi_\theta-\E\{\psi_\theta(L)\}\right]\to\calN(0,V_\theta)
\end{align*} 
where $V_\theta=\E\{F_\theta(L,R)F_\theta(L,R)\tr\}$ and
{\small
\begin{align*}
F_\theta(L,R)=\mathsf{1}_{R=1_d}\sum_{r\in\calR}O^r(\Lob)\{\psi_\theta(L)-u_\theta^r(\Lob)\}+\sum_{r\in\calR}\mathsf{1}_{R=r}u_\theta^r(\Lob)-\E\{\psi_\theta(L)\}\ .
\end{align*}
}
\end{proof}

\section{Proof of Theorem \ref{theta}}

{\bf Proof sketch:}
First, by \ref{assump-4A}, the convergence of $\hat{\theta}_N$ should be implied by the uniform convergence of $\hat{\bbP}_N\psi_\theta$ over $\theta$ in a compact set. Second, we study the convergence of $\E\{\psi_{\hat{\theta}_N}(L)\}$ and apply the Delta method to obtain the limiting distribution of $\hat{\theta}_N$. The functional version of the central limit theorem, \emph{i.e.} Donsker's theorem, is applied to achieve uniform convergence.

\begin{proof}
Denote the empirical average $N^{-1}\sum_{i=1}^N\psi_\theta(L_i)$ as $\bbP_N\psi_\theta$ and the centered and scaled version $\sqrt{N}[\bbP_N\psi_\theta-\E\{\psi_\theta(L)\}]$ as $\bbG_N\psi_\theta$. Recall the proposed weighted average is
\begin{align*}
\hat{\bbP}_N\psi_\theta=\frac{1}{N}\sum_{i=1}^N\left\{\mathsf{1}_{R_i=1_d}\hat{w}(L_i)\psi_\theta(L_i)\right\}\ .
\end{align*}
Since $\hat{\theta}_N$ is the solution to $\hat{\bbP}_N\psi_\theta=0$, by Lemma \ref{uniformbound},
\begin{align*}
\E\{\psi_{\hat{\theta}_N}(L)\}=\E\{\psi_{\hat{\theta}_N}(L)\}-\hat{\bbP}_N\psi_{\hat{\theta}_N}
\le\underset{\theta\in\Theta}{\sup}\left|\hat{\bbP}_N\psi_\theta-\E\{\psi_\theta(L)\}\right|=o_p(1)\ .
\end{align*}
By identifiability condition \ref{assump-4A}, $\|\hat{\theta}_N-\theta_0\|_2\xrightarrow{P}0$. 

Next, we investigate the asymptotic normality of $\hat{\theta}_N$. Although $\E\{\psi_{\hat{\theta}_N}(L)\}$ has a form of expectation over the population, it can be viewed as a random vector because $\hat{\theta}_N$ depends on the observations. Since $\hat{\theta}_N\xrightarrow{P}\theta_0$, one would expect that $\E\{\psi_{\hat{\theta}_N}(L)\}$ converges to $\E\{\psi_{\theta_0}(L)\}$ in some way. If the limiting distribution is known, one could apply the Delta method to obtain limiting distribution of $\hat{\theta}_N$. From Theorem \ref{psi}, we have the asymptotic normality of $\hat{\bbP}_N\psi_\theta$ for any fixed $\theta\in\Theta$. It is natural to consider
\begin{align}\label{difference-of-interest}
\left[\hat{\bbP}_N\psi_{\hat{\theta}_N}-\E\{\psi_{\hat{\theta}_N}(L)\}\right]
-\left[\hat{\bbP}_N\psi_{\theta_0}-\E\{\psi_{\theta_0}(L)\}\right]
\end{align}
The above difference has a similar form to the asymptotic equicontinuity, which can be derived if the function class is Donsker. More precisely, consider the class of $j$-th entry of the estimating functions, $\Psi_j:=\{\psi_{\theta,j}:\theta\in\Theta\}$. It is Donsker due to Theorem 19.5 in \cite{van2000asymptotic} and
\begin{align*}
J_{[ \ ]}\{1,\Psi_j,L_2(P)\}
&=\int_0^1\sqrt{\log n_{[ \ ]}\{\epsilon,\Psi_j,L_2(P)\}}d\epsilon\\
&\le\int_0^1\sqrt{\log n_{[ \ ]}\{\epsilon,\calH,L_2(P)\}}d\epsilon\\
&\le\int_0^1\sqrt{C_\calH}\epsilon^{-\frac{1}{2d_\calH}}d\epsilon
=\sqrt{C_\calH}\le\infty\ .
\end{align*}
Then, by Section 2.1.2 in \cite{wellner2013weak}, we have the following asymptotic equicontinuity: for any $\epsilon,\eta>0$, there exists $C_{\epsilon,\eta}>0$ and $N_{\epsilon,\eta}$ such that for all $N\ge N_{\epsilon,\eta}$, 
\begin{align*}
P\left(\underset{\psi_{\theta,j}:\rho_P(\psi_{\theta,j}-\psi_{\theta_0,j})<C_{\epsilon,\eta}}{\sup}\left|\bbG_N\psi_{\theta,j}-\bbG_N\psi_{\theta_0,j}\right|\ge\epsilon\right)\le\frac{\eta}{2}
\end{align*}
where the seminorm $\rho_P$ is defined as $\rho_P(f)=\{P(f-Pf)^2\}^{1/2}$. Consider
\begin{align*}
&\quads\bbG_N\psi_{\hat{\theta}_N,j}-\bbG_N\psi_{\theta_0,j}\\
&=\sqrt{N}\left[\bbP_N\psi_{\hat{\theta}_N,j}-\E\{\psi_{\hat{\theta}_N,j}(L)\}\right]
-\sqrt{N}\left[\bbP_N\psi_{\theta_0,j}-\E\{\psi_{\theta_0,j}(L)\}\right]\ .
\end{align*}
Notice that $\rho_P(f)\le\|f\|_{P,2}$. By \ref{assump-4B}, for any $\delta>0$, there exists an envelop function $f_{\delta,j}$ such that
{\small
\begin{align*}
P\left(\|\hat{\theta}_N-\theta_0\|_2<\delta\right)
\le P\left(\|\psi_{\hat{\theta}_N,j}-\psi_{\theta_0,j}\|_{P,2}<C_\delta\right)
\le P\left\{\rho_P(\psi_{\hat{\theta}_N,j}-\psi_{\theta_0,j})<C_\delta\right\}
\end{align*}
}
where $C_\delta=\|f_{\delta,j}\|_{P,2}\to0$ when $\delta\to0$. Thus, there exists $\delta_{\epsilon,\eta}$ small enough such that $C_\delta\le C_{\epsilon,\eta}$ for all $\delta\le\delta_{\epsilon,\eta}$. Then, by the consistency of $\hat{\theta}_N$, there exists $N_{\epsilon,\eta}^*$ such that for all $N\ge N_{\epsilon,\eta}^*$, 
\begin{align*}
P\left(\|\hat{\theta}_N-\theta_0\|_2\ge\delta_{\epsilon,\eta}\right)
\le\frac{\eta}{2}\ .
\end{align*}
Thus, 
\begin{align*}
P\left\{\rho_P(\psi_{\hat{\theta}_N,j}-\psi_{\theta_0,j})<C_{\epsilon,\eta}\right\}
>1-\frac{\eta}{2}\ .
\end{align*}
Note that the event $\rho_P(\psi_{\hat{\theta}_N,j}-\psi_{\theta_0,j})<C_{\epsilon,\eta}$ and
\begin{align*}
\underset{\psi_{\theta,j}:\rho_P(\psi_{\theta,j}-\psi_{\theta_0,j})<C_{\epsilon,\eta}}{\sup}\left|\bbG_N\psi_{\theta,j}-\bbG_N\psi_{\theta_0,j}\right|<\epsilon
\end{align*}
happening together implies $|\bbG_N\psi_{\hat{\theta}_N,j}-\bbG_N\psi_{\theta_0,j}|<\epsilon$. By taking complementary event, for any $N\ge\max\{N_{\epsilon,\eta},N_{\epsilon,\eta}^*\}$, we obtain
\begin{align*}
&P\left(\left|\bbG_N\psi_{\hat{\theta}_N,j}-\bbG_N\psi_{\theta_0,j}\right|\ge\epsilon\right)
\le P\left\{\rho_P(\psi_{\hat{\theta}_N,j}-\psi_{\theta_0,j})\ge C_{\epsilon,\eta}\right\}\\
&\quad\quad\quad\quad
+P\left(\underset{\psi_{\theta,j}:\rho_P(\psi_{\theta,j}-\psi_{\theta_0,j})<C_{\epsilon,\eta}}{\sup}\left|\bbG_N\psi_{\theta,j}-\bbG_N\psi_{\theta_0,j}\right|\ge\epsilon\right)\\
&\le\frac{\eta}{2}+\frac{\eta}{2}=\eta\ .
\end{align*}
That is, for each $j$-th entry,
\begin{align}\label{equicontinuity}
\sqrt{N}\left[\bbP_N\psi_{\hat{\theta}_N,j}-\E\{\psi_{\hat{\theta}_N,j}(L)\}\right]
-\sqrt{N}\left[\bbP_N\psi_{\theta_0,j}-\E\{\psi_{\theta_0,j}(L)\}\right]
\xrightarrow{P}0\ .
\end{align}
Therefore, by the comparison between terms in \eqref{difference-of-interest} and $\bbG_N\psi_{\hat{\theta}_N,j}-\bbG_N\psi_{\theta_0,j}$, we should consider
\begin{align*}
\sqrt{N}\left[\hat{\bbP}_N\psi_{\hat{\theta}_N,j}-\bbP_N\psi_{\hat{\theta}_N,j}\right]
-\sqrt{N}\left[\hat{\bbP}_N\psi_{\theta_0,j}-\bbP_N\psi_{\theta_0,j}\right]
\end{align*}
which can be decomposed as the following terms. 
\begin{align*}
&\quads\sum_{r\in\calR}\left(S_{\hat{\theta}_N,1}^r+S_{\hat{\theta}_N,2}^r+S_{\hat{\theta}_N,3}^r+S_{\hat{\theta}_N,5}^r-S_{\hat{\theta}_N,6}^r\right)\\
&-\sum_{r\in\calR}\left(S_{\theta_0,1}^r+S_{\theta_0,2}^r+S_{\theta_0,3}^r+S_{\theta_0,5}^r-S_{\theta_0,6}^r\right)
\end{align*}
where
\begin{align*}
S_{\theta,5}^r&=\frac{1}{N}\sum_{i=1}^N\left\{\mathsf{1}_{R_i=1_d}O^r(\Lob_i)-\mathsf{1}_{R_i=r}\right\}\psi_\theta(L_i)\ ,\\
S_{\theta,6}^r&=\frac{1}{N}\sum_{i=1}^N\left\{\mathsf{1}_{R_i=1_d}O^r(\Lob_i)-\mathsf{1}_{R_i=r}\right\}u^r_\theta(\Lob_i)\ .
\end{align*}
By Lemmas \ref{S1}, \ref{S2}, and \ref{S3}, $\sqrt{N}\left|S_{\theta,i}^r\right|=o_p(1),i=1,2,3$ for any missing pattern $r$ and $\theta\in\Theta$. Combing with Lemmas \ref{S5} and \ref{S6}, we have
\begin{align}\label{weighted-equicontinuity}
\sqrt{N}\left(\hat{\bbP}_N\psi_{\hat{\theta}_N}-\bbP_N\psi_{\hat{\theta}_N}-\hat{\bbP}_N\psi_{\theta_0}+\bbP_N\psi_{\theta_0}\right)
\xrightarrow{P}0\ .
\end{align}
By Equations \eqref{weighted-equicontinuity} and \eqref{equicontinuity}, we have
\begin{align*}
\sqrt{N}\left[\hat{\bbP}_N\psi_{\hat{\theta}_N}-\E\{\psi_{\hat{\theta}_N}(L)\}-\hat{\bbP}_N\psi_{\theta_0}+\E\{\psi_{\theta_0}(L)\}\right]
\xrightarrow{P}0\ .
\end{align*}
Since $\hat{\bbP}_N\psi_{\hat{\theta}_N}=0$ and $\E\{\psi_{\theta_0}(L)\}=0$, the above equation can be rewritten as
\begin{align*}
\sqrt{N}\left[\E\{\psi_{\hat{\theta}_N}(L)\}-\E\{\psi_{\theta_0}(L)\}+\hat{\bbP}_N\psi_{\theta_0}-\E\{\psi_{\theta_0}(L)\}\right]
\xrightarrow{P}0\ .
\end{align*}
By Theorem \ref{psi},
\begin{align*}
\sqrt{N}\left[\hat{\bbP}_N\psi_{\theta_0}-\E\{\psi_{\theta_0}(L)\}\right]
\overset{d}{\to}N(0,V_{\theta_0})\ .
\end{align*}
Since $D_{\theta_0}$ is nonsingular, by multivariate Delta method,
\begin{align*}
\sqrt{N}(\hat{\theta}_N-\theta_0)\overset{d}{\to}
N\left(0,D_{\theta_0}^{-1}V_{\theta_0}D_{\theta_0}^{-1\tr}\right)\ .
\end{align*}
Therefore, $\hat{\theta}_N$ is semiparametrically efficient.

Lastly, we look into the estimator for the asymptotic variance. We have the following decomposition:
\begin{align*}
\hat{D}_{\hat{\theta}_N}-D_{\theta_0}
&=\frac{1}{N}\sum_{i=1}^N\left[\mathsf{1}_{R_i=1_d}\left\{\hat{w}(L_i)-\sum_{r\in\calR}O^r(\Lob_i)\right\}\dot{\psi}_{\hat{\theta}_N}(L_i)\right]\\
&\quads+\frac{1}{N}\sum_{i=1}^N\mathsf{1}_{R_i=1_d}\frac{1}{P(R_i=1_d\mid L_i)}\dot{\psi}_{\hat{\theta}_N}(L_i)-D_{\hat{\theta}_N}+D_{\hat{\theta}_N}-D_{\theta_0}\ .
\end{align*}
where $\hat{w}(L_i)=\sum_{r\in\calR}O^r(\Lobi;\hatalphar)$. 

Consider the first term on the right hand side. By Theorem \ref{odds},
\begin{align*}
&\quads\|\hat{w}(l)-1/P(R=1_d\mid l)\|_\infty\\
&\le\sum_{r\in\calR}\|O^r(\ \cdot\ ;\hatalphar)-O^r\|_\infty
=O_p(\sqrt{K_r^2/N_r}+K_r^{1/2-\mu_1})\ .
\end{align*}
Let
\begin{align*}
\bfF=\frac{1}{N}\sum_{i=1}^N\mathsf{1}_{R_i=1_d}\dot{\psi}_{\hat{\theta}_N}(L_i)\dot{\psi}_{\hat{\theta}_N}(L_i)\tr\ .
\end{align*}
Following the similar arguments in Lemma \ref{linear}, one can see that
\begin{align*}
&\quads\left\|\frac{1}{N}\sum_{i=1}^N\left[\mathsf{1}_{R_i=1_d}\left\{\hat{w}(L_i)-\sum_{r\in\calR}O^r(\Lob_i)\right\}\dot{\psi}_{\hat{\theta}_N}(L_i)\right]\right\|_2^2\\
&\le\underset{\theta\in\Theta}{\sup}\frac{1}{N}\sum_{i=1}^N\mathsf{1}_{R_i=1_d}\left\{\hat{w}(L_i)-\sum_{r\in\calR}O^r(\Lob_i)\right\}^2\lambda_{\max}\left\{\bfF\right\}\\
&\le\lambda_{\max}'\|\hat{w}(l)-1/P(R=1_d\mid l)\|_\infty^2=o_p(1)\ .
\end{align*}

Consider the second term on the right hand side. Notice that $\dot{\psi}_\theta$ is the Jacobian matrix of $\psi_\theta$. We consider all the entries of $\dot{\psi}_\theta$ together and abbreviate the subscripts in the following statements. Define a set of functions $\calF_\Theta:=\{f_{\theta}:\theta\in\Theta\}$ where $f_\theta(L,R):=1_{R=1_d}/P(R=1_d\mid L)\dot{\psi}_\theta(L)$. Similar to Lemma \ref{bracket4}, one can show that
\begin{align*}
n_{[ \ ]}\{\epsilon/\delta_0,\calF_\Theta,L_1(P)\}
\le n_{[ \ ]}\{\epsilon,\calJ_\Theta,L_1(P)\}<\infty
\end{align*}
since $1_{R=1_d}/P(R=1_d\mid L)\le1/\delta_0$. Therefore, by Theorem 19.4 in \cite{van2000asymptotic}, $\calF_\Theta$ is P-Glivenko-Cantelli. Since the set $\calF_\Theta$ includes all entries of the Jacobian matrix, we consider the Frobenius/Euclidean norm of a matrix to construct the following convergence result.
\begin{align*}
\underset{f_\theta\in\calF_\Theta}{\sup}\left\|\bbP_Nf_\theta-Pf_\theta\right\|_F\xrightarrow{a.s.}0
\end{align*}
where $\|\cdot\|_F$ is the Euclidean norm of a matrix. The fact that $\|\cdot\|_2\le\|\cdot\|_F$ implies
\begin{align*}
\left\|\frac{1}{N}\sum_{i=1}^N\mathsf{1}_{R_i=1_d}\frac{1}{P(R_i=1_d\mid L_i)}\dot{\psi}_{\hat{\theta}_N}(L_i)-D_{\hat{\theta}_N}\right\|_2=o_p(1)\ .
\end{align*}

Finally, $D_{\hat{\theta}_N}\xrightarrow{P}D_{\theta_0}$ since $\|\hat{\theta}_N-\theta_0\|_2\xrightarrow{P}0$ and $\dot{\psi}_\theta$ is continuous in a neighborhood of $\theta_0$. Therefore, $\hat{D}_{\hat{\theta}_N}$ is a consistent estimator of $D_{\theta_0}$. 

We skip the details but provide a skeleton of the following proof. Notice that each component of $\hat{F}_i$ converges to the corresponding true value. Therefore, $\hat{F}_i$ and $V_{\hat{\theta}_N}$ are consistent estimators of $F_{\theta_0}(L_i,R_i)$ and $V_{\theta_0}$ respectively. Since $\hat{D}_{\hat{\theta}_N}^{-1}V_{\hat{\theta}_N}\hat{D}_{\hat{\theta}_N}^{-1\tr}$ is a standard sandwich estimator, it is easy to show it is a consistent estimator of the above asymptotic variance.
\end{proof}

\section{Related Lemmas}
\label{sec:lemma}
\begin{lemma}[Weyl's inequality]\label{Weyl}
Let $\bfA$ and $\bfB$ be $m \times m$ Hermitian matrices and $\bfC=\bfA-\bfB$. Suppose their respective eigenvalues $\mu_i,\nu_i,\rho_i$ are ordered as follows: 
\begin{align*}
\bfA:\quad\mu_1\ge\cdots\ge\mu_m\ ,\\
\bfB:\quad\nu_1\ge\cdots\ge\nu_m\ ,\\
\bfC:\quad\rho_1\ge\cdots\ge\rho_m\ .
\end{align*}
Then, the following inequalities hold.
\begin{align*}
\rho_m\le\mu_i-\nu_i\le\rho_1,\quad i=1,\cdots,m\ .
\end{align*}
In particular, if $\bfC$ is positive semi-definite, plugging $\rho_m\ge 0$ into the above inequalities leads to
\begin{align*}
\mu_i\ge\nu_i,\quad i=1,\cdots,m\ .
\end{align*}
\end{lemma}

\begin{lemma}[Bernstein’s inequality] \label{Bernstein}
Let $\{\bfA_i\}_{i=1}^N$ be a sequence of independent random matrices with dimensions $d_1\times d_2$. Assume that $\E\{\bfA_i\}=\mathbf{0}_{d_1,d_2}$ and $\|\bfA_i\|_2\le c$ almost surely for all $i=1,\cdots,N$ and some constant $c$. Also assume that
\begin{align*}
\max\left\{\left\|\sum_{i=1}^N\E(\bfA_i\bfA_i\tr)\right\|_2,\left\|\sum_{i=1}^N\E(\bfA_i\tr \bfA_i)\right\|_2\right\}\le\sigma^2\ .
\end{align*}
Then, for all $t \ge 0$,
\begin{align*}
P\left(\left\|\sum_{i=1}^N\bfA_i\right\|_2\ge t\right)\le(d_1+d_2)\exp\left(-\frac{t^2/2}{\sigma^2+ct/3}\right)\ .
\end{align*}
\end{lemma}

\begin{lemma}\label{alpha}
Under Assumptions \ref{assump1} and \ref{assump2}, the minimizer $\hatalphar$ satisfies 
\begin{align*}
\|\hatalphar-\alpha_{K_r}^*\|_2=O_p\left(\sqrt{\frac{K_r}{N_r}}+K_r^{-\mu_1}\right)=o_p(1)\ .
\end{align*}
\end{lemma}

\begin{proof}
It suffices to show for any $\epsilon>0$, there exists $C_\epsilon$ and $N_\epsilon$ such that 
\begin{align}\label{Op-alpha}
P\left\{\|\hatalphar-\alpha_{K_r}^*\|_2>C_\epsilon\left(\sqrt{\frac{K_r}{N_r}}+K_r^{-\mu_1}\right)\right\}\le\epsilon
\end{align}
for any $N\ge N_\epsilon$. It means that the minimizer $\hatalphar$ is in a small neighbourhood of $\alpha_{K_r}^*$ with probability higher than $1-\epsilon$. Consider the set $\Delta=\{\delta\in\bbR^{K_r}:\|\delta\|_2\le C(\sqrt{K_r/N_r}+K_r^{-\mu_1})\}$ for an arbitrary constant $C$. Since $\calLr_\lambda$ is a convex function of $\alphar$, the minimizer $\hatalphar\in\alpha_{K_r}^*+\Delta$ if $\inf_{\delta\in\partial\Delta}{\inf}\calLr_\lambda(\alpha_{K_r}^*+\delta)> \calLr_\lambda(\alpha_{K_r}^*)$. Thus, considering the complementary event, we have
{\small
\begin{align*}
P\left\{\|\hatalphar-\alpha_{K_r}^*\|_2>C\left(\sqrt{\frac{K_r}{N_r}}+K_r^{-\mu_1}\right)\right\}
\le
P\left\{\underset{\delta\in\partial\Delta}{\inf}\calLr_\lambda(\alpha_{K_r}^*+\delta)\le\calLr_\lambda(\alpha_{K_r}^*)\right\}\ .
\end{align*} 
}
Recall that for any $r\neq 1_d$ and any $\lambda>0$, the objective function is
\begin{align*}
\calLr_\lambda(\alphar)
=\frac{1}{N}\sum_{i=1}^N\left\{
\mathsf{1}_{R_i=1_d}O^r(\Lobi;\alphar)\hat{Q}^{\PAr}(L_i)-\mathsf{1}_{R=r}\log O^r(\Lobi;\alphar)
\right\}+\lambda J^r(\alphar)
\end{align*}
where $O^r(\lob;\alphar)=\exp\{\Phir(\lob)\tr\alphar\}$ and $J^r(\alphar)=\sum_{k=1}^{K_r}t_k|\alphar_k|$. To investigate $\inf_{\delta\in\partial\Delta}\calLr_\lambda(\alpha_{K_r}^*+\delta)-\calLr_\lambda(\alpha_{K_r}^*)$, we apply the mean value theorem. There exists some $\tilde{\alpha}^r$ satisfying $\tilde{\alpha}^r-\alpha_{K_r}^*\in\mathrm{int}(\Delta)$, which is the interior of $\Delta$, such that for any $\delta\in\Delta$, 
\begin{align*}
&\quads\calLr_\lambda(\alpha_{K_r}^*+\delta)-\calLr_\lambda(\alpha_{K_r}^*)\\
&=\delta\tr\left.\frac{\partial\calLr_N(\alphar)}{\partial\alphar}\right\vert_{\alpha_{K_r}^*}
+\frac{1}{2}\delta\tr\left\{\left.\frac{\partial^2\calLr_N(\alphar)}{(\partial\alphar)^2}\right\vert_{\tilde{\alpha}^r}\right\}\delta
+\lambda J^r(\alpha_{K_r}^*+\delta)-\lambda J^r(\alpha_{K_r}^*)\ .
\end{align*}
By the triangle inequality and Cauchy inequality, the difference between penalties satisfies
\begin{align*}
&\quads J^r(\alpha_{K_r}^*+\delta)-J^r(\alpha_{K_r}^*)\\
&=\sum_{k=1}^{K_r}\left(\left|\alphar_{*,k}+\delta_k^r\right|-\left|\alphar_{*,k}\right|\right)t_k
\ge-\sum_{k=1}^{K_r}\left|\delta_k^r\right|t_k
\ge-\sqrt{\sum_{k=1}^{K_r}t_k^2}\left\|\delta\right\|_2\ .
\end{align*}
Denote the constant $\sqrt{\sum_{k=1}^{K_r}t_k^2}$ as $c_{\mathrm{lin}}$. Then, by the Cauchy inequality again, 
\begin{align*}
\calLr_\lambda(\alpha_{K_r}^*+\delta)-\calLr_\lambda(\alpha_{K_r}^*)
&\ge
\frac{1}{2}\delta\tr
\left\{\left.\frac{\partial^2\calLr_N(\alphar)}{(\partial\alphar)^2}\right\vert_{\tilde{\alpha}^r}\right\}
\delta
-\left\{\left\|\left.\frac{\partial\calLr_N(\alphar)}{\partial\alphar}\right\vert_{\alpha_{K_r}^*}\right\|_2+\lambda c_{\mathrm{lin}}\right\}
\left\|\delta\right\|_2
\ .
\end{align*}
First, let's have a look at the quadratic coefficient. By Lemma \ref{quadratic}, the quadratic terms are bounded from below. More precisely, for any $\epsilon>0$, there exists $N_\epsilon^*$ such that for any $N\ge N_\epsilon^*$,
\begin{align*}
P\left[\delta\tr\left\{\left.\frac{\partial^2\calLr_N(\alphar)}{(\partial\alphar)^2}\right\vert_{\tilde{\alpha}^r}\right\}\delta
\ge C_{\mathrm{quad}}\|\delta\|_2^2\right]\ge1-\frac{1}{2}\epsilon\ .
\end{align*} 
Next, let's investigate the bound of the linear coefficient. By \ref{assump-2E}, $\lambda=O(\sqrt{K_r/N})$. By Lemma \ref{linear}, for any $\epsilon>0$, there exists $N_\epsilon'$ and a constant $C_\epsilon'$ such that for any $N\ge N_\epsilon'$,
\begin{align*}
P\left[
\left\{\left\|\left.\frac{\partial\calLr_N(\alphar)}{\partial\alphar}\right\vert_{\alpha_{K_r}^*}\right\|_2+\lambda c_{\mathrm{lin}}\right\}
\ge C_\epsilon'\left(\sqrt{\frac{K_r}{N}}+K_r^{-\mu_1}\right)\right]
\le\frac{1}{2}\epsilon\ .
\end{align*} 
Considering the complement of the above event and the fact that $P(A\cap B)=P(A)+P(B)-P(A\cup B)\ge P(A)+P(B)-1$, we have
{\small
\begin{align*}
P\left\{\calLr_\lambda(\alpha_{K_r}^*+\delta)-\calLr_\lambda(\alpha_{K_r}^*)
\ge\frac{C_{\mathrm{quad}}}{2}\|\delta\|_2^2
-C_\epsilon'\left(\sqrt{\frac{K_r}{N}}+K_r^{-\mu_1}\right)\|\delta\|_2
\right\}\ge 1-\epsilon
\end{align*}
}
for any $N\ge\max\{N_\epsilon^*,N_\epsilon'\}$. Note that $\partial\Delta=\{\delta\in\bbR^{K_r}:\|\delta\|_2=C(\sqrt{K_r/N_r}+K_r^{-\mu_1})\}$. Choosing $C>2C_\epsilon'/C_{\mathrm{quad}}$, we have $P\{\inf_{\delta\in\partial\Delta}\calLr_\lambda(\alpha_{K_r}^*+\delta)\ge\calLr_\lambda(\alpha_{K_r}^*)\}\ge1-\epsilon$ for any $\epsilon>0$. Therefore, inequality \eqref{Op-alpha} holds which completes the proof.
\end{proof}

\begin{lemma}\label{quadratic}
There exists a constant $C_{\mathrm{quad}}$ such that the Hessian matrix of $\calLr_N(\alphar)$ at $\tilde{\alpha}^r$ satisfies
\begin{align*}
\lim_{N\to\infty}P\left[
\lambda_{\min}\left\{\left.\frac{\partial^2\calLr_N(\alphar)}{(\partial\alphar)^2}\right\vert_{\tilde{\alpha}^r}\right\}
\ge C_{\mathrm{quad}}
\right]=1
\end{align*} 
where $\lambda_{\min}(\cdot)$ represents the minimal eigenvalue of the matrix. 
\end{lemma}

\begin{proof}
Denote the Hessian matrix of $\calLr_N(\alphar)$ at $\tilde{\alpha}^r$ as $\bfA$:
\begin{align*}
\bfA=\left.\frac{\partial^2\calLr_N(\alphar)}{(\partial\alphar)^2}\right\vert_{\tilde{\alpha}^r}
=\frac{1}{N}\sum_{i=1}^N\left\{\mathsf{1}_{R_i=1_d}O^r(\Lobi;\tilde{\alpha}^r)\hat{Q}^{\PAr}(L_i)\Phir(\Lob_i)\Phir(\Lob_i)\tr\right\}\ .
\end{align*} 
Recall that the set $\Delta=\{\delta\in\bbR^{K_r}:\|\delta\|_2\le C(\sqrt{K_r/N_r}+K_r^{-\mu_1})\}$ and $\tilde{\alpha}^r-\alpha_{K_r}^*\in\mathrm{int}(\Delta)$. Following the similar arguments in the proof of Theorem \ref{odds} (Appendix \ref{sec:proof-odds}), it can be easily shown that
\begin{align*}
&\quads O^r(\lob;\tilde{\alpha}^r)\\
&\ge O^r(\lob)-\left|O^r(\lob)-O^r(\lob;\alpha_{K_r}^*)\right|
-4C_0\left|\Phir(\lob)\tr\tilde{\alpha}^r-\Phir(\lob)\tr\alpha_{K_r}^*\right|\\
&\ge c_0-C_1K_r^{-\mu_1}-4C_0C(K_r/\sqrt{N_r}+K_r^{\frac{1}{2}-\mu_1})\ .
\end{align*}
Besides, for any $s\in\PAr$,
\begin{align*}
&\quads \hat{Q}^s(l)
\ge Q^r(l)-\left|\hat{Q}^s(l)-Q^s(l)\right|\\
&\ge \sum_{\Xi\in\Pi_s}c_0^{|\Xi|}-C_1K_r^{-\mu_1}-4C_0C(K_r/\sqrt{N_r}+K_r^{\frac{1}{2}-\mu_1})\ .
\end{align*}
Then, there exists $N_\Delta$ such that $O^r(\lob;\tilde{\alpha}^r)>c_0/2$ holds for any $N\ge N_\Delta$. Let 
\begin{align*}
\bfB&=\frac{1}{N}\sum_{i=1}^N\left\{\frac{1}{2}c_0\mathsf{1}_{R_i=1_d}\Phir(\Lob_i)\Phir(\Lob_i)\tr\right\},\\
\bfC&=\frac{1}{2}c_0\E\left\{\mathsf{1}_{R=1_d}\Phir(\Lob)\Phir(\Lob)\tr\right\}
=\frac{1}{2}c_0\E\left\{P(R=1_d\mid\Lob)\Phir(\Lob)\Phir(\Lob)\tr\right\},\\
\bfD&=\frac{1}{2}c_0\delta_0\E\left\{\Phir(\Lob)\Phir(\Lob)\tr\right\}\ .
\end{align*} 
It's easy to see that matrices $\bfA,\bfB,\bfC,\bfD$ are symmetric. Based on the above discussions, $\bfA-\bfB$ is positive semi-definite for large enough $N$. By \ref{assump-1B}, $\bfC-\bfD$ is also positive semi-definite. Applying Lemma \ref{Weyl}, we have $\lambda_{\min}(\bfA)\ge\lambda_{\min}(\bfB)$, $\lambda_{\min}(\bfC)\ge\lambda_{\min}(\bfD)$ and $|\lambda_{\min}(\bfB)-\lambda_{\min}(\bfC)|\le\max\{|\lambda_{\min}(\bfB-\bfC)|,|\lambda_{\max}(\bfB-\bfC)|\}=\|\bfB-\bfC\|_2$.
Therefore, $\lambda_{\min}(\bfA)\ge\lambda_{\min}(\bfD)-\|\bfB-\bfC\|_2\ge c_0\delta_0\lambda_{\min}^*/2-\|\bfB-\bfC\|_2$. To study $\|\bfB-\bfC\|_2$, we apply Lemma \ref{Bernstein}. Let 
\begin{align*}
\bfE_i=\frac{1}{N}\Big[\mathsf{1}_{R_i=1_d}\Phir(\Lob_i)\Phir(\Lob_i)\tr-\E\left\{\mathsf{1}_{R=1_d}\Phir(\Lob)\Phir(\Lob)\tr\right\}\Big]\ .
\end{align*} 
So, $\E\{\bfE_i\}=\mathbf{0}_{K_r,K_r}$. By the triangle inequality, Lemma \ref{Weyl} and the fact that $\|\cdot\|_2\le\|\cdot\|_F$,
{\small
\begin{align*}
\|\bfE_i\|_2
&\le\frac{1}{N}\mathsf{1}_{R_i=1_d}\|\Phir(\Lob_i)\Phir(\Lob_i)\tr\|_F
+\frac{1}{N}\|\E\left\{\mathsf{1}_{R=1_d}\Phir(\Lob)\Phir(\Lob)\tr\right\}\|_2\\
&\le\frac{1}{N}\sqrt{\textrm{trace}\{\Phir(\Lob_i)\Phir(\Lob_i)\tr\Phir(\Lob_i)\Phir(\Lob_i)\tr\}}
+\frac{1}{N}\|\E\left\{\Phir(\Lob)\Phir(\Lob)\tr\right\}\|_2\\
&=\frac{1}{N}\|\Phir(\Lob_i)\|_2^2
+\frac{1}{N}\|\E\left\{\Phir(\Lob)\Phir(\Lob)\tr\right\}\|_2\ .
\end{align*} 
}
By \ref{assump-2E}, \ref{assump-2F} and the fact that $N_r/N<1$, $\|\bfE_i\|_2=O(K_r/N_r)$. Similarly,
\begin{align*}
&\quads\left\|\sum_{i=1}^N\E(\bfE_i\bfE_i\tr)\right\|_2\\
&\le\frac{1}{N}\left\|\E\{\mathsf{1}_{R=1_d}\Phir(\Lob)\Phir(\Lob)\tr\Phir(\Lob)\Phir(\Lob)\tr\}\right\|_2\\
&\quads+\frac{1}{N}\left\|\E\{\mathsf{1}_{R=1_d}\Phir(\Lob)\Phir(\Lob)\tr\}\E\left\{\mathsf{1}_{R=1_d}\Phir(\Lob)\Phir(\Lob)\tr\right\}\right\|_2\\
&\le\frac{1}{N}\underset{\lob\domr}{\sup}\|\Phir(\lob)\|_2^2\|\E\left\{\Phir(\Lob)\Phir(\Lob)\tr\right\}\|_2
+\frac{1}{N}\|\E\left\{\Phir(\Lob)\Phir(\Lob)\tr\right\}\|_2^2\\
&=O(K_r/N_r)\ .
\end{align*} 
Taking $t=C\sqrt{K_r\log K_r/N_r}$ in Lemma \ref{Bernstein} for an arbitrary constant $C$, we have
\begin{align*}
\exp\left(\left\|\sum_{i=1}^N\bfE_i\right\|_2\ge t\right)\le2K_r\exp(-C'\log K_r)
\end{align*}
for large enough $N$ and some constant $C'$. In other words, $\|\bfB-\bfC\|_2=O_p(\sqrt{K_r\log K_r/N_r})=o_p(1)$. Therefore, for any $\epsilon$ there exists $N_{\Delta,\epsilon}$ such that 
\begin{align*}
P\left\{\lambda_{\min}(\bfA)\ge\frac{1}{4}c_0\delta_0\lambda_{\min}^*\right\}
\ge1-\epsilon
\end{align*}
for any $N\ge\max\{N_\Delta,N_{\Delta,\epsilon}\}$.
\end{proof}

\begin{lemma}\label{linear}
The gradient of $\calLr_N(\alphar)$ at $\alpha_{K_r}^*$ satisfies
\begin{align*}
\left\|\left.\frac{\partial\calLr_N(\alphar)}{\partial\alphar}\right\vert_{\alpha_{K_r}^*}\right\|_2
&=O_p\left(\sqrt{\frac{K_r}{N}}+K_r^{-\mu_1}\right)\ .
\end{align*} 
\end{lemma}

\begin{proof}
The gradient of $\calLr_N(\alphar)$ at $\alpha_{K_r}^*$ is
\begin{align*}
\left.\frac{\partial\calLr_N(\alphar)}{\partial\alphar}\right\vert_{\alpha_{K_r}^*}
=\frac{1}{N}\sum_{i=1}^N\left\{\mathsf{1}_{R_i=r}-\mathsf{1}_{R_i=1_d}O^r(\Lobi;\alpha_{K_r}^*)\right\}\Phir(\Lob_i)\ .
\end{align*}
Thus, by the triangle inequality,
{\small
\begin{align*}
\left\|\left.\frac{\partial\calLr_N(\alphar)}{\partial\alphar}\right\vert_{\alpha_{K_r}^*}\right\|_2
&\le\left\|\frac{1}{N}\sum_{i=1}^N\left\{\mathsf{1}_{R_i=r}-\mathsf{1}_{R_i=1_d}O^r(\Lob_i)\right\}\Phir(\Lob_i)\right\|_2\\
&\quads+\left\|\frac{1}{N}\sum_{i=1}^N\mathsf{1}_{R_i=1_d}\left\{O^r(\Lob_i)-O^r(\Lobi;\alpha_{K_r}^*)\right\}\Phir(\Lob_i)\right\|_2\ .
\end{align*} 
}
Consider the first term on the right hand side. Let $A_i=\{\mathsf{1}_{R_i=r}-\mathsf{1}_{R_i=1_d}O^r(\Lob_i)\}\Phir(\Lob_i)$. It's easy to see that $\{A_i\}_{i=1}^N$ are i.i.d. and $\E(A_i)=0$. Thus, 
\begin{align*}
&\quads\E\left\|\frac{1}{N}\sum_{i=1}^NA_i\right\|_2^2
=\frac{1}{N}\E(A_i\tr A_i)\\
&=\frac{1}{N}\E\left[\sum_{k=1}^{K_r}
\left\{\mathsf{1}_{R=r}+\mathsf{1}_{R=1_d}O^r(\Lob)O^r(\Lob)\right\}
\phir_k(\Lob)\phir_k(\Lob)\right]\\
&\le\frac{C_0^2+1}{N}\E\|\Phir(\Lob)\|_2^2\ .
\end{align*}
By \ref{assump-2E}, $\E\|\sum_{i=1}^NA_i/N\|_2^2=O(K_r/N)$. By the Markov inequality, this implies $\|\sum_{i=1}^NA_i/N\|_2=O_p(\sqrt{K_r/N})$. As for the second term on the right hand side, let $\xi=(\xi_1,\cdots,\xi_N)$ where $\xi_i=\mathsf{1}_{R_i=1_d}\{O^r(\Lob_i)-O^r(\Lobi;\alpha_{K_r}^*)\}$ and $B=(B_1,\cdots,B_N)$ where $B_i=\mathsf{1}_{R_i=1_d}\Phir(\Lob_i)$. Then,
\begin{align*}
\left\|\frac{1}{N}\sum_{i=1}^N\xi_iB_i\right\|_2^2
=\frac{1}{N^2}\xi\tr BB\tr\xi
=\frac{1}{N}\xi\tr\left\{\frac{1}{N}\sum_{i=1}^N\mathsf{1}_{R_i=1_d}\Phir(\Lob_i)\Phir(\Lob_i)\tr\right\}\xi\ .
\end{align*}
Following the similar arguments in the proof of Lemma \ref{quadratic}, it's easy to see that 
\begin{align*}
\lambda_{\max}\left\{\frac{1}{N}\sum_{i=1}^N\mathsf{1}_{R_i=1_d}\Phir(\Lob_i)\Phir(\Lob_i)\tr\right\}
\le\lambda_{\max}^*+o_p(1)\ .
\end{align*}
By \ref{assump-2D}, $|\xi_i|\le C_1K_r^{-\mu_1}$. Thus, $\|\frac{1}{N}\sum_{i=1}^N\xi_iB_i\|_2=O_p(K_r^{-\mu_1})$ and 
\begin{align*}
\left\|\left.\frac{\partial\calLr_N(\alphar)}{\partial\alphar}\right\vert_{\alpha_{K_r}^*}\right\|_2
=O_p\left(\sqrt{\frac{K_r}{N}}+K_r^{-\mu_1}\right)\ .
\end{align*}
\end{proof}

\begin{lemma}\label{S1}
Under Assumptions \ref{assump1}--\ref{assump3}, for any missing pattern $r$, 
$$
\sup_{\theta\in\Theta}|\sqrt{N}S_{\theta,1}^r|=o_p(1).
$$
\end{lemma}

\begin{proof}
Consider the following empirical process.
\begin{align*}
\bbG_N(f_{\theta,1})=\sqrt{N}\left[\frac{1}{N}\sum_{i=1}^Nf_{\theta,1}(L_i,R_i)-\E\left\{f_{\theta,1}(L,R)\right\}\right]
\end{align*}
where $f_{\theta,1}(L,R)=\mathsf{1}_{R=1_d}\{O(\Lob)-O^r(\Lob)\}\{\psi_\theta(L)-u^r_\theta(\Lob)\}$ and $O$ is an arbitrary function, which can be viewed as an estimator of true propensity odds $O^r$. By Theorem \ref{odds}, for any $\gamma>0$, there exists constants $C_\gamma>0$ and $N_\gamma >0$ such that for any $N\ge N_\gamma$,
\begin{align*}
P\left\{\left\|O^r(\ \cdot\ ;\hatalphar)-O^r\right\|_\infty\ge C_\gamma\left(\sqrt{\frac{K_r^2}{N_r}}+K_r^{\frac{1}{2}-\mu_1}\right)\right\}
\le\gamma.
\end{align*}
Let $\delta_1=C_\gamma(\sqrt{K_r^2/N_r}+K_r^{1/2-\mu_1})$ and consider the set of functions 
\begin{align*}
\calF_1=\left\{f_{\theta,1}:\left\|O-O^r\right\|_\infty\le\delta_1,\theta\in\Theta\right\}\ .
\end{align*}
By identifying assumption \eqref{eqn:chen-density}, for any $f_{\theta,1}\in\calF_1$, 
\begin{align*}
\E\left\{f_{\theta,1}(L,R)\right\}
&=\E\left[\E\left\{f_{\theta,1}(L,R)\mid \lob,R\right\}\right]\\
&=\E\left[\mathsf{1}_{R=1_d}\E\left\{f_{\theta,1}(L,R)\mid \lob,R=1_d\right\}\right]
=0\ .
\end{align*}
Define $\hat{f}_{\theta,1}(L,R):=\mathsf{1}_{R=1_d}\{O^r(\Lob;\hatalphar)-O^r(\Lob)\}\{\psi_\theta(L)-u^r_\theta(\Lob)\}$. To simplify notations, vectors $A>B$ means that $A_j>B_j$ for each entry, and vector $A>c$ means that $A_j>c$ for each entry where $c$ is a constant. 

Notice that $\sup_{\theta\in\Theta}|\sqrt{N}S_{\theta,1}^r|=\sup_{\theta\in\Theta}|\bbG_N(\hat{f}_{\theta,1})|$. Thus, 
\begin{align*}
1-\gamma
\le P\left(\hat{f}_{\theta,1}\in\calF_1\right)
\le P\left(\underset{\theta\in\Theta}{\sup}\left|\sqrt{N}S_{\theta,1}^r\right|\le\underset{f_{\theta,1}\in\calF_1}{\sup}|\bbG_N(f_{\theta,1})|\right)\ .
\end{align*}
By Markov's inequality, for any $\xi>0$, we have 
\begin{align*}
P\left(\underset{f_{\theta,1}\in\calF_1}{\sup}|\bbG_N(f_{\theta,1})|
\ge\frac{1}{\xi}\E\underset{f_{\theta,1}\in\calF_1}{\sup}\left|\bbG_N(f_{\theta,1})\right|\right)
\le\xi\ .
\end{align*}
If we can show $\E\sup_{f_{\theta,1}\in\calF_1}|\bbG_N(f_{\theta,1})|=o_p(1)$, then for any $\eta>0$ and fixed $\xi>0$, there exists $N_{\xi,\eta}$ and $\sigma_{\xi,\eta}$ such that for any $N\ge N_{\xi,\eta}$,
\begin{align*}
P\left(\frac{1}{\xi}\E\underset{f_{\theta,1}\in\calF_1}{\sup}\left|\bbG_N(f_{\theta,1})\right|\ge\sigma_{\xi,\eta}\right)\le\eta\ .
\end{align*}
Then, for any $\epsilon>0$, by taking $\gamma=\xi=\eta=\frac{\epsilon}{3}$ and appropriately choosing $C_\gamma$, $N_\gamma$, $N_{\xi,\eta}$ and $\sigma_{\xi,\eta}$, we have the above inequalities and for any $N\ge N_\epsilon=\max\{N_\gamma,N_{\xi,\eta}\}$,
\begin{align*}
P\left(\underset{\theta\in\Theta}{\sup}\left|\sqrt{N}S_{\theta,1}^r\right|\ge\sigma_{\xi,\eta}\right)
\le\gamma+\xi+\eta=\epsilon\ .
\end{align*}
That is, $\sup_{\theta\in\Theta}|\sqrt{N}S_{\theta,1}^r|=o_p(1)$.

To show $\E\sup_{f_{\theta,1}\in\calF_1}|\bbG_N(f_{\theta,1})|=o_p(1)$, we utilize the maximal inequality with bracketing (Corollary 19.35 in \cite{van2000asymptotic}). Define the envelop function $F_1(L):=\sup_{\theta\in\Theta}|\psi_\theta(L)-u^r_\theta(\Lob)|\delta_1$. It's easy to see $|f_{\theta,1}(L,R)|\le F_1(L)$ for any $f_{\theta,1}\in\calF_1$. Besides, due to \ref{assump-3D}, for each $j$-th entry,
\begin{align*}
\|F_{1,j}\|_{P,2}
=\sqrt{\int F_{1,j}(L)^2dP(L)}
=\sqrt{\E\left[\underset{\theta}{\sup}\left\{\psi_{\theta,j}(L)-u_{\theta,j}^r(\Lob)\right\}^2\delta_1^2\right]}
\le C_3\delta_1\ .
\end{align*}
To save notations, $\psi_\theta$ and $u_\theta^r$ are used as their $j$-th entry. We also omit the subscripts ``$j$'' of some sets of functions where the related inequalities should hold for each $j$-th entry.

By the maximal inequality,
\begin{align*}
\E\underset{f_{\theta,1}\in\calF_1}{\sup}\left|\bbG_N(f_{\theta,1})\right|
=O_p\left(J_{[ \ ]}\{C_3\delta_1,\calF_1,L_2(P)\}\right)\ .
\end{align*}
To study the entropy integral of $\calF_1$, we split function $f_{\theta,1}$ into two parts and consider two sets of functions $\calG_1=\{g_1:\|g_1\|_\infty\le\delta_1\}$ where $g_1(L)=O(\Lob)-O^r(\Lob)$ and $\calH_1=\{h_{\theta,1}:\theta\in\Theta\}$ where $h_{\theta,1}(L)=\psi_\theta(L)-u^r_\theta(\Lob)$. Notice that $\|g_1\|_\infty\le\delta_1$, $\|h_{\theta,1}(L)\|_{P,2}\le C_3$ and $\delta_1\le1$ when $N$ is large enough. By Lemma \ref{bracket1},
\begin{align*}
n_{[ \ ]}\left\{4\left(C_3+1\right)\epsilon,\calF_1,L_2(P)\right\}
\le n_{[ \ ]}\{\epsilon,\calG_1,L^\infty\}n_{[ \ ]}\{\epsilon,\calH_1,L_2(P)\}\ .
\end{align*}
Define $\tilde{\calG}_1:=\{g_1:\|g_1\|_\infty\le C\}$ for some constant $C$ and $\calO:=\tilde{\calG}_1+O^r=\{O:\|O-O^r\|_\infty\le C\}$. It is obvious that $\calG=\delta_1/C\tilde{\calG}$. Since $O^r$ is a fixed function,
\begin{align*}
&\quads n_{[ \ ]}\left\{\epsilon,\calG_1,L^\infty\right\}
=n_{[ \ ]}\left\{\epsilon,\delta_1/C\tilde{\calG}_1,L^\infty\right\}\\
&=n_{[ \ ]}\left\{C\epsilon/\delta_1,\tilde{\calG}_1,L^\infty\right\}
=n_{[ \ ]}\left\{C\epsilon/\delta_1,\calO,L^\infty\right\}\ .
\end{align*}
The true propensity score odds $O^r$ is unknown, but its roughness is controlled by \ref{assump-3B}. Thus, we should not consider much more rough functions. In other words, our models for propensity score odds should satisfy a similar smoothness condition. There exists appropriate constant $C_\calO$ such that $\calO\subset\calM^r$. Thus,
\begin{align*}
n_{[ \ ]}\{\epsilon,\calO,L^\infty\}
\le n_{[ \ ]}\{\epsilon,\calM^r,L^\infty\}\ .
\end{align*}
Define a set of functions $\calU^r=\{u^r_\theta:\theta\in\Theta\}$. Notice that $\calH_1\subset\Psi-\calU^r$ By Lemma \ref{bracket2}, Assumption \ref{assump-3E} and Lemma \ref{bracket3},
\begin{align*}
n_{[ \ ]}\{2\epsilon,\calH_1,L_2(P)\}
\le n_{[ \ ]}\{\epsilon,\calH,L_2(P)\}n_{[ \ ]}\{\epsilon,\calU^r,L_2(P)\}
\le n_{[ \ ]}\{\epsilon,\calH,L_2(P)\}^2\ .
\end{align*}
Combine the above inequalities and recall \ref{assump-3B} and \ref{assump-3C},
\begin{align*}
J_{[ \ ]}\{\|F_1\|_{P,1},\calF_1,L_2(P)\}
&\le\int_0^{C_3\delta_1}\sqrt{\log n_{[ \ ]}\left\{\frac{C_\calO\epsilon}{4(C_3+1)\delta_1},\calM^r,L^\infty\right\}}d\epsilon\\
&\quads+\sqrt{2}\int_0^{C_3\delta_1}\sqrt{\log n_{[ \ ]}\left\{\frac{\epsilon}{8(C_3+1)\delta_1},\calH,L_2(P)\right\}}d\epsilon\\
&\le\sqrt{C_\calM}\int_0^{C_3\delta_1}\{4(C_3+1)\delta_1/(C_\calO\epsilon)\}^{\frac{1}{2d_\calM}}d\epsilon\\
&\quads+\sqrt{2C_\calH}\int_0^{C_3\delta_1}\{8(C_3+1)\delta_1/\epsilon\}^{\frac{1}{2d_\calH}}d\epsilon\\
&=\sqrt{C_\calM}\{4(C_3+1)/C_\calO\}^{\frac{1}{2d_\calM}}C_3^{1-\frac{1}{2d_\calM}}\delta_1\\
&\quads+\sqrt{2C_\calH}\{8(C_3+1)\}^{\frac{1}{2d_\calH}}C_3^{1-\frac{1}{2d_\calH}}\delta_1\\
&\to0
\end{align*}
since $d_\calM,d_\calH>1/2$ and $\delta_1\to0$ as $N\to\infty$. Therefore, $\E\sup_{f_{\theta,1}\in\calF_1}|\bbG_N(f_{\theta,1})|=O_p(o_p(1))=o_p(1)$ and $\sup_{\theta\in\Theta}|\sqrt{N}S_{\theta,1}^r|=o_p(1)$.
\end{proof}

\begin{lemma}\label{S2}
Under Assumptions \ref{assump1}--\ref{assump3}, for any missing pattern $r$, $\sup_{\theta\in\Theta}|\sqrt{N}S_{\theta,2}^r|=o_p(1)$.
\end{lemma}

\begin{proof}
Consider the following empirical process.
\begin{align*}
\bbG_N(f_{\theta,2})=\sqrt{N}\left[\frac{1}{N}\sum_{i=1}^Nf_{\theta,2}(L_i,R_i)-\E\left\{f_{\theta,2}(L,R)\right\}\right]
\end{align*}
where $f_{\theta,2}(L,R)=\{\mathsf{1}_{R=1_d}O(\Lob)-\mathsf{1}_{R=r}\}
\{u^r_\theta(\Lob)-U(\Lob)\}$ and $O$ and $U$ are arbitrary functions. By Theorem \ref{odds}, for any $\gamma>0$, there exists constants $C_\gamma>0$ and $N_\gamma >0$ such that for any $N\ge N_\gamma$,
\begin{align*}
P\left\{\left\|O^r(\ \cdot\ ;\hatalphar)-O^r\right\|_{P,2}\ge C_\gamma\left(\sqrt{\frac{K_r}{N_r}}+K_r^{-\mu_1}\right)\right\}
\le\gamma\ .
\end{align*}
Besides, by \ref{assump-3A}, $\sup_{\lob\in\domr}|u^r_\theta(\lob)-\Phir(\lob)\tr\beta_\theta^r|\le C_2K_r^{-\mu_2}$. So, we consider the set of functions
\begin{align*}
\calF_2=\left\{f_{\theta,2}:\left\|O-O^r\right\|_{P,2}\le\delta_1',\|u^r_\theta-U\|_\infty\le\delta_2,\theta\in\Theta\right\}
\end{align*}
where $\delta_1'=C_\gamma(\sqrt{K_r/N_r}+K_r^{-\mu_1})$ and $\delta_2=C_2K_r^{-\mu_2}$. Then, for any $f_{\theta,2}\in\calF_2$,
\begin{align*}
\E\left\{f_{\theta,2}(L,R)\right\}
&=\E\left[\left\{\mathsf{1}_{R=1_d}O^r(\Lob)-\mathsf{1}_{R=r}\right\}\left\{u^r_\theta(\Lob)-U(\Lob)\right\}\right]\\
&\quads+\E\left[\mathsf{1}_{R=1_d}\left\{O(\Lob)-O^r(\Lob)\right\}\left\{u^r_\theta(\Lob)-U(\Lob)\right\}\right]\\
&\le0+\left\|O-O^r\right\|_{P,2}\|u^r_\theta-U\|_{P,2}\\
&\le\delta_1'\delta_2
=C_2C_\gamma\left(\frac{K_r^{\frac{1}{2}-\mu_2}}{\sqrt{N_r}}+K_r^{-\mu_1-\mu_2}\right)=o_p(N^{-\frac{1}{2}})\ .
\end{align*}
The last line holds due to the fact that $\|\cdot\|_{P,2}\le\|\cdot\|_\infty$ and \ref{assump-3A} and \ref{assump-3E}. 
Plug in our estimator and define 
$\hat{f}_{\theta,2}(L,R):=\{\mathsf{1}_{R=1_d}O^r(\Lob;\hatalphar)-\mathsf{1}_{R=r}\}\{u^r_\theta(\Lob)-\Phir(\Lob)\tr\beta_\theta^r\}$. Then, $\sup_{\theta\in\Theta}|\sqrt{N}S_{\theta,2}^r|\le\sup_{\theta\in\Theta}|\bbG_N(\hat{f}_{\theta,2})|+\sqrt{N}\delta_1'\delta_2$ and 
\begin{align*}
P\left(\underset{\theta\in\Theta}{\sup}\left|\sqrt{N}S_{\theta,2}^r\right|>
\underset{f_{\theta,2}\in\calF_2}{\sup}|\bbG_N(f_{\theta,2})|+\sqrt{N}\delta_1'\delta_2\right)
\le P\left(\hat{f}_{\theta,2}\notin\calF_2\right)
\le\gamma\ .
\end{align*}
Similarly, we need to show $\E\sup_{f_{\theta,2}\in\calF_2}|\bbG_N(f_{\theta,2})|=o_p(1)$. Define the envelop function $F_2:=(C_0+1)\delta_2$. It's easy to see that $|f_{\theta,2}(L,R)|\le F_2$ for any $f_{\theta,2}\in\calF_2$ when $N$ is large enough. By the maximal inequality with bracketing, 
\begin{align*}
\E\underset{f_{\theta,2}\in\calF_2}{\sup}\left|\bbG_N(f_{\theta,2})\right|
=O_p\left(J_{[ \ ]}\{\|F_2\|_{P,2},\calF_2,L_2(P)\}\right)\ .
\end{align*}
To study the entropy integral of $\calF_2$, we first compare it with $\calF_2'=\{f_{\theta,2}:\|O-O^r\|_{P,2}\le\delta_1',\|u^r_\theta-U\|_{P,2}\le\delta_2,\theta\in\Theta\}$. It is apparent $\calF_2\subset\calF_2'$. Then, we split function $f_{\theta,2}$ into two parts and consider two sets of functions $\calG_2=\{g_2:\|O-O^r\|_{P,2}\le\delta_1'\}$ where $g_2(L,R)=\mathsf{1}_{R=1_d}O(\Lob)-\mathsf{1}_{R=r}$ and $\calH_2=\{h_{\theta,2}:\|h_{\theta,2}\|_{P,2}\le\delta_2,\theta\in\Theta\}$ where $h_{\theta,2}(L)=u^r_\theta(\Lob)-U(\Lob)$. Notice that $\|g_2\|_{P,2}\le C_0+1$ and $\|h_{\theta,2}\|_{P,2}\le\delta_2\le1$ when $N$ is large enough. By Lemma \ref{bracket1},
\begin{align*}
n_{[ \ ]}\{4(C_0+2)\epsilon,\calF_2,L_2(P)\}
\le n_{[ \ ]}\{\epsilon,\calG_2,L_2(P)\}n_{[ \ ]}\{\epsilon,\calH_2,L_2(P)\}\ .
\end{align*}
Notice that $\calG_2+(\mathsf{1}_{R=r}-\mathsf{1}_{R=1_d}O^r)=\mathsf{1}_{R=1_d}\calG_1$. Since $\mathsf{1}_{R=r}-\mathsf{1}_{R=1_d}O^r$ is a fixed function, and $\|\mathsf{1}_{R=1_d}\|_\infty\le1$, by Lemma \ref{bracket4},
\begin{align*}
n_{[ \ ]}\{\epsilon,\calG_2,L_2(P)\}
=n_{[ \ ]}\{\epsilon,\mathsf{1}_{R=1_d}\calG_1,L_2(P)\}
\le n_{[ \ ]}\{\epsilon,\calG_1,L_2(P)\}\ .
\end{align*}
It is obvious that any $\epsilon$-brackets equipped with $\|\cdot\|_\infty$ norm are also $\epsilon$-brackets in $L_2(P)$. With similar arguments in the proof of Lemma \ref{S1}, we have
\begin{align*}
n_{[ \ ]}\{\epsilon,\calG_1,L_2(P)\}
\le n_{[ \ ]}\{\epsilon,\calG_1,L^\infty\}
\le n_{[ \ ]}\left\{C_\calO\epsilon/\delta_1',\calM^r,L^\infty\right\}\ .
\end{align*}
Define a set of functions $\tilde{\calH}_2=\{h_{\theta,2}:\|h_{\theta,2}\|_{P,2}\le C,\theta\in\Theta\}$. Similarly,
\begin{align*}
n_{[ \ ]}\{\epsilon,\calH_2,L_2(P)\}
= n_{[ \ ]}\left\{\epsilon,\delta_2/C\tilde{\calH}_2,L_2(P)\right\}
=n_{[ \ ]}\left\{C\epsilon/\delta_2,\tilde{\calH}_2,L_2(P)\right\}\ .
\end{align*}
Similarly, we split $\tilde{\calH}_2$ into two parts. Define a set of functions $\hat{\calU}^r=\{U:\exists u^r_\theta\in\calU^r \ s.t. \ \|u^r_\theta-U\|_\infty\le C,\}$ where $\calU^r=\{u^r_\theta:\theta\in\Theta\}$. By Lemma \ref{bracket2},
\begin{align*}
n_{[ \ ]}\{2\epsilon,\tilde{\calH}_2,L_2(P)\}
\le n_{[ \ ]}\{\epsilon,\calU^r,L_2(P)\}n_{[ \ ]}\{\epsilon,\hat{\calU}^r,L_2(P)\}\ .
\end{align*}
Also define a set of functions $\E\calH^r:=\{g^r(\lob):=\E\{f(L)\mid\Lob=\lob,R=r\},f\in\calH\}$. Although the set $\calU^r$ is unknown, we should not consider much more rough functions than those in $\E\calH^r$. Therefore, there exists a constant $C_{\hat{\calU}^r}$ such that $\hat{\calU}^r\subset\E\calH^r$. Thus, by Lemma \ref{bracket3},
\begin{align*}
n_{[ \ ]}\{\epsilon,\hat{\calU}^r,L_2(P)\}
\le n_{[ \ ]}\{\epsilon,\E\calH^r,L_2(P)\}
\le n_{[ \ ]}\{\epsilon,\calH,L_2(P)\}\ .
\end{align*}
By \ref{assump-3B}, \ref{assump-3C}, and the above inequalities,
\begin{align*}
&\quads J_{[ \ ]}\{\|F_2\|_{P,2},\calF_2,L_2(P)\}\\
&\le\int_0^{(C_0+1)\delta_2}\sqrt{\log n_{[ \ ]}\left\{\frac{C_\calO\epsilon}{4(C_0+2)\delta_1'},\calM^r,L_2(P)\right\}}d\epsilon\\
&\quads+\sqrt{2}\int_0^{(C_0+1)\delta_2}\sqrt{\log n_{[ \ ]}\left\{\frac{C_{\hat{\calU}^r}\epsilon}{8(C_0+2)\delta_2},\calH,L_2(P)\right\}}d\epsilon\\
&\le\sqrt{C_\calM}\{4(C_0+2)\delta_1'/C_\calO\}^{\frac{1}{2d_\calM}}\{(C_0+1)\delta_2\}^{1-\frac{1}{2d_\calM}}\\
&\quads+\sqrt{2C_\calH}\{8(C_0+2)/C_{\hat{\calU}^r}\}^{\frac{1}{2d_\calH}}(C_0+1)^{1-\frac{1}{2d_\calH}}\delta_2\\
&\to0
\end{align*}
since $d_\calM,d_\calH>1/2$ and $\delta_1',\delta_2\to0$ as $N\to\infty$. So, $\E\sup_{f_{\theta,2}\in\calF_2}|\bbG_N(f_{\theta,2})|=O_p(o_p(1))=o_p(1)$ and $\sup_{\theta\in\Theta}|\sqrt{N}S_{\theta,2}^r|=o_p(1)$.
\end{proof}

\begin{lemma}\label{S3}
Under Assumptions \ref{assump1}--\ref{assump3}, for any missing pattern $r$, 
$$
\sup_{\theta\in\Theta}|\sqrt{N}S_{\theta,3}^r|=o_p(1).$$
\end{lemma}

\begin{proof}
Notice that $S_{\theta,3}^r$ is related to the balancing error:
\begin{align*}
\underset{\theta\in\Theta}{\sup}\left|\sqrt{N}S_{\theta,3}^r\right|
&=\underset{\theta\in\Theta}{\sup}\left|\frac{1}{N}\sum_{i=1}^N\left\{\mathsf{1}_{R_i=1_d}O^r(\Lobi;\hatalphar)-\mathsf{1}_{R_i=r}\right\}\Phir(\Lob_i)\tr\beta_\theta^r\right|\\
&\le\lambda\left\{\gamma\sqrt{K_r}+2(1-\gamma)\sqrt{\PEN_2({\Phir}\tr\hatalphar)}\right\}\sqrt{\PEN_2({\Phir}\tr\beta_\theta^r)}
\end{align*} 
where $\Phir(\lob)\tr\hatalphar=\log O^r(\lob;{\hatalphar})$ denotes the log transformation of the propensity odds model. Due to the similar reason, the roughness of the approximation functions are bounded. Besides, by \ref{assump-2D}, $\lambda=o(1/\sqrt{K_rN_r})$. Thus, $\sup_{\theta\in\Theta}|\sqrt{N}S_{\theta,3}^r|=o_p(1)$.
\end{proof}

\begin{lemma} \label{uniformbound}
Suppose that Assumptions \ref{assump1}--\ref{assump4} hold. Then,
\begin{align*}
\underset{\theta\in\Theta}{\sup}\left|\hat{\bbP}_N\psi_\theta-\E\{\psi_\theta(L)\}\right|=o_p(1)\ .
\end{align*}
\end{lemma}

\begin{proof}
By Lemmas \ref{S1}, \ref{S2}, and \ref{S3}, we only need to show $\sup_{\theta\in\Theta}|S_{\theta,4}^r|=o_p(1)$ where
\begin{align*}
S_{\theta,4}^r
&=\frac{1}{N}\sum_{i=1}^N\mathsf{1}_{R_i=1_d}O^r(\Lob_i)\left\{\psi_\theta(L_i)-u^r_\theta(\Lob_i)\right\}\\
&\quads+\frac{1}{N}\sum_{i=1}^N\mathsf{1}_{R_i=r}u^r_\theta(\Lob_i)-\E\{\mathsf{1}_{R=r}\psi_\theta(L)\}\ .
\end{align*}
Study the following decomposition. Let $\calF_a=\{f_{\theta,a}:\theta\in\Theta\}$ where $f_{\theta,a}(L,R)=\mathsf{1}_{R=r}u^r_\theta(\Lob)$. It's easy to see that for any $\epsilon>0$,
\begin{align*}
n_{[ \ ]}\{\epsilon,\calF_a,L_2(P)\}
\le n_{[ \ ]}\{\epsilon,\calU^r,L_2(P)\}
\le n_{[ \ ]}\{\epsilon,\calH,L_2(P)\}<\infty\ .
\end{align*}
For any measurable function $f$, $\|f(L)\|_{P,2}^2=\E\{f(L)^2\}\ge\{\E|f(L)|\}^2=\|f\|_{P,1}^2$. Thus, 
\begin{align*}
n_{[ \ ]}\{\epsilon,\calF_a,L_1(P)\}\le n_{[ \ ]}\{\epsilon,\calF_a,L_2(P)\}\ .
\end{align*}
By Theorem 19.4 in \cite{van2000asymptotic}, $\calF_a$ is Glivenko-Cantelli. Thus,
\begin{align*}
\underset{\theta\in\Theta}{\sup}\left|\bbP_Nf_{\theta,a}-Pf_{\theta,a}\right|\xrightarrow{a.s.}0\ .
\end{align*}
Also let $\calF_b=\{f_{\theta,b}:\theta\in\Theta\}$ where $f_{\theta,b}(L,R)=\mathsf{1}_{R=1_d}O^r(\Lob)\psi_\theta(L)$ and $\calF_c=\{f_{\theta,c}:\theta\in\Theta\}$ where $f_{\theta,b}(L,R)=\mathsf{1}_{R=1_d}O^r(\Lob)u^r_\theta(\Lob)$. Similarly,
\begin{align*}
\underset{\theta\in\Theta}{\sup}\left|\bbP_Nf_{\theta,b}-Pf_{\theta,b}\right|\xrightarrow{a.s.}0
\textrm{ and }
\underset{\theta\in\Theta}{\sup}\left|\bbP_Nf_{\theta,c}-Pf_{\theta,c}\right|\xrightarrow{a.s.}0\ .
\end{align*}
Notice that $\E\{f_{\theta,b}(L,R)\}=\E\{f_{\theta,c}(L,R)\}$. Besides, the convergence almost surely implies the convergence in probability. Thus, $\sup_{\theta\in\Theta}|S_{\theta,4}^r|=o_p(1)$. Then,
\begin{align*}
\underset{\theta\in\Theta}{\sup}\left|\hat{\bbP}_N\psi_\theta-\E\{\psi_\theta(L)\}\right|=o_p(1)\ .
\end{align*}
\end{proof}

\begin{lemma} \label{S5}
Under Assumptions \ref{assump1}--\ref{assump4}, we have
\begin{align*}
\sqrt{N}\left|S_{\hat{\theta}_N,5}^r-S_{\theta_0,5}^r\right|=o_p(1)\ .
\end{align*}
\end{lemma}

\begin{proof}
Consider the following empirical process.
\begin{align*}
\bbG_N(f_{\theta,5})=\sqrt{N}\left[\frac{1}{N}\sum_{i=1}^Nf_{\theta,5}(L_i,R_i)-\E\left\{f_{\theta,5}(L,R)\right\}\right]
\end{align*}
where $f_{\theta,5}(L,R)=\{\mathsf{1}_{R=1_d}O^r(\Lob)-\mathsf{1}_{R=r}\}\{\psi_\theta(L)-\psi_{\theta_0}(L)\}$. Pick any decreasing sequence $\{\delta_m\}\to0$. Since $\|\hat{\theta}_N-\theta_0\|_2=o_p(1)$, for any $\gamma>0$ and each $\delta_m$, there exists a constant $N_{\delta_m,\gamma}>0$ such that for any $N\ge N_{\delta_m,\gamma}$,
\begin{align}\label{ineq-theta}
P\left(\|\hat{\theta}_N-\theta_0\|_2\ge\delta_m\right)\le\gamma
\end{align}
Consider the set of functions $\calF_5=\{f_{\theta,5}:\|\theta-\theta_0\|_2\le\delta_m\}$. It is easy to check that $\E\{f_{\theta,5}(L,R)\}=0$. Plug in our estimator and define $\hat{f}_{\theta,5}(L,R):=\{\mathsf{1}_{R=1_d}O^r(\Lob)-\mathsf{1}_{R=r}\}\{\psi_{\hat{\theta}_N}(L)-\psi_{\theta_0}(L)\}$. Notice that $\sqrt{N}(S_{\hat{\theta}_N,5}^r-S_{\theta_0,5}^r)=\bbG_N(\hat{f}_{\theta,5})$. Thus, 
\begin{align*}
P\left(\sqrt{N}\left|S_{\hat{\theta}_N,5}^r-S_{\theta_0,5}^r\right|>\underset{f_{\theta,5}\in\calF_5}{\sup}|\bbG_N(f_{\theta,5})|\right)
\le P\left(\hat{f}_{\theta,5}\notin\calF_5\right)
\le\gamma\ .
\end{align*}
Similarly, we only need to show $\E\sup_{f_{\theta,5}\in\calF_5}|\bbG_N(f_{\theta,5})|=o_p(1)$. Define the envelop function $F_5(L):=(C_0+1)f_{\delta_m}(L)$ where $f_\delta$ is the envelop function in \ref{assump-4B}. So, $|f_{\theta,5}(L,R)|\le F_5(L)$ for any $f_{\theta,5}\in\calF_5$. Besides, $\|F_5\|_{P,2}\le (C_0+1)\|f_{\delta_m}\|_{P,2}$. Due to the maximal inequality,
\begin{align*}
\E\underset{f_{\theta,5}\in\calF_5}{\sup}\left|\bbG_N(f_{\theta,5})\right|
=O_p\left(J_{[ \ ]}\{\|F_5\|_{P,2},\calF_5,L_2(P)\}\right)\ .
\end{align*}
Define a set of functions $\calG_5=\{g_{\theta,5}:\|\theta-\theta_0\|_2\le\delta_m\}$ where $g_{\theta,5}(L)=\psi_\theta(L)-\psi_{\theta_0}(L)$. Since $\|\mathsf{1}_{R=1_d}O^r-\mathsf{1}_{R=r}\|_\infty\le(C_0+1)$, by Lemma \ref{bracket4},
\begin{align*}
n_{[ \ ]}\{(C_0+1)\epsilon,\calF_5,L_2(P)\}
\le n_{[ \ ]}\{\epsilon,\calG_5,L_2(P)\}\ .
\end{align*}
Define a set of functions $\tilde{\calG}_5=\{\psi_\theta:\|\theta-\theta_0\|_2\le\delta_m\}$. Since $\psi_{\theta_0}$ is a fixed function, $n_{[ \ ]}\{\epsilon,\calG_5,L_2(P)\}
=n_{[ \ ]}\{\epsilon,\tilde{\calG}_5,L_2(P)\}$. Since $\delta_m\to0$ as $N\to\infty$, we can take $\delta_m$ small enough such that the set $\{\theta:\|\theta-\theta_0\|_2\le\delta_m\}\subset\Theta$. So, $\tilde{\calG}_5\subset\calH$, and $n_{[ \ ]}\{\epsilon,\tilde{\calG}_5,L_2(P)\}
\le n_{[ \ ]}\{\epsilon,\calH,L_2(P)\}$. Then,
\begin{align*}
J_{[ \ ]}\{\|F_5\|_{P,2},\calF_5,L_2(P)\}
&\le\int_0^{(C_0+1)\|f_{\delta_m}\|_{P,2}}\sqrt{\log n_{[ \ ]}\left\{\frac{\epsilon}{C_0+1},\calH,L_2(P)\right\}}d\epsilon\\
&\le\sqrt{C_\calH}\int_0^{(C_0+1)\|f_{\delta_m}\|_{P,2}}\{(C_0+1)/\epsilon\}^{\frac{1}{2d_\calH}}d\epsilon\\
&\le\sqrt{C_\calH}(C_0+1)\|f_{\delta_m}\|_{P,2}^{1-\frac{1}{2d_\calH}}\\
&\to0
\end{align*}
since $d_\calH>1/2$ and $\|f_{\delta_m}\|_{P,2}\to0$ as $N\to\infty$. Thus, $\E\sup_{f_{\theta,5}\in\calF_5}|\bbG_N(f_{\theta,5})|=o_p(1)$ and $\sqrt{N}|S_{\hat{\theta}_N,5}^r-S_{\theta_0,5}^r|=o_p(1)$.
\end{proof}

\begin{lemma} \label{S6}
Under Assumptions \ref{assump1}--\ref{assump4}, we have
\begin{align*}
\sqrt{N}\left|S_{\hat{\theta}_N,6}^r-S_{\theta_0,6}^r\right|=o_p(1)\ .
\end{align*}
\end{lemma}

\begin{proof}
Consider the following empirical process.
\begin{align*}
\bbG_N(f_{\theta,6})=\sqrt{N}\left[\frac{1}{N}\sum_{i=1}^Nf_{\theta,6}(L_i,R_i)-\E\left\{f_{\theta,6}(L,R)\right\}\right]
\end{align*}
where $f_{\theta,6}(L,R)=\{\mathsf{1}_{R=1_d}O^r(\Lob)-\mathsf{1}_{R=r}\}\{u^r_\theta(\Lob)-u^r_{\theta_0}(\Lob)\}$. Similarly, inequality \eqref{ineq-theta} holds and $\E\{f_{\theta,6}(L,R)\}=0$. Consider the set of functions $\calF_6=\{f_{\theta,6}:\|\theta-\theta_0\|_2\le\delta_m\}$. To show $\sqrt{N}|S_{\hat{\theta}_N,6}^r-S_{\theta_0,6}^r|=o_p(1)$, we need to show $\E\sup_{f_{\theta,6}\in\calF_6}|\bbG_N(f_{\theta,6})|=o_p(1)$. Define the envelop function $F_6(L):=(C_0+1)\E\{f_{\delta_m}(L)\mid\Lob\}$. It's easy to see that $|f_{\theta,6}(L,R)|\le F_6(L)$ for any $f_{\theta,6}\in\calF_6$ and $\|F_6\|_{P,2}\le (C_0+1)\|f_{\delta_m}\|_{P,2}$. Apply the maximal inequality,
\begin{align*}
\E\underset{f_{\theta,6}\in\calF_6}{\sup}\left|\bbG_N(f_{\theta,6})\right|
=O_p\left(J_{[ \ ]}\{\|F_6\|_{P,2},\calF_6,L_2(P)\}\right)\ .
\end{align*}
Define a set of functions $\calG_6=\{g_{\theta,6}:\|\theta-\theta_0\|_2\le\delta\}$ where $g_{\theta,6}(L)=u^r_\theta(\Lob)-u^r_{\theta_0}(\Lob)$. Since $\|\mathsf{1}_{R=1_d}O^r-\mathsf{1}_{R=r}\|_\infty\le(C_0+1)$, by Lemma \ref{bracket4},
\begin{align*}
n_{[ \ ]}\{(C_0+1)\epsilon,\calF_6,L_2(P)\}
\le n_{[ \ ]}\{\epsilon,\calG_6,L_2(P)\}\ .
\end{align*}
Define a set of functions $\tilde{\calG}_6=\{u^r_\theta:\|\theta-\theta_0\|_2\le\delta\}$. Similarly, since $u^r_{\theta_0}$ is a fixed function, $n_{[ \ ]}\{\epsilon,\calG_6,L_2(P)\}=n_{[ \ ]}\{\epsilon,\tilde{\calG}_6,L_2(P)\}$. Take $\delta_m$ small enough such that the set $\{\theta:\|\theta-\theta_0\|_2\le\delta_m\}\subset\Theta$. Then, $\tilde{\calG}_6\subset\calU^r$, and by Lemma \ref{bracket3},
\begin{align*}
n_{[ \ ]}\{\epsilon,\tilde{\calG}_6,L_2(P)\}
\le n_{[ \ ]}\{\epsilon,\calU^r,L_2(P)\}
\le n_{[ \ ]}\{\epsilon,\calH,L_2(P)\}\ .
\end{align*}
Therefore,
\begin{align*}
J_{[ \ ]}\{\|F_6\|_{P,2},\calF_6,L_2(P)\}
&\le\int_0^{(C_0+1)\|f_{\delta_m}\|_{P,2}}\sqrt{\log n_{[ \ ]}\left(\epsilon/(C_0+1),\calH,L_2(P)\right)}d\epsilon\\
&\le\sqrt{C_\calH}(C_0+1)\|f_{\delta_m}\|_{P,2}^{1-\frac{1}{2d_\calH}}\\
&\to0
\end{align*}
since $d_\calH>1/2$ and $\|f_{\delta_m}\|_{P,2}\to0$ as $N\to\infty$. Thus, $\E\sup_{f_{\theta,6}\in\calF_6}|\bbG_N(f_{\theta,6})|=O_p(o_p(1))=o_p(1)$ and $\sqrt{N}|S_{\hat{\theta}_N,6}^r-S_{\theta_0,6}^r|=o_p(1)$.
\end{proof}

\begin{lemma} \label{bracket1}
Consider the set of functions $\calF=\{f:=gh,g\in\calG,h\in\calH\}$. Assume that $\|g\|_\infty\le c_g$ for all $g\in\calG$ and $\|h\|_{P,2}\le c_h$ for all $h\in\calH$.
Then, for any $\epsilon\le\min\{c_g,c_h\}$,
\begin{align*}
n_{[ \ ]}\{4(c_g+c_h)\epsilon,\calF,L_2(P)\}
\le n_{[ \ ]}\{\epsilon,\calG,L^\infty\}n_{[ \ ]}\{\epsilon,\calH,L_2(P)\}\ .
\end{align*}
\end{lemma}

\begin{proof}
Suppose $\{u_i,v_i\}_{i=1}^n$ are the $\epsilon$-brackets that can cover $\calG$ and $\{U_j,V_j\}_{j=1}^m$ are the $\epsilon$-brackets that can cover $\calH$. Define the bracket $[\mathsf{U}_k,\mathsf{V}_k]$ for $k=(i-1)m+j$ where $i=1,\cdots,n,j=1,\cdots,m$:
\begin{align*}
\mathsf{U}_k(x)=\min\{u_i(x)U_j(x),u_i(x)V_j(x),v_i(x)U_j(x),v_i(x)V_j(x)\}\ ,\\
\mathsf{V}_k(x)=\max\{u_i(x)U_j(x),u_i(x)V_j(x),v_i(x)U_j(x),v_i(x)V_j(x)\}\ .
\end{align*}
For any function $f\in\calF$, there exists functions $g\in\calG$ and $h\in\calH$ such that $f=gh$. Besides, we can find two pairs of functions $(u_{i_0},v_{i_0})$ and $(U_{j_0},V_{j_0})$ such that $u_{i_0}(x)\le g(x)\le v_{i_0}(x)$, $U_{j_0}(x)\le h(x)\le V{j_0}(x)$, $\|u_{i_0}-v_{i_0}\|_\infty\le\epsilon$, and $\|U_{j_0}-V_{j_0}\|_{P,2}\le\epsilon$. Then, $\mathsf{U}_{k_0}(x)\le f(x)\le \mathsf{V}_{k_0}(x)$ where $k_0=(i_0-1)m+j_0$. Then, we look at the size of the new brackets. By simple algebra,
\begin{align*}
\|\mathsf{U}_k-\mathsf{V}_k\|_{P,2}
&\le\left\|\left(|u_i|+|v_i|\right)|U_j-V_j|+\left(|U_j|+|V_j|\right)|u_i-v_i|\right\|_{P,2}\\
&\le\|u_i\|_\infty\|U_j-V_j\|_{P,2}+\|v_i\|_\infty\|U_j-V_j\|_{P,2}\\
&+\|u_i-v_i\|_\infty\|U_j\|_{P,2}+\|u_i-v_i\|_\infty\|V_j\|_{P,2}\\
&\le2\epsilon(c_g+\epsilon)+2(c_h+\epsilon)\epsilon=2(c_g+c_h+2\epsilon)\epsilon\ .
\end{align*}
Furthermore, for any $\epsilon\le\min\{c_g,c_h\}$, we have $2(c_g+c_h+2\epsilon)\epsilon\le 4(c_g+c_h)\epsilon$.
Therefore,
\begin{align*}
n_{[ \ ]}\{4(c_g+c_h)\epsilon,\calF,L_2(P)\}
\le n_{[ \ ]}\{\epsilon,\calG,L^\infty\}n_{[ \ ]}\{\epsilon,\calH,L_2(P)\}\ .
\end{align*}
\end{proof}

\begin{lemma} \label{bracket2}
Consider the set of functions $\calF=\calH+\calG=\{f:=g+h,g\in\calG,h\in\calH\}$. Assume that $\|g\|_{P,2}\le c_g$ for all $g\in\calG$ and $\|h\|_{P,2}\le c_h$ for all $h\in\calH$.
Then,
\begin{align*}
n_{[ \ ]} \{2\epsilon,\mathcal{F},L_2(P)\}
\le n_{[ \ ]}\{\epsilon,\calG,L_2(P)\}n_{[ \ ]}\{\epsilon,\calH,L_2(P)\}\ .
\end{align*}
\end{lemma}

\begin{proof}
Suppose $\{u_i,v_i\}_{i=1}^n$ are the $\epsilon$-brackets that can cover $\calG$ and $\{U_j,V_j\}_{j=1}^m$ are the $\epsilon$-brackets that can cover $\calH$. Define the bracket $[\mathsf{U}_k,\mathsf{V}_k]$ for $k=(i-1)m+j$ and $i=1,\cdots,n,j=1,\cdots,m$:
\begin{align*}
\mathsf{U}_k(x)=u_i(x)+U_j(x)\ ,\\
\mathsf{V}_k(x)=v_i(x)+V_j(x)\ .
\end{align*}
For any function $f\in\calF$, there exists functions $g\in\calG$ and $h\in\calH$ such that $f=g+h$. Besides, we can find two pairs of functions $(u_{i_0},v_{i_0})$ and $(U_{j_0},V_{j_0})$ such that $u_{i_0}(x)\le g(x)\le v_{i_0}(x)$, $U_{j_0}(x)\le h(x)\le V{j_0}(x)$, $\|u_{i_0}-v_{i_0}\|_{P,2}\le\epsilon$, and $\|U_{j_0}-V_{j_0}\|_{P,2}\le\epsilon$. Then, $\mathsf{U}_{k_0}(x)\le f(x)\le \mathsf{V}_{k_0}(x)$ where $k_0=(i_0-1)m+j_0$ and
\begin{align*}
\|\mathsf{U}_k-\mathsf{V}_k\|_{P,2} 
\le\|u_i-v_i\|_{P,2}+\|U_j-V_j\|_{P,2}
\le 2\epsilon\ .
\end{align*}
Therefore,
\begin{align*}
n_{[ \ ]}\{2\epsilon,\mathcal{F},L_2(P)\}
\le n_{[ \ ]}\{\epsilon,\calG,L_2(P)\}n_{[ \ ]}\{\epsilon,\calH,L_2(P)\}\ .
\end{align*}
\end{proof}

\begin{lemma} \label{bracket3}
Let $\calH$, $\calU^r$ and $\E\calH^r$ be the sets of functions as we defined before. Then, 
\begin{align*}
n_{[ \ ]}\{\epsilon,\calU^r,L_2(P)\}
&\le n_{[ \ ]}\{\epsilon,\calH,L_2(P)\}\ ,\\
n_{[ \ ]}\{\epsilon,\E\calH^r,L_2(P)\}
&\le n_{[ \ ]}\{\epsilon,\calH,L_2(P)\}\ .
\end{align*}
\end{lemma}

\begin{proof}
Suppose $\{u_i,v_i\}_{i=1}^n$ are the $\epsilon$-brackets that can cover $\calH$. Define $U_i(\lob)=\E\{u_i(L)\mid\Lob=\lob,R=r\}$ and $V_i(\lob)=\E\{v_i(L)\mid\Lob=\lob,R=r\}$. Then, for any $u^r\in\calU^r$, there exists $\psi_\theta\in\calH$ such that $u^r(\lob)=\E\{\psi_\theta(L)\mid\Lob=\lob,R=r\}$ with a pair of functions $(u_0,v_0)$ satisfying $u_0(l)\le\psi_\theta(l)\le v_0(l)$ and $\|u_0(L)-v_0(L)\|_{P,2}\le\epsilon$. Then, $U_0(\lob)\le u^r(\lob)\le V_0(\lob)$ and 
\begin{align*}
\|U_0(\Lob)-V_0(\Lob)\|_{P,2}
&=\E[\E\{u_0(L)-v_0(L)\mid\Lob=\lob,R=r\}^2]\\
&\le\E\E[\{u_0(L)-v_0(L)\}^2\mid\Lob=\lob,R=r]\\
&=\E\{u_0(L)-v_0(L)\}^2
=\|u_0-v_0\|_{P,2}
\le\epsilon\ .
\end{align*}
So, $\{U_i,V_i\}_{i=1}^n$ are the $\epsilon$-brackets that can cover $\calU^r$ and
\begin{align*}
n_{[ \ ]}\{\epsilon,\calU^r,L_2(P)\} 
\le n_{[ \ ]}\{\epsilon,\calH,L_2(P)\}\ .
\end{align*}
For any $g^r\in\E\calH^r$, there exists $f\in\calH$ such that $g^r(\lob)=\E\{f(L)\mid\Lob=\lob,R=r\}$. Similarly,
\begin{align*}
n_{[ \ ]}\{\epsilon,\E\calH^r,L_2(P)\} 
\le n_{[ \ ]}\{\epsilon,\calH,L_2(P)\}\ .
\end{align*}
\end{proof}

\begin{lemma} \label{bracket4}
Let $h$ be a fixed bounded function. Assume $\|h\|_\infty\le c_h$. We consider two function classes $\calF=\{f:f(x):=g(x)h(x),g\in\calG\}$ and $\calG=\{g:\|g\|_{P,2}\le c\}$ for a fixed constant $c$. Then, 
\begin{align*}
n_{[ \ ]}\{ c_h\epsilon,\calF,L_2(P)\} 
\le n_{[ \ ]}\{\epsilon,\calG,L_2(P)\}\ .
\end{align*}
\end{lemma}

\begin{proof}
Suppose $\{u_i,v_i\}_{i=1}^n$ are the $\epsilon$-brackets that can cover $\calG$. That is, for any $g\in\calG$, we can find a pair of functions $(u_0,v_0)$ such that $u_0(x)\le g(x)\le v_0(x)$ and $\|u_0-v_0\|_{P,2}\le\epsilon$. Then, for any $x$, either $u_0(x)h(x)\le g(x)h(x)\le v_0(x)h(x)$ or $u_0(x)h(x)\ge g(x)h(x)\ge v_0(x)h(x)$ holds. Define $U_i(x)=\min\{u_i(x)h(x),v_i(x)h(x)\}$ and $V_i(x)=\max\{u_i(x)h(x),v_i(x)h(x)\}$. For any $f\in\calF$, there exists $g\in\calG$ such that $f=gh$ and a pair of functions $(U_0,V_0)$ such $U_0(x)\le f(x)\le V_0(x)$ and 
\begin{align*}
\|U_i-V_i\|_{P,2}
=\|(u_i-v_i)h\|_{P,2}
\le\|h\|_\infty\|(u_i-v_i)\|_{P,2}
\le c_h\epsilon\ .
\end{align*}
So, $\{U_i,V_i\}_{i=1}^n$ are the $c_h\epsilon$-brackets that can cover $\calF$ and
\begin{align*}
n_{[ \ ]}\{c_h\epsilon,\calF,L_2(P)\} 
\le n_{[ \ ]}\{\epsilon,\calG,L_2(P)\}\ .
\end{align*}
\end{proof}

\end{appendices}

\bibliographystyle{chicago}
\bibliography{refer.bib}

\end{document}